\documentclass[aps,prb,twocolumn, floatfix,groupedaddress,11.9pt,usenames,dvipsnames]{revtex4-1}
\pdfoutput=1
\usepackage{amssymb}
\usepackage{amsmath}
\usepackage{natbib}
\usepackage{braket}
\usepackage{float}
\usepackage[pdfencoding=auto, psdextra]{hyperref}
\hypersetup{linktocpage}
\usepackage{graphicx}
\usepackage{mathrsfs}
\usepackage{float}
\usepackage{amsbsy}
\usepackage{dcolumn}
\usepackage{bm}
\usepackage{dsfont}
\usepackage{color}
\usepackage{enumitem}

\newcommand{\Z}{\mathbb{Z}}
\newcommand{\mcH}{\mathcal{H}}

\newcommand{\gammabar}{\bar \gamma}
\newcommand{\hatotimes}{\hat{\otimes}}
\newcommand{\Mt}{\mathsf{M}}
\newcommand{\Dt}{\mathsf{D}}
\newcommand{\Ft}{\mathsf{F}}

\newcommand{\Bt}{\mathsf{B}}
\newcommand{\M}[1]{\mathsf{#1}}
\newcommand{\ct}{\mathcal{C}}
\newcommand{\bdot}{\boldsymbol{{\cdot}}}
\newcommand{\hatsigma}{\hat{\sigma}}

\usepackage{xifthen}
\usepackage{xstring}
 \usepackage{amsthm}

\newtheorem{myprop}{Proposition}

\usepackage{tikz}
\usetikzlibrary{calc}
\usetikzlibrary{backgrounds}
\usetikzlibrary{fit}
\usepackage{ifthen}
 \allowdisplaybreaks
 \usetikzlibrary{arrows,3d}
\usetikzlibrary{positioning}
\usetikzlibrary{decorations.markings}
\usetikzlibrary{patterns}
\usetikzlibrary{decorations.pathreplacing,decorations.markings}
\usepackage{stmaryrd}
\usetikzlibrary{trees,shapes} 
 \tikzset{ferc/.style={circle,draw=red,thick,scale=0.6},
    berc/.style={circle,draw=black,thick,scale=0.6}}
    
\tikzset{
    ->-/.style={postaction={decorate},decoration={%
    markings,
    mark=at position #1 with {\arrow[scale=1.5]{>}} 
        }%
    },
    -<-/.style={postaction={decorate},decoration={%
    markings,
    mark=at position #1 with {\arrow[scale=1.5]{<}} 
        }%
    },
    Sbase/.style={baseline={(current bounding box.center)}}
    }

\tikzset{OES/.style={postaction={
    decorate,
    decoration={
      show path construction,
      moveto code={},
      lineto code={
        \path [#1]
        (\tikzinputsegmentfirst) -- (\tikzinputsegmentlast);
      },
      curveto code={
        \path [#1] (\tikzinputsegmentfirst)
        .. controls
        (\tikzinputsegmentsupporta) and (\tikzinputsegmentsupportb)
        ..
        (\tikzinputsegmentlast);
      },
      closepath code={
        \path [#1]
        (\tikzinputsegmentfirst) -- (\tikzinputsegmentlast);
      },
    },
  }
  },
    AR1/.style={postaction={decorate,decoration={
        markings,
        mark=at position #1 with {\arrow[scale=1.5]{>}}
      }}},
  AR2/.style={postaction={decorate,decoration={
        markings,
        mark=at position #1 with {\arrow[scale=1.5]{<}}
      }}}
  }
  \tikzset{ OP/.style={ellipse,inner sep=0.3mm,fill=white,draw=black,text=black,every node/.style={scale=0.7}},
SS/.style={scale=0.8, every node/.style={scale=0.8}}, 
SSa/.style={scale=1.3, every node/.style={scale=0.9}},
Sc/.style={Sbase,scale=#1, every node/.style={scale=0.8}}}
\tikzset{Dtriangle/.style={regular polygon, regular polygon sides=3,draw=black,fill=gray!30,shape border rotate=180,inner sep=0.2mm,scale=0.8,text=black,minimum size=6mm}}
\tikzset{Tsq/.style={regular polygon, regular polygon sides=4,draw=black,fill=gray!30,shape border rotate=0,inner sep=0.2mm,scale=0.8,text=black,minimum size=6mm}}
\tikzset{Mrec/.style={rectangle,draw=black,fill=gray!30,shape border rotate=0,inner sep=0.2mm,scale=0.8,text=black,minimum size=6mm}}
\tikzset{Sqthirty/.style={regular polygon, regular polygon sides=4,draw=black,fill=gray!30,shape border rotate=30,inner sep=0.2mm,scale=0.8,text=black,minimum size=6mm}}
\tikzset{Sq/.style={regular polygon, regular polygon sides=4,draw=black,fill=gray!30,shape border rotate=#1,inner sep=0.2mm,scale=0.8,text=black,minimum size=6mm}}
\tikzset{Bcir/.style={circle,draw=black,fill=gray!30,inner sep=0.2mm,scale=0.8,text=black,minimum size=3mm}} \tikzset{Bpri/.style={circle,draw=black,fill=black!50,inner sep=0.2mm,scale=0.8,text=black,minimum size=3mm}}
\tikzset{Utriangle/.style={regular polygon, regular polygon sides=3,draw=black,fill=gray!30,inner sep=0.2mm,scale=0.8,text=black,minimum size=6mm}}

\newcommand{\Fer}{
\draw[red,->-=0.75] (1,0)--(0,0)node(a){};
\draw[red,-<-=0.75](1,0)--(2,0)node(c){}; 
\draw[red,-<-=0.75] (1,0)node(b){}--(1,-1)node(d){};
\node[Dtriangle] at (1,0){} ;   \node[above] at (0,0) {$e_1$};\node[above] at (2,0) {$e_0$};
  \node[right] at (1,-1) {$e$};     }

\newcommand{\Bos}{
 \draw[red,->-=0.75] (1,0)--(0,0);     \draw[red,,-<-=0.75] (1,0)--(2,0); 
    \draw[black,->-=0.75] (1,0)--(1,1);
    \node[Bcir] at (1,0) {  } ;
    \node[above] at (0,0) {$v$};\node[above] at (2,0) {$v$};
  \node[right] at (1,1) {$v$};
    }
    
\newcommand{\DS}{\draw[red,OES={AR2=0.5}](0,0)--++(0.5,0)edge[red,-<-=0.5]+(0,-1.25)node[Dtriangle]{}--++(1.5,0)node[Bcir]{}edge[black,->-=0.5]+(0,1.25)--++(1.5,0)edge[red,-<-=0.5]+(0,-1.25)node[Dtriangle]{}--+(0.5,0);\draw[red,loosely dotted] (-0.5,0)--(0,0);
\draw[red,loosely dotted] (4,0)--(4.5,0);}

\newcommand{\DPone}{\draw[red,OES={AR2=0.5}](0,0)--++(0.5,0)edge[black,->-=0.5]+(0,1.25)node[Bcir]{}--++(1.5,0)node[Dtriangle]{}edge[red,-<-=0.5]+(0,-1.25)--++(1.5,0)edge[black,->-=0.5]+(0,1.25)node[Bcir]{}--+(0.5,0); \draw[red,loosely dotted] (-0.5,0)--(0,0);
\draw[red,loosely dotted] (4,0)--(4.5,0);}

\newcommand{\DSsw}{
\path (0:3)++(60:3) node[above]{$\mathrm{s}$};
\draw[red](30:{sqrt(3)})--++(0,0.25,-0.5)node(a){};
\draw[red](30:{2*sqrt(3)})--++(0,0,-1)node(b){};
\draw[dotted,->-=0.75](0,0)--++(0:3)edge[dotted,->-=0.75]+(60:3);  \draw[dotted,->-=0.75](0,0)--++(60:3)edge[dotted,->-=0.75]+(0:3);
\draw[dotted,-<-=0.75](3,0)--+(120:3);
\draw[red,-<-=0.75] (30:{sqrt(3)})--+(30:{sqrt(3)});
\draw[red,->-=0.75] (30:{sqrt(3)})--+(-90:{sqrt(3)/2});
\draw[red,-<-=0.75] (30:{sqrt(3)})--+(150:{sqrt(3)/2});
\draw[red,->-=0.75] (30:{2*sqrt(3)})--+(-30:{sqrt(3)/2});
\draw[red,-<-=0.75] (30:{2*sqrt(3)})--+(90:{sqrt(3)/2});
\path (30:{sqrt(3)})node[Utriangle]{} --
 (30:{2*sqrt(3)})node[Dtriangle]{};
 \foreach \p in {(3/2,0),(60:3/2),($(3,0)+(60:3/2)$),($(3/2,0)+(60:3)$)}
 {\node[Bpri] at \p {}; \draw[] \p --+(0,0,1); } \node[Bpri] at (30:{3/2*sqrt(3}) {};\draw[] (30:{3/2*sqrt(3})  --+(0,-0.25,0.5);}

\newcommand{\DP}{
\draw[dotted,->-=0.75](0,0)--+(-60:3);
\draw[dotted,->-=0.75](-60:3)--+(60:3);
\draw[dotted,->-=0.75](0,0) --+(0:3);
\draw[red,->-=0.75](-30:{sqrt(3)})--(-60:3/2);
\draw[red,->-=0.75](-30:{sqrt(3)})--(-30:{3/2*sqrt(3)});
\draw[red,-<-=0.75](-30:{sqrt(3)})--(0:3/2);
\draw[red](-30:{sqrt(3)})--+(0,0,-1)node(a){};
\node[Dtriangle] at (-30:{sqrt(3)}){};
\node[Bpri]at (0:3/2){};
\node[Bpri]at (-60:3/2){};
\node[Bpri]at (-30:{3/2*sqrt(3)}){};
\draw[](0:3/2)--+(0,0,1)node(c){};
\draw[](-60:3/2)--+(0,0,1)node(d){};
\draw[](-30:{3/2*sqrt(3)})--+(0,0,1)node(b){};
}

  \tikzset{newf/.pic={
\draw[red,->-=0.85](-30:{sqrt(3)})--node[pos=0.5](c){}(-60:3/2);
\draw[red,->-=0.85](-30:{sqrt(3)})--node[pos=0.5](a){}(-30:{3/2*sqrt(3)});
\draw[red,-<-=0.85](-30:{sqrt(3)})--node[pos=0.5](b){}(0:3/2);
\draw[red,-<-=0.85](-30:{sqrt(3)})--node[pos=0.6](d){}++(0,0,-1)node[right=0.5mm,black]{$f$};
\node[Dtriangle] at (-30:{sqrt(3)}){};
\node[left] at (-60:3/2) {$f_{01}$};\node[right] at (-30:{3/2*sqrt(3)}) {$f_{12}$}; \node[above] at (0:3/2) {$f_{02}$};
  }}

\tikzset{newfbar/.pic={
\draw[red,-<-=0.85](30:{sqrt(3)})--node[pos=0.5](c){}(60:3/2);
\draw[red,-<-=0.85](30:{sqrt(3)})--node[pos=0.5](a){}(30:{3/2*sqrt(3)});
\draw[red,->-=0.85](30:{sqrt(3)})--node[pos=0.5](b){}(0:3/2);
\draw[red,-<-=0.75](30:{sqrt(3)})--node[pos=0.6](d){}+(0,0.3,-0.6)node[black,right=1pt]{$f$};
\node[Utriangle] at (30:{sqrt(3)}){ };
\node[left] at (60:3/2) {$f_{01}$};\node[right] at (30:{3/2*sqrt(3)}) {$f_{12}$}; \node[below] at (0:3/2) {$f_{02}$};
}}

\newcommand{\Bnorth}{ \draw[red,->-=0.25] (0,0)node[above,black]{$e$}--node[pos=0.5](b){}(1,0);
 \draw[red,->-=0.75] (1,0)--node[pos=0.5](c){}(2,0)node[above,black]{$e$};
 \node[Bcir]at (1,0){$ $};
 \draw[->-=0.75](1,0)--node[pos=0.5](a){}+(0,0,1.5)node[left,black]{$e$};}

\newcommand{\Bnorthp}{ \draw[red,->-=0.25] (0,0)node[above,black]{$e$}--node[pos=0.5](b){}(1,0);
 \draw[red,->-=0.75] (1,0)--node[pos=0.5](c){}(2,0)node[above,black]{$e$};
 \node[Bpri]at (1,0){$ $};
 \draw[->-=0.75](1,0)--node[pos=0.5](a){}+(0,0,1.5)node[above,black]{$e$};}

\tikzset{newt/.pic={
\draw[red,->-=0.85](-30:{sqrt(3)})--node[pos=0.5](c){}(-60:3/2);
\draw[red,->-=0.85](-30:{sqrt(3)})--node[pos=0.5](a){}(-30:{3/2*sqrt(3)});
\draw[red,-<-=0.85](-30:{sqrt(3)})--node[pos=0.5](b){}(0:3/2);
\node[Tsq] at (-30:{sqrt(3)}){$\mathsf{T}$};
\draw[red,->-=0.75](-30:{sqrt(3)})--node[pos=0.5](d){}+(0,-0.25,0.5);
\node[left] at (-60:3/2) {$f'_{01}$};\node[right] at (-30:{3/2*sqrt(3)}) {$f'_{12}$}; \node[above] at (0:3/2) {$f'_{02}$}; \path (-30:{sqrt(3)})++(0,0,1)node[below]{$f'$};
  }}
  \tikzset{newtbar/.pic={
\draw[red,-<-=0.85](30:{sqrt(3)})--node[pos=0.5](c){}(60:3/2);
\draw[red,-<-=0.85](30:{sqrt(3)})--node[pos=0.5](a){}(30:{3/2*sqrt(3)});
\draw[red,->-=0.85](30:{sqrt(3)})--node[pos=0.5](b){}(0:3/2);
\node[Tsq] at (30:{sqrt(3)}){ $\mathsf{\bar{T}}$};
\draw[red,->-=0.75](30:{sqrt(3)})--node[pos=0.8](d){}+(0,0,1);
\node[left] at (60:3/2) {$f'_{01}$};\node[right] at (30:{3/2*sqrt(3)}) {$f'_{12}$}; \node[below] at (0:3/2) {$f'_{02}$}; \path (30:{sqrt(3)})++(0,0,1)node[left=0.1pt]{$f'$};
}}
\newcommand{\newft}{\draw[red,->-=0.85](-30:{sqrt(3)})--node[pos=0.5](c){}(-60:3/2);
\draw[red,->-=0.85](-30:{sqrt(3)})--node[pos=0.5](a){}(-30:{3/2*sqrt(3)});
\draw[red,-<-=0.85](-30:{sqrt(3)})--node[pos=0.5](b){}(0:3/2);
\node[Tsq] at (-30:{sqrt(3)}){ };
\draw[red,-<-=0.85](-30:{sqrt(3)})--node[pos=0.5](d){}+(0,0,1);
\node[left] at (-60:3/2) {$f'_{01}$};\node[right] at (-30:{3/2*sqrt(3)}) {$f'_{12}$}; \node[above] at (0:3/2) {$f'_{02}$};
\begin{scope}[shift={(0,0,1)}]
\draw[red,->-=0.85](-30:{sqrt(3)})--node[pos=0.5](c){}(-60:3/2);
\draw[red,->-=0.85](-30:{sqrt(3)})--node[pos=0.5](a){}(-30:{3/2*sqrt(3)});
\draw[red,-<-=0.85](-30:{sqrt(3)})--node[pos=0.5](b){}(0:3/2);
\draw[red,-<-=0.85](-30:{sqrt(3)})--node[pos=0.5](d){}+(0,0,-1);
\node[Dtriangle] at (-30:{sqrt(3)}){};
\node[left] at (-60:3/2) {$f_{01}$};\node[right] at (-30:{3/2*sqrt(3)}) {$f_{12}$}; \node[above] at (0:3/2) {$f_{02}$};
\end{scope}}
\newcommand{\ftorder}{\draw[red,->-=0.85](-30:{sqrt(3)})--node[pos=0.5](c){}(-60:3/2);
\draw[red,->-=0.85](-30:{sqrt(3)})--node[pos=0.5](a){}(-30:{3/2*sqrt(3)});
\draw[red,-<-=0.85](-30:{sqrt(3)})--node[pos=0.5](b){}(0:3/2);
\node[Tsq] at (-30:{sqrt(3)}){ };
\draw[red,-<-=0.85](-30:{sqrt(3)})--node[pos=0.5](d){}+(0,0,1);
\node[left] at (-60:3/2) {(i)};\node[right] at (-30:{3/2*sqrt(3)}) {(iii)}; \node[above] at (0:3/2) {(vi)};
\begin{scope}[shift={(0,0,1)}]
\draw[red,->-=0.85](-30:{sqrt(3)})--node[pos=0.5](c){}(-60:3/2);
\draw[red,->-=0.85](-30:{sqrt(3)})--node[pos=0.5](a){}(-30:{3/2*sqrt(3)});
\draw[red,-<-=0.85](-30:{sqrt(3)})--node[pos=0.5](b){}(0:3/2);
\draw[red,-<-=0.85](-30:{sqrt(3)})--node[pos=0.5](d){}+(0,0,-1);
\node[Dtriangle] at (-30:{sqrt(3)}){};
\node[left] at (-60:3/2) {(ii)};\node[right] at (-30:{3/2*sqrt(3)}) {(iv)}; \node[above] at (0:3/2) {(v)};
\end{scope}}

\newcommand{\newftbar}{\draw[red,-<-=0.85](30:{sqrt(3)})--node[pos=0.5](c){}(60:3/2);
\draw[red,-<-=0.85](30:{sqrt(3)})--node[pos=0.5](a){}(30:{3/2*sqrt(3)});
\draw[red,->-=0.85](30:{sqrt(3)})--node[pos=0.5](b){}(0:3/2);
\node[Tsq] at (30:{sqrt(3)}){  };
\draw[red](30:{sqrt(3)})--node[pos=0.8](d){}+(0,0,1);
\node[left] at (60:3/2) {$f'_{01}$};\node[right] at (30:{3/2*sqrt(3)}) {$f'_{12}$}; \node[below] at (0:3/2) {$f'_{02}$};
\begin{scope}[shift={(0,0,1)}]
\draw[red,-<-=0.85](30:{sqrt(3)})--node[pos=0.5](c){}(60:3/2);
\draw[red,-<-=0.85](30:{sqrt(3)})--node[pos=0.5](a){}(30:{3/2*sqrt(3)});
\draw[red,->-=0.85](30:{sqrt(3)})--node[pos=0.5](b){}(0:3/2);
\node[Utriangle] at (30:{sqrt(3)}){ };
\node[left] at (60:3/2) {$f_{01}$};\node[right] at (30:{3/2*sqrt(3)}) {$f_{12}$}; \node[below] at (0:3/2) {$f_{02}$};
\end{scope}
}
\newcommand{\ftbarorder}{\draw[red,-<-=0.85](30:{sqrt(3)})--node[pos=0.5](c){}(60:3/2);
\draw[red,-<-=0.85](30:{sqrt(3)})--node[pos=0.5](a){}(30:{3/2*sqrt(3)});
\draw[red,->-=0.85](30:{sqrt(3)})--node[pos=0.5](b){}(0:3/2);
\node[Tsq] at (30:{sqrt(3)}){  };
\draw[red](30:{sqrt(3)})--node[pos=0.8](d){}+(0,0,1);
\node[left] at (60:3/2) {(vi)};\node[right] at (30:{3/2*sqrt(3)}) {(iv)}; \node[below] at (0:3/2) {(i)};
\begin{scope}[shift={(0,0,1)}]
\draw[red,-<-=0.85](30:{sqrt(3)})--node[pos=0.5](c){}(60:3/2);
\draw[red,-<-=0.85](30:{sqrt(3)})--node[pos=0.5](a){}(30:{3/2*sqrt(3)});
\draw[red,->-=0.85](30:{sqrt(3)})--node[pos=0.5](b){}(0:3/2);
\node[Utriangle] at (30:{sqrt(3)}){ };
\node[left] at (60:3/2) {(v)};\node[right] at (30:{3/2*sqrt(3)}) {(iii)}; \node[below] at (0:3/2) {(ii)};
\end{scope}
}

 \newcommand{\post}{
     \coordinate (C) at (-30:{sqrt(3)}) ;
    \draw[dotted,->-=0.75] (0,0)node[above]{ }--+(0:3)node[above]{ };
    \draw[dotted,->-=0.75] (0,0)--+(-60:3)node[below]{ };
    \draw[dotted,-<-=0.75] (0:3)--(-60:3);
    \draw[red,->-=0.75] (C)--(-60:3/2);
     \draw[red,->-=0.75] (C)--+(-30:{sqrt(3)/2});
      \draw[red,-<-=0.75] (C)--(0:3/2);
    \node[Dtriangle] at (C) {$ $};
}
  \newcommand{\negt}{ \coordinate (C) at (30:{sqrt(3)}) ;
    \draw[dotted,->-=0.75] (0,0)node[below]{ }--+(0:3)node[below]{ };
    \draw[dotted,->-=0.75] (0,0)--+(60:3)node[above]{ };
    \draw[dotted,-<-=0.75] (0:3)--(60:3);
    \draw[red,-<-=0.75] (C)--(60:3/2);
     \draw[red,-<-=0.75] (C)--+(30:{sqrt(3)/2});
      \draw[red,->-=0.75] (C)--(0:3/2);
    \node[Utriangle] at (C) {$ $};}
    
\newcommand{\interpos}{\draw[->-=0.75](0,0)--(0:3);
\draw[->-=0.75](0,0)--(-60:3);
\draw[->-=0.75](-60:3)--(0:3);
\foreach \theta in {-20,-40}
{\draw[->-=0.85] (\theta:0.4)--++(\theta:0.3) ;}
\foreach \theta in {-15,-30,-45}
{\draw[->-=0.85] (\theta:0.9)--++(\theta:0.3) ;}
\foreach \theta in {200,220}
{\draw[-<-=0.5] (0:3)++(\theta:0.4)--++(\theta:0.3) ;}
\foreach \theta in {195,210,225}
{\draw[-<-=0.5] (0:3)++(\theta:0.9)--++(\theta:0.3) ;}
\draw[OES={AR1=1}](-45:1.4)--++(-40:0.3)++(-40:0.2)++(0:0.3)++(40:0.2)--++(40:0.3);
\draw[->-=0.9](3/2-0.15,0)++(0,-0.3)--++(0.3,0);
\draw[->-=0.9](3/2-0.15,0)++(0,-0.8)--++(0.3,0);
\draw[->-=0.9](3/2-0.15,0)++(0,-1.3)--++(0.3,0);
\draw[OES={AR1=1}](-55:1.8)--++(-50:0.3)++(-50:0.2)--++(0:0.3)++(50:0.2)--++(50:0.3);}

\newcommand{\interneg}{\draw[->-=0.5](0,0)--(0:3);
\draw[->-=0.5](0,0)--(60:3);
\draw[->-=0.5](60:3)--(0:3);
\foreach \theta in {20,40}
{\draw[->-=0.85] (\theta:0.4)--++(\theta:0.3) ;}
\foreach \theta in {15,30,45}
{\draw[->-=0.85] (\theta:0.9)--++(\theta:0.3) ;}
\foreach \theta in {160,140}
{\draw[-<-=0.5] (0:3)++(\theta:0.4)--++(\theta:0.3) ;}
\foreach \theta in {165,150,135}
{\draw[-<-=0.5] (0:3)++(\theta:0.9)--++(\theta:0.3) ;}
\draw[OES={AR1=1}](45:1.4)--++(40:0.3)++(40:0.2)++(0:0.3)++(-40:0.2)--++(-40:0.3);
\draw[->-=0.9](3/2-0.15,0)++(0,0.3)--++(0.3,0);
\draw[->-=0.9](3/2-0.15,0)++(0,0.8)--++(0.3,0);
\draw[->-=0.9](3/2-0.15,0)++(0,1.3)--++(0.3,0);
\draw[OES={AR1=1}](55:1.8)--++(50:0.3)++(50:0.2)--++(0:0.3)++(-50:0.2)--++(-50:0.3);}

    \newcommand{\Fposf}[4]{\path ($1/3*#1+1/3*#2+1/3*#3$)node[Dtriangle]{}edge[->-=0.75,red]($1/2*#1+1/2*#3$)edge[->-=0.75,red]($1/2*#1+1/2*#2$) edge[->-=0.75,red]($1/2*#2+1/2*#3$)edge[-<-=0.75,red]++#4; }
     \newcommand{\Fnegf}[4]{\path ($1/3*#1+1/3*#2+1/3*#3$)node[Utriangle]{}edge[->-=0.75,red]($1/2*#1+1/2*#3$)edge[-<-=0.75,red]($1/2*#1+1/2*#2$) edge[-<-=0.75,red]($1/2*#2+1/2*#3$)edge[-<-=0.75,red]++#4; } 
  
 \newcommand{\Btwopf}[3]{\node[Bpri] at  ($1/2*#1+1/2*#2$){}; \draw[->-=0.75] ($1/2*#1+1/2*#2$)--++#3; }
 
    \newcommand{\Tposf}[4]{ \path ($1/3*#1+1/3*#2+1/3*#3$)node[Tsq]{}edge[-<-=0.75,red]($1/2*#1+1/2*#3$)edge[->-=0.75,red]($1/2*#1+1/2*#2$) edge[->-=0.75,red]($1/2*#2+1/2*#3$);
 \draw[->-=0.75,red] ($1/3*#1+1/3*#2+1/3*#3$)--++#4; }
     \newcommand{\Tnegf}[4]{\path ($1/3*#1+1/3*#2+1/3*#3$)node[Tsq]{}edge[->-=0.75,red]($1/2*#1+1/2*#3$)edge[-<-=0.75,red]($1/2*#1+1/2*#2$) edge[-<-=0.75,red]($1/2*#2+1/2*#3$);
     \draw[->-=0.75,red] ($1/3*#1+1/3*#2+1/3*#3$)--++#4; } 
     
\newcommand{\randt}{\draw[OES={AR1=0.75},dotted](0,0)node(1){1}--(1,0)node(2){2};\draw[OES={AR1=0.75},dotted](3,1)node(3){3}edge[->-=0.75](2)--(2,1.5)node(4){4}--(4,2.5)node(5){5}--(4,0)node(6){6};
    \draw[OES={AR1=0.75},dotted](2)--(4)--(1,2)node(7){7}edge[-<-=0.75](1)edge[-<-=0.75](5)edge[-<-=0.75](2);
     \draw[OES={AR1=0.75},dotted](3)edge[->-=0.75](5)--(6)edge[-<-=0.75](2);
   \draw[->-=0.75,dotted](0,2.5)node(8){8}--(7);
    \draw[->-=0.75,dotted](8)--(1);
     \draw[->-=0.25,dotted](0,4)node(9){9}--(8);
     
   \draw[OES={AR1=0.75},dotted](9)--(1,4)node(10){10}--(4,4)node(11){11};\draw[OES={AR1=0.75},dotted](8)--(11)--(7);  
    \draw[->-=0.75,dotted](11)--(5);\draw[->-=0.75,dotted](8)--(10);}
    
     \tikzstyle{ndl}=[rectangle, minimum height=.4pt, minimum width=8pt, inner sep=0pt, draw, rotate=#1]
     
     \newcommand{\ttorus}{\foreach \j in {0,1,2}
{\draw[dotted,OES={AR1=0.75}](0,\j)--(1,\j)--(2,\j);
\draw[dotted,OES={AR1=0.75}](\j,0)--(\j,1)--(\j,2); }
\foreach \j in {(0,0),(0,1),(1,0),(1,1)}
{\draw[dotted,OES={AR1=0.75}] \j --++(45:{sqrt(2)}); }}

\newcommand{\DPOprime}{\draw[red](0,0)--++(0.5,0)edge[black!99]+(0,1.25)node[Bcir]{}--++(1.5,0)node[Dtriangle]{}edge[red]+(0,-1.25)--++(1.5,0)edge[black!99]+(0,1.25)node[Bcir]{}--+(0.5,0); \draw[red,dashed] (-0.5,0)--(0,0);
\draw[red,dashed] (4,0)--(4.5,0);}
\newcommand{\DSOprime}{\draw[red](0,0)--++(0.5,0)edge[red!99]+(0,-1.25)node[Dtriangle]{}--++(1.5,0)node[Bcir]{}edge[black!60]+(0,1.25)--++(1.5,0)edge[red!99]+(0,-1.25)node[Dtriangle]{}--+(0.5,0);\draw[red,dashed] (-0.5,0)--(0,0);
\draw[red,dashed] (4,0)--(4.5,0); \path (0.5,0)++(1.5,0)--node[OP]{$O_\Dt$}++(1.5,0); }

\begin{document}

\author{Sujeet K. \surname{Shukla}}
\affiliation{Department of Physics, University of Washington Seattle, USA 98105}
\author{Tyler D. \surname{Ellison}}
\affiliation{Department of Physics, University of Washington Seattle, USA 98105}
\author{Lukasz \surname{Fidkowski}}
\affiliation{Department of Physics, University of Washington Seattle, USA 98105}

\title{A tensor network approach to 2D bosonization}

\date{\today}

\begin{abstract} 
We present a 2D bosonization duality using the language of tensor networks.  Specifically, we construct a tensor network operator (TNO) that implements an exact 2D bosonization duality. The primary benefit of the TNO is that it allows for bosonization at the level of quantum states. Thus, we use the TNO to provide an explicit algorithm for bosonizing fermionic projected entangled pair states (fPEPs). A key step in the algorithm is to account for a choice of spin-structure, encoded in a set of bonds of the bosonized fPEPS. This enables our tensor network approach to bosonization to be applied to systems on arbitrary triangulations of orientable 2D manifolds.
\end{abstract}


\maketitle

\tableofcontents
\section{Introduction}

The Jordan-Wigner transformation is a well established example of a bosonization duality -- it maps a system of spinless complex fermions to a system of spins [\onlinecite{Jordan28}].  The duality has led to many fruitful applications to one dimensional systems, where it equates 1D fermionic models and spin chains. However, while the Jordan-Wigner transformation is a powerful tool in one dimension, there are challenges to applying it to higher dimensional systems. To implement the Jordan-Wigner transformation in dimensions greater than one, the duality is applied along a 1D path which snakes through the fermionic system. In general, this yields a transformation that maps local fermionic Hamiltonians to \textit{non}-local bosonic Hamiltonians. 

Generalizations of the Jordan Wigner transformation to two dimensions have since overcome this obstacle and indeed map local fermionic Hamiltonians to local bosonic Hamiltonians [\onlinecite{ChenYA17},\onlinecite{Ball05},\onlinecite{Verstraete05},\onlinecite{Kitaev06a}]. Similar to the one dimensional Jordan-Wigner transformation, these two dimensional bosonization dualities are expressed at the level of operators.  That is, they define a mapping of operators, where operators that act on fermionic degrees of freedom are mapped to operators that act on spins. Such a mapping of operators \footnote{This is technically a $C^*$ algebra automorphism.} necessarily comes from conjugating by some unitary operator on the Hilbert space [\onlinecite{Neumann31}].  However, finding the explicit form of this unitary, and thereby obtaining the action of the duality at the level of quantum states, is challenging.

In this work, we formulate a two dimensional bosonization duality at the level of quantum states. Specifically, we identify a tensor network representation of the duality in Ref. [\onlinecite{ChenYA17}]. This is to say, we construct a tensor network operator (TNO)  which, by conjugation, maps operators according to the transformation in Ref. [\onlinecite{ChenYA17}]. Moreover, the TNO may be applied directly to fermionic tensor network states to map them to bosonic states. 
Further, we show that bosonized fermionic projected entangled pair states (fPEPS) may be written explicitly as bosonic projected entangled pair states (bPEPS). 


The TNO inherits two of the main features of the transformation detailed in Ref. [\onlinecite{ChenYA17}]. First, the mapping of operators in Ref. [\onlinecite{ChenYA17}] makes the physical interpretation of two dimensional bosonization transparent -- fermionic excitations are mapped to \textit{emergent} fermions in a $\Z_2$ gauge theory.  Operators that create pairs of fermions are explicitly mapped to operators that create pairs of emergent fermions, which are interpreted as bound states of a bosonic gauge charge and flux.  The gauge constraint on the bosonic side of the duality expressly prohibits unbound charge and flux excitations. Consequently, our TNO clearly maps the subspace of states with an even number of fermions to a constrained Hilbert space with a basis given by configurations of emergent fermions. Second, the bosonization duality of Ref. [\onlinecite{ChenYA17}] carefully accounts for spin-structure -- a mathematical input necessary for bosonization dualities -- while in other treatments, spin-structure is hidden in seemingly arbitrary choices. In our construction, a choice of spin-structure is then specified by a certain set of bonds in the TNO.  Importantly, keeping track of the spin-structure allows us to establish our tensor network bosonization for fermionic systems on arbitrary triangulations of closed, orientable 2D manifolds. 

For context, our approach to bosonization is analogous to a method employed in Ref. [\onlinecite{Haegeman15}] for gauging symmetries at the level of quantum states. In Ref. [\onlinecite{Haegeman15}], a TNO is used to map a state with a global symmetry to a state with the corresponding gauge symmetry. Indeed, symmetries may be gauged by using a duality [\onlinecite{Levin12}], and the TNOs in Ref. [\onlinecite{Haegeman15}] can be understood as a tensor network representation of the duality corresponding to gauging the symmetry. We note that, using the methods of Ref. [\onlinecite{Haegeman15}] to gauge the fermion parity symmetry in a fermionic system, one obtains a TNO that is closely related to our bosonization TNO.  However, unlike the bosonization TNO, the TNO corresponding to gauging fermion parity maps to a system with fermionic degrees of freedom (although, see [\onlinecite{Zohar18}]). The inverse (or Hermitian conjugate) of our bosonization TNO (this maps a bosonic state to a fermionic state) can be understood as ``un-gauging'' fermion parity or ``fermion condensation'' [\onlinecite{Ellison19},\onlinecite{Williamson17a},\onlinecite{Aasen17}].

We emphasize that our bosonization duality is distinct from the efforts to express fermionic tensor networks in terms of bosonic tensor networks. Refs. [\onlinecite{Corboz09},\onlinecite{Corboz10},\onlinecite{Kraus10}] develop strategies for rewriting fermionic tensor network states as bosonic tensor network states. However, these do not change the state -- only its tensor network representation. The bosonization duality, in contrast, maps unentangled fermionic states to long-range entangled bosonic states. Nonetheless, our bosonization duality may prove useful for analyzing fermionic states, since expectation values of local fermionic operators can be recovered by computing the expectation value of the transformed operators in the bosonized tensor network state. Furthermore, our bosonization duality and the subsequent rewriting as an explicit bosonic tensor network state preserves the locality of the tensor network and only increases the bond dimension by a factor of 2.

The remainder of the paper is structured as follows. We begin by introducing the formalism of $\Z_2$-graded Hilbert spaces and $\Z_2$-graded tensor networks. We find the language of $\Z_2$-graded tensor networks especially convenient for expressing our bosonization TNO, and we use the notation established in section \ref{sec:gradedTN} throughout the text.  We encourage readers that are familiar with the formalism of $\Z_2$-grading to briefly skim section \ref{sec:gradedTN} to simply acquaint themselves with our notation.  Before constructing the bosonization TNO, we review the 2D bosonization duality of Ref. [\onlinecite{ChenYA17}], in section \ref{subsec: review of 2D}. Subsequently, in section \ref{TNOrepresentation} we construct the TNO that implements this 2D bosonization duality at the level of states.  After applying the bosonization TNO to a fermionic tensor network state, the resulting state is not explicitly a bosonic tensor network state.  Therefore, section \ref{sec: bosonization of fPEPS} is devoted to describing an algorithmic procedure for ``removing the grading'' and rewriting a bosonized fPEPS as a bPEPS. The procedure involves summing over inequivalent spin-structures, discussed in section \ref{subsec: choosing SS}. Lastly, we note that we describe a tensor network representation of 1D bosonization in Appendices \ref{operatorduality1d} and \ref{subsec:TNbosonization1d}.

\section{\texorpdfstring{ $\Z_2$}{Z2}-graded tensor networks}
\label{sec:gradedTN}
{Our bosonization TNO is naturally expressed in terms of $\Z_2$-graded tensor networks. Therefore, the purpose of this section is to give a concise introduction to $\Z_2$-graded tensor networks and establish the notation used throughout the text. For a similar exposition of $\Z_2$-graded tensor networks, one can consult Refs.~[\onlinecite{Bultinck17},\onlinecite{Bultinck17a}]. We start by defining $\Z_2$-graded Hilbert spaces and $\Z_2$-graded tensors. Then, we introduce the contraction map to ``glue'' together $\Z_2$-graded tensors.  The contraction map allows us to define a linear action of tensors on each other and to form $\Z_2$-graded tensor networks. Accordingly, we describe a representation of a fermionic operator algebra in terms of $\Z_2$-graded tensors and present a diagrammatic representation for $\Z_2$-graded tensor networks.}

\subsection{\texorpdfstring{$\Z_2$}{Z2}-graded Hilbert spaces}
A $\Z_2$-graded Hilbert space is a Hilbert space $\mcH$ with a natural direct sum decomposition: $\mcH=\mcH^0 \bigoplus \mcH^1$.  A vector $|j)\in \mcH$ lying solely in either $\mcH^0$ or $\mcH^1$ has a $\{0,1\}$ valued \textit{grading} denoted as $|j|$, where $|j|=0$ if $|j) \in \mcH^0$ and $|j|=1$ if $|j) \in \mcH^1$. (We use round brackets for vectors in $\Z_2$-graded Hilbert spaces.) In the context of fermionic systems, we consider $\mcH^0$ to be the subspace spanned by states with an even number of fermions and $\mcH^1$ to be the subspace spanned by states with an odd number of fermions. Thus, the grading of a vector and its fermion parity coincide. For this reason, we use grading and parity interchangeably. Further, we refer to vectors with a definite parity as homogeneous vectors, and we call states formed from a superposition of both even and odd parity vectors inhomogeneous. 

To capture the physics of a many-body fermionic system, we will need a generalization of the usual tensor product -- the  graded tensor product $\hatotimes$.  For graded Hilbert spaces $\mcH_a$ and $\mcH_b$, we define the graded tensor product space $\mcH_a \hatotimes \mcH_b$ to be the quotient space:
\begin{align} \label{gradedtensorspace}
    \mcH_a \hatotimes \mcH_b \equiv \frac{(\mcH_a \otimes \mcH_b)\bigoplus (\mcH_b \otimes \mcH_a)}{\sim}.
\end{align}
Here, $\otimes$ is the usual (unsymmetrized) tensor product of Hilbert spaces, and $\sim$ denotes the relation:
\begin{align} \label{gradedtensorrelation}
    |j)_a \otimes |k)_b \sim (-1)^{|j||k|} |k)_b \otimes |j)_a
\end{align}
for $|j)_a \in \mcH_a$ and $|k)_b \in \mcH_b$ both with definite grading. The Hilbert space $\mcH_a \hatotimes \mcH_b$ is itself a graded Hilbert space with the equivalence class $|j)_a \hatotimes |k)_b \in \mcH_a \hatotimes \mcH_b$ having a  grading of $|j| + |k| \text{ mod }2$. As a consequence of Eq.~\eqref{gradedtensorrelation}, we have: 
\begin{align} \label{gradedtensor}
    |j)_a \hatotimes |k)_b = (-1)^{|j||k|} |k)_b \hatotimes |j)_a.
\end{align}
This property of the graded tensor product is key to describing fermions, as it encodes the exchange statistics of the fermions. One can see that the graded tensor product captures the familiar notion of a fermionic Fock space by representing the equivalence class $|j)_a \hatotimes |k)_b$ by the vector ${\frac{1}{2}\Big(|j)_a \otimes |k)_b + (-1)^{|j||k|} |k)_b \otimes |j)_a\Big)}$. When $|j)_a$ and $|k)_b$ are both fermion parity odd, we have an anti-symmetric combination -- the Slater determinant. \par


Before moving on to describe $\Z_2$-graded tensors, we would like to note that Hilbert spaces for bosonic systems also fit into the framework of $\Z_2$-graded Hilbert spaces. A bosonic Hilbert space can be understood as a $\Z_2$-graded Hilbert space for which $\mcH^1$, the space of vectors with odd grading, is empty, leaving $\mcH= \mcH^0$. {The graded tensor product between two bosonic Hilbert spaces reduces to the symmetrized tensor product between the Hilbert spaces, as is standard in tensor networks for bosonic systems.} In a slight abuse of notation, we will denote vectors $|j\rangle$ in bosonic Hilbert spaces with angled brackets. In what follows, we will freely take graded tensor products of states in bosonic Hilbert spaces and states in fermionic Hilbert spaces, and the angled brackets are to remind us that those vectors necessarily have trivial grading.
\subsection{\texorpdfstring{ $\Z_2$}{Z2}-graded tensors}
A rank $N$ $\Z_2$-graded tensor $\M{T}$ is an element of the graded tensor product of $N$ $\Z_2$-graded Hilbert spaces, i.e., $\M{T} \in \mcH_{1} \hat \otimes \ldots \hat \otimes \mcH_N$. Similar to tensors used to study bosonic systems, $\Z_2$-graded tensors admit a convenient graphical representation. Let us consider a specific example with $N=4$ for illustration: 
\begin{align} \label{tensorexample}
\begin{tikzpicture}[Sbase,scale=0.5]
\draw[red,->-=0.5] (0,0)--(0:2)node[above,black]{$p$};     \draw[red,,->-=0.5] (0,0)--(-90:2)node[right,black]{$q$};    
\draw[black,-<-=0.5] (0,0)--(180:2)node[above,black]{$r$};
\draw[red,-<-=0.5] (0,0)--(90:2)node[above,black]{$s$};
\node[Tsq] (v) at (0,0) {$\M{T}$};
\end{tikzpicture}
&\equiv \sum_{a,b,c,d} T_{abcd} |a)_{p}|b)_{q} \langle c|_{r} ( d|_{s} .
\end{align}

On the left hand side of Eq.~\eqref{tensorexample}, we have a diagrammatic representation of the tensor $\M{T}\in \mcH_{p} \hatotimes \mcH_{q} \hatotimes \mcH^*_{r}\hatotimes \mcH^*_{s}$, where $\mcH^*$ is the dual Hilbert space of $\mcH$. In the diagram, the characters at the end of the legs label the Hilbert spaces, and the orientation of the leg indicates whether we consider the Hilbert space to be a dual Hilbert space.  (Legs oriented towards the node correspond to a dual Hilbert space.) Further, we have used red legs for $\Z_2$-graded Hilbert spaces and black legs for bosonic Hilbert spaces.

The right hand side of Eq.~\eqref{tensorexample} is the tensor component form of $\M{T}$ with \textit{component values} $T_{abcd}$.  Note that we have suppressed the $\hatotimes$ between vectors, and as previously mentioned, we use angled brackets for vectors which necessarily have trivial grading [$\langle c|_r$ in Eq.~\eqref{tensorexample}]. Thus, the vector $|a)_{p}|b)_{q} \langle c|_{r} ( d|_{s} $ has a grading of $|a|+|b|+|d| \text{ mod }2$. Since the graded tensor product of Hilbert spaces is a graded Hilbert space, a tensor can be either homogeneous or inhomogeneous. A homogeneous tensor has  nonzero component values only for vectors sharing the same parity, and otherwise, the tensor is inhomogeneous.

It is important to note that the tensor $\M{T}$ is independent of the ordering of vectors in Eq.~\eqref{tensorexample}, but the component values ($T_{abcd}$) can depend on the ordering. For example, if we swap the order of $|a)_p$ and $|b)_q$, we get:
\begin{align}\label{reordering sign}
    \M{T}=\sum_{a,b,c,d} T_{abcd}(-1)^{|a||b|} |b)_{q} |a)_{p}\langle c|_{r} (d|_{s}.
\end{align}
Hence, the tensor components have an additional sign $(-1)^{|a||b|}$ with the new choice of ordering.  The  ordering should therefore be interpreted as a particular choice of orthonormal basis with which to express the tensor. We will often refer to the choice of ordering of the vectors in the component form of a tensor as a choice of \textit{internal ordering}.

\subsection{Contraction map and tensor action} \label{subsec:contraction}
To form tensor networks, we require a map to ``glue'' together tensors.  To this end, we define  the \textit{contraction map}: 
\begin{align} \label{contraction}
    \ct:\mcH^* \hatotimes \mcH &\to \mathbb{C} \nonumber \\ 
  (j| \hatotimes |k) &\mapsto (j|k) = \delta_{jk}.
\end{align}
Notice that a reordering of vectors may be necessary  before evaluating $\ct$.  For example:
\begin{align}\label{supetrace}
    \ct\big[ |k) \hatotimes (j| \big] = \ct\big[ (-1)^{|j||k|} (j| \hatotimes |k) \big]=(-1)^{|j||k|}\delta_{jk}.
\end{align}
Interpreting $ \ct\big[ |k) \hatotimes (j| \big]$ as $\text{tr}[|k) \hatotimes (j|]$, we see that it differs from the usual trace by a sign, $(-1)^{|j||k|}$. This phase is referred to as the \textit{supertrace sign}.

In general, the indices to be contracted need not be next to each other in an algebraic expression. For this reason, we introduce the superscript notation:
\begin{align}
    |k)^\ct \hatotimes (j|^\ct\equiv \ct\big[ |k) \hatotimes (j| \big].
\end{align}
A dual vector and a vector with matching superscripts $\ct$ are to first be reordered then contracted with the map $\ct$.

We now provide examples to illustrate the contraction of $\Z_2$-graded tensors. We consider the following three  even parity tensors to guide the discussion: 
\begin{align}
    \begin{tikzpicture}[Sbase]
    \draw[red,->-=0.5] (-1,0)node[below,black]{$p$} --+(1,0);\draw[red,-<-=0.5] (0,0) --+(1,0)node[below,black]{$q$};   \draw[-<-=0.75](0,0)--++(0,1)node[right]{$r$}; \node[Utriangle] at (0,0) {$\M{A}$}; 
    \end{tikzpicture} \equiv &  \sum_{a,b,c} A_{abc}  (a|_q \langle b |_r (c|_p \nonumber \\
    \begin{tikzpicture}[Sbase]
    \draw[red,-<-=0.5] (-1,0)node[below,black]{$q$} --+(1,0);\draw[red,-<-=0.5] (0,0) --+(1,0)node[below,black]{$s$};   \node[Bcir] at (0,0) {$\M{B}$}; 
    \end{tikzpicture} \equiv &  \sum_{d,e} B_{de} |d)_{q} ( e|_{s} \nonumber \\
     \begin{tikzpicture}[Sbase]
    \draw[red,-<-=0.5] (-1,0)node[below,black]{$s$} --+(1,0);   \node[Tsq] at (0,0) {$\M{C}$}; 
    \end{tikzpicture} \equiv &  \sum_{f} C_{f} |f)_{s}\nonumber.
\end{align}
 
First, we contract the $s$ leg of $\M{B}$ with the $s$ leg of $\M{C}$. The resulting tensor is denoted as $\M{B}\bdot \M{C}$: 
\begin{align} \label{BC example}
\begin{tikzpicture}[Sbase]
 \draw[red,-<-=0.5] (-1,0) --+(1,0);\draw[red,-<-=0.5] (0,0) --+(1,0);   \node[Bcir] at (0,0) {$\M{B}$};   \node[Tsq] at (1,0) {$\M{C}$}; 
\end{tikzpicture} \equiv & \ct_s\left[\M{B}\hatotimes \M{C}\right]\nonumber \\ =& \sum_{d,e} B_{de} |d)_{q} (e|_{s}^{\ct} \sum_{f} C_{f} |f)_{s}^{\ct} \nonumber \\
=& \sum_{d,e}  B_{de} C_{e}|d)_{q}\equiv \M{B}\bdot \M{C},
\end{align}
where $\ct_{s}\left[ \cdots \right]$ refers to contraction of the $s$ index. Notice that $\M{C}$ is a $\Z_2$-graded vector, and $\M{B}$ is a $\Z_2$-graded matrix. We see that $\M{B}$ acts on $\M{C}$ by contraction and gives a new vector, $\M{B}\bdot \M{C}$. Hence  $\M{B}$ can represent linear operators on $\Z_2$-graded vector spaces. 

Second, we contract  $\M{A}$ with $\M{B}\bdot \M{C}$ by contracting the $q$ leg of $\M{A}$ with that of $\M{B}\bdot \M{C}$ to produce a new tensor $\M{A}\bdot \M{B}\bdot \M{C}$: 
\begin{align*} \label{ABC example}
\begin{tikzpicture}[Sbase,scale=0.9]
     \draw[red,->-=0.5] (-1,0) --+(1,0);
     \draw[red,-<-=0.5] (0,0) --+(1,0); \draw[red,-<-=0.5] (1,0) --+(1,0);   \draw[-<-=0.75](0,0)--++(0,1);
     \node[Utriangle] at (0,0) {$\M{A}$}; \node[Bcir] at (1,0) {$\M{B}$}; \node[Tsq] at (2,0) {$\M{C}$}; 
    \end{tikzpicture} & \equiv \ct_q\left[\M{A}\hatotimes \M{B}\bdot \M{C}\right]\nonumber\\*
    &=\sum_{a,b,c} A_{abc}  (a|^\ct_q \langle b |_r (c|_p \sum_{d,e}  B_{de} C_{e}|d)_{q}^{\ct} \nonumber \\
&= \sum_{a,b,c,e}  (-1)^{|a||c|}A_{abc}B_{ae}C_{e} \langle b|_{r}  (c|_p \nonumber\\ & \equiv \M{A}\bdot \M{B}\bdot\M{C}.
\end{align*}
Note that the sign $(-1)^{|a||c|}$ comes from moving $(a|_q$ past $\langle b|_r (c|_p$ in order to perform the contraction.  We can say that $\M{A}$ acted on $\M{B}\bdot \M{C}$ by contraction to produce a tensor $\M{A}\bdot\M{B}\bdot \M{C}$. 

In general, contraction of any two tensors can be interpreted in this way: a tensor $\M{T}$ \textit{acts} on another tensor $\M{S}$ by contraction to produce a tensor $\M{T}\bdot \M{S}$. Letting $ind$ be the set of indices contracted between $\M{T}$ and $\M{S}$, we have:
\begin{align}
\M{T}\bdot \M{S}=\ct_{ind}[\M{T} \hatotimes \M{S}],
\end{align}
where $\ct_{ind}[...]$ refers to contraction over the indices in the set $ind$. Note that, since $\M{T}\bdot \M{S}$ depends on the set $ind$, we should ideally write it as $\M{T}\bdot_{ind} \M{S}$. However, the set $ind$ is typically clear from context, so we omit the subscript for notational convenience.
 
\subsection{$\Z_2$-graded representation of a fermionic operator algebra}
\label{subsec: representing fermions}
Now that we have defined tensors' linear action via contraction, we establish a representation for the fermionic operator algebra of a spinless complex fermion using $\Z_2$-graded tensors.  The representation is essential for the construction of the bosonization TNO, since the bosonization TNO maps fermionic operators represented by $\Z_2$-graded tensors to bosonic operators.

The operator algebra of a spinless complex fermion at a site $p$ is generated by the familiar fermionic creation and annihilation operators: $c^{\dagger}_p$, $c_p$.  However, it is often convenient to instead work with a generating set formed by the two \textit{Majorana operators}:
\begin{align} \label{gammadef}
    \gamma_{p} = c^{\dagger}_{p}+c_{p} \quad\text{and}\quad \bar \gamma_{p} = i(c^{\dagger}_{p}-c_{p}). 
\end{align}
These are Hermitian, unitary operators: 
\begin{align} \label{gamma-prop}
    \gamma_{p}^{\dagger}=\gamma_{p}\text{,} \quad \bar \gamma_{p}^{\dagger}=\bar \gamma_{p}\text{,}\quad \gamma_{p}^2 =\bar \gamma_{p}^2 = 1\text{,}
\end{align} and they satisfy the following commutation relations: 
\begin{align} \label{majoranacommutation}
    \lbrace \gamma_{p}, \bar \gamma_{p'} \rbrace = & 0 \nonumber \\ 
     \lbrace \gamma_{p}, \gamma_{p'} \rbrace =&  \lbrace \bar \gamma_{p}, \bar\gamma_{p'} \rbrace=2\delta_{ p , p'},  
\end{align}
where braces denote the anti-commutator and $\delta_{p,p'}$ is the Kronecker delta function. Furthermore, the \textit{fermion parity operator} $P_p$ is given by:
\begin{align}\label{parity-def}
    P_p=(-1)^{c_p^{\dagger}c_p}=-i \gamma_p\gammabar_p.
\end{align}

We now show that $\gamma_{p}$, $\gammabar_p$ and $P_p$ can be represented as rank-2 $\Z_2$-graded tensors. Letting $|1)$ and $|0)$ represent the fermion occupied and unoccupied states respectively, then the creation and annihilation operators have the following canonical representations: 
\begin{align}
    c_p |1)_p &= |0)_p, \, c_p|0)_p =0 \nonumber\\
    c_p^{\dagger} |1)_p &= 0, \, c_p^{\dagger} |0)_p = |1)_p.
\end{align}
Using Eq.~\eqref{gammadef}, this leads to the following representation of Majorana operators: 
\begin{align} \label{gammagraded}
    \begin{tikzpicture}[Sbase]
        \draw[red,OES={AR2=0.5}] (-1,0)node[below,black]{$p$}--(0,0)--(1,0)node[below,black]{$p$}; \node[OP] {$\gamma$};
    \end{tikzpicture} &\equiv |1)_p(0|_p+|0)_p(1|_p \\ \nonumber
    &=  \sum_{a} |a+1)_p(a|_p  \\
    \begin{tikzpicture}[Sbase]
        \draw[red,OES={AR2=0.5}] (-1,0)node[below,black]{$p$}--(0,0)--(1,0)node[below,black]{$p$}; \node[OP] {$\gammabar$};
    \end{tikzpicture} \label{gammabargraded}
    &\equiv i|1)_p(0|_p-i|0)_p(1|_p  \\ \nonumber 
    &= \sum_{a} (-1)^{a} i |a+1)_p(a|_p. 
\end{align}
Here, and throughout the paper, indices are assumed to take binary values, unless stated otherwise. Thus, $ \sum_a \equiv \sum_{a=0}^1$, and $(a+1) \equiv (a+1) \text{ mod }2$, etc. 
In Appendix \ref{app: Majorana tensors}, we show that the algebraic properties of the Majorana operators are indeed satisfied by the tensor representations in Eqs.~\eqref{gammagraded} and \eqref{gammabargraded}. Furthermore, using Eq.~\eqref{parity-def}, fermion parity $P$ can be represented as:
\begin{align} \label{parity definition}
\begin{tikzpicture}[Sbase]
        \draw[red,OES={AR2=0.5}](-1,0)node[below,black]{$p$}--(0,0)--(1,0)node[below,black]{$p$};\node[OP] {$P$};
    \end{tikzpicture}
    &\equiv  \sum_a (-1)^a|a)_p(a|_p \nonumber\\
    &= |0)_p(0|_p -|1)_p(1|_p.
\end{align}
Eq.~\eqref{parity definition} agrees with the intuition that the $\Z_2$-grading of a vector corresponds to the fermion parity of the state.

\subsection{\texorpdfstring{ $\Z_2$}{Z2}-graded tensor network diagrams}
To establish a general theory of $\Z_2$-graded tensor networks, we need to make sure that tensor diagrams can unambiguously represent the algebraic values. For example, given the tensor network diagram:
\begin{align}
    \begin{tikzpicture}[Sbase,scale=0.9]
     \draw[red,->-=0.5] (-1,0) --+(1,0);
     \draw[red,-<-=0.5] (0,0) --+(1,0); \draw[red,-<-=0.5] (1,0) --+(1,0);   \draw[-<-=0.75](0,0)--++(0,1);
     \node[Utriangle] at (0,0) {$\M{A}$}; \node[Bcir] at (1,0) {$\M{B}$}; \node[Tsq] at (2,0) {$\M{C}$}; 
    \end{tikzpicture},
\end{align}
how do we know whether it represents the tensor  $\M{A}\bdot\M{B}\bdot \M{C}$ or $\M{B}\bdot\M{C}\bdot \M{A}$, or any other order of action of tensors $\M{A}$, $\M{B}$, and  $\M{C}$? Unlike bosonic tensors, $\Z_2$-graded tensors do not commute with each other, and hence, in general, $\M{A}\bdot\M{B}\bdot \M{C}$ and $\M{B}\bdot\M{C}\bdot \M{A}$ are different tensors. 
If $\M{T}$ and $\M{S}$ are homogeneous tensors, then the commutation relation of graded tensor products in Eq.~\eqref{gradedtensor} implies the following commutation relation:
\begin{align}\label{STphase}
   \M{T}\bdot \M{S}=(-1)^{|\M{T}||\M{S}|}\M{S}\bdot \M{T}.
\end{align}
In particular, as long as only one tensor is odd, we  have $\M{T}\bdot \M{S}=\M{S}\bdot \M{T}$, and the order of action of these tensors does not matter. Extending this argument, we see that for a set of homogeneous tensors $\lbrace \M{A},\M{B},\M{C},\ldots \rbrace$, as long as at most one tensor is odd, the order of contraction does not matter.  

What happens when more than one odd tensor appears in a TN? An example of such a tensor network is given in the following diagram, where we assume $\M{A}$ is an even tensor:
\begin{align}\label{gAgorggA}
  \begin{tikzpicture}[Sbase,scale=0.9]
\draw[red,->-=0.5] (-2,0)--+(1,0);
\draw[red,->-=0.5] (-1,0)--+(1,0);
\draw[red,-<-=0.5] (0,0)--+(1,0); 
\draw[red,-<-=0.5] (1,0)--+(1,0);  
\draw[-<-=0.75](0,0)--+(0,1);
\node[Utriangle] at (0,0) {$\M{A}$}; 
\node[OP] at (1,0) {$\gammabar$}; 
\node[OP] at (-1,0) {$\gamma $}; 
    \end{tikzpicture} .
\end{align}
How should this tensor network diagram be read algebraically? For instance, it could represent either $\gamma \bdot \M{A} \bdot\gammabar$ or $ \gammabar \bdot\gamma  \bdot \M{A}$, among other possibilities. This is problematic because, according to Eq.~\eqref{STphase}, $\gamma \bdot \M{A} \bdot \gammabar=-\gammabar \bdot \gamma \bdot \M{A}$. Hence, the algebraic value of the this tensor network diagram is ill defined. 

To remove this ambiguity, we need to indicate the order in which $\gamma$ and $\gammabar$ are applied. We do this by adopting the following simple notation: if two or more odd tensors appear in a diagram, we place numbers next to their nodes to indicate their relative order. For example, $\gamma\bdot A \bdot\gammabar$ and $\gammabar \bdot \gamma\bdot \M{A} $ are then respectively represented by the following diagrams: 
\begin{align}
 \gamma\bdot A \bdot\gammabar \equiv  \begin{tikzpicture}[Sbase,scale=0.9]
\draw[red,->-=0.5] (-2,0)--+(1,0);
\draw[red,->-=0.5] (-1,0)--+(1,0);
\draw[red,-<-=0.5] (0,0)--+(1,0); 
\draw[red,-<-=0.5] (1,0)--+(1,0);  
\draw[-<-=0.75](0,0)--+(0,1);
\node[Utriangle] at (0,0) {$\M{A}$}; 
\node[OP] at (1,0) {$\gammabar$}; \path (1,0)++(0,-0.5)node[red,scale=0.7] {$1$}; 
\node[OP] at (-1,0) {$\gamma $}; \path (-1,0)++(0,-0.5)node[red,scale=0.7] {$2$};
    \end{tikzpicture} \nonumber\\
  \gammabar \bdot \gamma\bdot A \equiv  \begin{tikzpicture}[Sbase,scale=0.9]
\draw[red,->-=0.5] (-2,0)--+(1,0);
\draw[red,->-=0.5] (-1,0)--+(1,0);
\draw[red,-<-=0.5] (0,0)--+(1,0); 
\draw[red,-<-=0.5] (1,0)--+(1,0);  
\draw[-<-=0.75](0,0)--+(0,1);
\node[Utriangle] at (0,0) {$\M{A}$}; 
\node[OP] at (1,0) {$\gammabar$}; \path (1,0)++(0,-0.5)node[red,scale=0.7] {$2$}; 
\node[OP] at (-1,0) {$\gamma $}; \path (-1,0)++(0,-0.5)node[red,scale=0.7] {$1$};
    \end{tikzpicture} .
\end{align}
In fact, the first diagram can also represent any tensor network in which  $\gammabar$ is applied \textit{before} $\gamma$, so it can also represent $ \gamma \bdot \gammabar \bdot A $ or $ A\bdot \gamma \bdot \gammabar$. Similarly, the second diagram can also represent $  \gammabar \bdot A \bdot \gamma $ and $ A\bdot\gammabar\bdot\gamma$ (recall that we assume $A$ is an even tensor).

\section{Tensor network bosonization duality in 2D} \label{2Dbosonization}

In this section, we use the formalism of $\Z_2$-graded tensor networks to construct a TNO that implements the exact 2D bosonization duality of Ref. [\onlinecite{ChenYA17}]. We start by reviewing the operator-level duality, and then show that it can be naturally represented by a TNO, which we refer to as the bosonization TNO.  The TNO representation allows us to easily compute the action of bosonization on quantum states (as opposed to just the action on operators).  In particular, in section \ref{sec: bosonization of fPEPS}, we use the bosonization TNO to map fermionic tensor network states to bosonic tensor network states.  

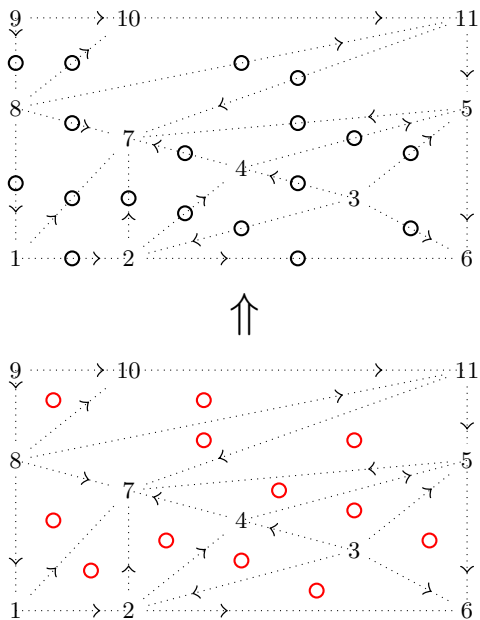
\begin{figure}
    \centering
     \begin{tikzpicture}[Sbase,xscale=1.5,yscale=0.8]
  \randt;
     
\node[berc] at ($1/2*(1)+1/2*(2) $){};
\node[berc] at ($1/2*(1)+1/2*(7) $){};
\node[berc] at ($1/2*(2)+1/2*(7) $){};
\node[berc] at ($1/2*(2)+1/2*(4) $){};
\node[berc] at ($1/2*(2)+1/2*(3) $){};
\node[berc] at ($1/2*(2)+1/2*(6) $){};
\node[berc] at ($1/2*(3)+1/2*(5) $){};
\node[berc] at ($1/2*(3)+1/2*(6) $){};
\node[berc] at ($1/2*(3)+1/2*(4) $){};
\node[berc] at ($1/2*(4)+1/2*(7) $){};
\node[berc] at ($1/2*(4)+1/2*(5) $){};
\node[berc] at ($1/2*(5)+1/2*(7) $){};
\node[berc] at ($1/2*(8)+1/2*(7) $){};
\node[berc] at ($1/2*(11)+1/2*(7) $){};
\node[berc] at ($1/2*(8)+1/2*(11) $){};
\node[berc] at ($1/2*(8)+1/2*(10) $){};
\node[berc] at ($1/2*(8)+1/2*(9) $){};
\node[berc] at ($1/2*(8)+1/2*(1) $){};
 \end{tikzpicture}
 
 \begin{tikzpicture}[Sbase]
  \node[scale=2] at (4,-1.5) {$ \Uparrow$};
 \end{tikzpicture}
 
    \begin{tikzpicture}[Sbase,xscale=1.5,yscale=0.8]
   \randt'
\node[ferc] at ($1/3*(1)+1/3*(2)+1/3*(7) $){};
\node[ferc] at ($1/3*(2)+1/3*(4)+1/3*(7) $){};
\node[ferc] at ($1/3*(2)+1/3*(4)+1/3*(3) $){};
\node[ferc] at ($1/3*(3)+1/3*(4)+1/3*(5) $){};
\node[ferc] at ($1/3*(2)+1/3*(3)+1/3*(6) $){};
\node[ferc] at ($1/3*(3)+1/3*(5)+1/3*(6) $){};
\node[ferc] at ($1/3*(1)+1/3*(7)+1/3*(8) $){};
\node[ferc] at ($1/3*(4)+1/3*(5)+1/3*(7) $){};
\node[ferc] at ($1/3*(7)+1/3*(11)+1/3*(8) $){};
\node[ferc] at ($1/3*(5)+1/3*(7)+1/3*(11) $){};
\node[ferc] at ($1/3*(8)+1/3*(9)+1/3*(10) $){};
\node[ferc] at ($1/3*(8)+1/3*(11)+1/3*(10) $){};
    \end{tikzpicture}
 \caption{ The bosonization duality maps a system of spinless complex fermions to a system of spin-1/2 degrees of freedom. The bottom picture shows the fermionic degrees of freedom (red circles) at each triangular face $f$. The top picture shows the spin-1/2 bosonic degrees of freedom (black circles) on each edge $e$.}
    \label{fig:2Dlattice}
    \label{fig:my_label}
\end{figure}

\subsection{Review of the operator-level bosonization duality}
\label{subsec: review of 2D}
To begin, we describe the lattice on which the duality is defined and set some notation.  The duality in Ref. [\onlinecite{ChenYA17}] can be defined on an arbitrary triangulation of a 2D manifold with boundary [\onlinecite{Ellison19},\onlinecite{Djordje18}].
It is also required that the lattice has a branching structure, i.e. each edge has an orientation (see Fig.~\ref{fig:2Dlattice}) such that the edges around any triangle do not form a cycle. The branching structure yields an ordering of the vertices around a triangle and allows us to define an orientation of each triangle relative to the orientation of the underlying oriented manifold. We denote the ordered vertices of the triangular face $f$ as $f_0$, $f_1$ and $f_2$, where $f_j$ is the $j$-vertex of the triangle, and $j$ refers to the number of edges of the triangle $f$ that point toward $f_j$. We adopt the convention that a triangle is positively oriented if $f_0$, $f_1$ and $f_2$  appear in counter-clockwise order, and otherwise it is negatively oriented. Further, we label the edges of $f$ by $f_{01}$, $f_{12}$, and $f_{02}$, such that $f_{jk}$ is the edge pointing from $f_j$ to $f_k$. We also find it convenient to denote the endpoints of the edge $e$ as $e_0$ and $e_1$, with $e$ pointing from $e_0$ to $e_1$. 

Let us illustrate the notation above using examples from Fig.~\ref{fig:2Dlattice}. (Note that the vertices in Fig.~\ref{fig:2Dlattice} are labeled with integers arbitrarily simply to guide the discussion. They do not denote a global ordering of the vertices.) If $f = \langle 3,2,4 \rangle$ then $f_0 = \langle 3 \rangle$, $f_1 = \langle 2 \rangle$, $f_2 = \langle 4 \rangle $ and $f_{01} = \langle 3,2 \rangle$, $f_{12} = \langle 2,4 \rangle$, $f_{02} =\langle 3,4 \rangle $.  Further, if $e=\langle 2,4 \rangle$, then $e_0=2$, $e_1=4$.  The triangle ${\langle 8,1,7 \rangle}$ is positively oriented while ${\langle 3,2,4\rangle}$ is negatively oriented.

We are now in a position to describe the fermionic degrees of freedom on the lattice and the corresponding operator algebra mapped by the bosonization duality. Each triangle $f$ hosts a spinless complex fermion, and as explained in section \ref{subsec: representing fermions}, its operator algebra is generated by Majorana operators, $\gamma_f$ and $\gammabar_f$. The total fermionic algebra is generated by the set of $\gamma_f, \gammabar_f$ for all triangles $f$. 

The bosonization duality is defined on a subset of the full fermionic operator algebra to ensure that the duality maps local operators to local operators. Specifically, the duality is defined on the subalgebra of even operators $\mathcal{E}$, i.e., the operators that commute with the global fermion parity operator $\prod_{f}P_{f}$, where
\begin{align}
    P_{f}=-i\gamma_f\gammabar_f
\end{align}
is the fermion parity operator at $f$. $\mathcal{E}$ is generated by fermion parity $ P_{f}$ at each triangle $f$, and \textit{hopping operators} $S_e$ at each edge $e$ defined as:
\begin{align} \label{eq:Sdefinition}
    S_{e} = i(-1)^{\eta_e} \gamma_{L_{e}}\gammabar_{R_{e}}.
\end{align}
Here, $L_{e}$ and $R_{e}$ denote the triangle to the left and right of the edge $e$, respectively. For example, in Fig.~\ref{fig:2Dlattice}, we have $L_{\langle 2,4 \rangle }=\langle {2,4,7} \rangle$ and $R_{\langle 2,4 \rangle} = \langle 3,2,4 \rangle$. 
$(-1)^{\eta_e}$ is a sign that comes from a choice of the so-called \textit{spin-structure} $\eta$ [\onlinecite{Ellison19}]. We postpone a detailed discussion of  spin-structure until section \ref{subsec: choosing SS} below. For now, $\eta$ should be understood as a chosen set of edges with $\eta_e$ defined as:
\begin{align}
    \eta_e= \begin{cases}
     1 & \text{if }e \in \eta\\
     0 & \text{otherwise.}
    \end{cases}
\end{align}
As we will explain below, $\eta$ is dependent upon the branching structure, and roughly speaking, ensures that the bosonization duality is uniform across the 2D manifold. 

 
We now discuss the relations satisfied by the generators of the even algebra $\mathcal{E}$. First, all parity operators commute with each other: $P_fP_{f'} = P_{f'}P_f,$ for all $f,f'$. However, not all hopping operators commute with each other. Instead, they satisfy the following commutation relations:
\begin{align}
   S_e S_{e'} = (-1)^{\delta_{L_e,L_{e'}}} (-1)^{\delta_{R_e,R_{e'}}} S_{e'}S_e.
\end{align}
That is, two hopping operators anticommute if and only if they have a common triangle to the left or to the right. For example, in Fig.~\ref{fig:2Dlattice}, $S_{\langle 2,4\rangle}$ and $S_{\langle 3,2\rangle}$ anti-commute because they have a common triangle to the right: $R_{\langle 2,4\rangle} = R_{\langle 3,2\rangle}= \langle324 \rangle$. However, $S_{\langle 2,4\rangle}$ and $S_{\langle 3,4\rangle}$ commute because they do not have a common right or left triangle. Parity operators and hopping operators anti-commute if they share a triangle:
\begin{align}\label{id:PfSeReg}
    S_e P_f = (-1)^{\delta_{e \subset f}} P_f S_e,
\end{align}
otherwise they commute. ($\delta_{e \subset f}=1$ if edge $e$ is part of the triangle, otherwise it is 0.)
 Furthermore, the fermion parity operators and hopping operators are not independent, since for each vertex $v$, they satisfy the relation [\onlinecite{Ellison19}]:
\begin{align} \label{SePfidentity}
    \prod_{e:e_0=v}S_e\prod_{e:e_1=v}S_e \prod_{f:f_0,f_2=v} P_f = 1.
\end{align}
In equation \eqref{SePfidentity}, the first product is over all edges $e$ for which the $e_0$ vertex is $v$, the second product is over all edges $e$ for which  $v$ is the $e_1$ vertex, and the last product is over all triangles for which $v$ is either a $0$-vertex or a $2$-vertex. Note that the sign of $(-1)^{\eta_e}$ in the definition of the hoping operator [Eq.~\eqref{eq:Sdefinition}] is crucial to obtain $1$ on the right hand side of Eq.~\eqref{SePfidentity}. This completes our description of the algebra $\mathcal{E}$ on the fermionic side of the duality\footnote{For each non-contractible cycle of the manifold there is an additional relation between the parity operators and hopping operators. These relations correspond to a certain product of $S_e$ and $P_f$ along the cycle. With an appropriate choice of $\eta$ we are in the $+1$ sector of these relations. See \unexpanded{[\onlinecite{Djordje18}]} for more detail.}, and we move on to describe the bosonic side of the duality. 

On the bosonic side of the duality, as shown in Fig.~\ref{fig:2Dlattice}, we have a spin-1/2 degree of freedom at each edge $e$.  The operator algebra at $e$ is generated by the Pauli operators $X_e$ and $Z_e$, and the full bosonic algebra is generated by the set containing $X_e$ and $Z_e$ for all edges $e$. The bosonization duality maps to just a subalgebra of the full bosonic algebra, where the subalgebra is defined by a certain $\Z_2$ gauge constraint.  The explicit form for the gauge constraint will emerge naturally from the mapping of operators described below. 

The bosonization duality $\mathfrak{D}$, is a homomorphism 
from the algebra of fermion parity even operators $\mathcal{E}$ to a particular bosonic subalgebra.  
$\mathfrak{D}$ is defined by its action on the generators of $\mathcal{E}$, $P_f$ and $S_e$. It maps fermion parity $P_{f}$ to an operator that measures the $\Z_2$ flux at triangle $f$, namely:
\begin{align}\label{def:Wf}
    W_{f} \equiv  Z_{f_{01}}Z_{f_{12}}Z_{f_{02}}. 
\end{align}
Since $S_e$ and $P_f$ anticommute whenever $e$ borders $f$, a natural first guess for the image of $S_e$ under $\mathfrak{D}$ is the operator $X_e$.  $X_e$ creates a pair of $\Z_2$ fluxes and hence anti-commutes with the operator that measures flux on a neighboring triangle. However, mapping $S_e$ to $X_e$ does not preserve the commutation relations with the other hopping operators.  To remedy this, we dress $X_e$ with Pauli $Z$ operators:
\begin{align} \label{def:Ue}
    U_e \equiv X_e \prod_{ f \in \lbrace L_e, R_e \rbrace} Z_{f_{01}}^{\delta_{e,f_{12}}}.
\end{align}
In words, the expression in Eq.~\eqref{def:Ue} says that if $e$ is the $f_{12}$ edge of the triangle to the left, then we include a factor of $Z_{f_{01}}$ on the $f_{01}$ edge of that triangle. Likewise, if $e$ is the $f_{12}$ edge of the triangle to the right, then we include a factor of $Z_{f_{01}}$ on the $f_{01}$ edge of that triangle. For example, looking at Fig.~\ref{fig:2Dlattice}, we have $U_{\langle 3,4\rangle} = X_{\langle 3,4\rangle}$, $U_{\langle 2,4\rangle} = X_{\langle 2,4\rangle} Z_{\langle 3,2\rangle}$, $U_{\langle 5,7 \rangle} = X_{\langle 2,7\rangle} Z_{\langle 4,5\rangle}Z_{\langle 11,5 \rangle}$, etc. 

Lastly, we must check that the relation in Eq.~\eqref{id:PfSeReg} is preserved by the bosonization duality. For each vertex $v$, we find:
\begin{align} \label{2UWrel1}
 \prod_{e:e_0=v}U_e\prod_{e:e_1=v}U_e  \prod_{f:f_0,f_2=v} W_f=G_v,
\end{align}
where $G_v$ is equal to:
\begin{align} \label{Gp}
    G_v=\prod_{e \supset v }X_{e}  \prod_{f: f_0=v} W_{f}. 
\end{align}
{The first product in Eq.~\eqref{Gp} is over all edges $e$ connected to $v$.}
Thus, to preserve the relation \eqref{id:PfSeReg}, we need to impose the gauge constraint $G_v=1$ for all $v$. 

Denoting by $\mathcal{G}$ the bosonic subalgebra generated by the set of $W_{f}$ and $U_e$ with the gauge constraint $G_v=1$ for all $v$, we see that the 2D bosonization duality $\mathfrak{D}$ is a {{bijective}} map from $\mathcal{E}$ to $\mathcal{G}$ defined by:
\begin{align} \label{def:duality2D}
\mathfrak{D}(P_{f}) &= W_{f}, \nonumber\\
\mathfrak{D}(S_e) &= U_e.
\end{align}
The choice of spin-structure $\eta$ ensures that the gauge constraint on the bosonic side of the duality is $G_v=1$ at every vertex $v$. In section \ref{sec: bosonization of fPEPS}, we detail a prescription for choosing a suitable spin structure $\eta$. 
\subsection{TNO representation of the 2D duality }
\label{TNOrepresentation}

Having reviewed both $\Z_2$-graded tensor networks and the operator-level 2D bosonization duality, we can now describe one of our main results -- a realization of 2D bosonization at the level of quantum states.  To accomplish this, we represent the bosonization duality  $\mathfrak{D}$ in Eq.~\eqref{def:duality2D} using a TNO, $\M{D}$. We say that a TNO $\M{D}$ represents the duality $\mathfrak{D}$, if it satisfies: 
\begin{align} \label{dualitytno1}
\begin{tikzpicture}[Sbase]
 \draw[red,loosely dotted] (-0.25,0)--(0,0);
\draw[red,loosely dotted] (3,0)--(3.25,0);
\foreach \j in {0.5,1,...,2.5}
{\draw[red] (\j,0)--+(0,-1); 
}
\foreach \j in {0.25,0.75,...,2.75}
{\draw[black!99] (\j,0)--+(0,1);}
\draw[draw=black,fill=white] (0.95,-0.8) rectangle node{A} ++(1.1,0.4);
\draw[black,fill=white] (0,-0.15) rectangle (3,0.15); \node at (1.5,0) {$ \Dt$};
\end{tikzpicture} = & \begin{tikzpicture}[Sbase]
\node at (1.5,0) {$ \Dt$};\draw[red,loosely dotted] (-0.25,0)--(0,0);
\draw[red,loosely dotted] (3,0)--(3.25,0);
\foreach \j in {0.5,1,...,2.5}
{\draw[red] (\j,0)--+(0,-1); 
}
\foreach \j in {0.25,0.75,...,2.75}
{\draw[black!99] (\j,0) --+(0,1); 
}
\draw[draw=black,fill=white] (0.65,0.35) rectangle node{$\mathfrak{D}(A)$} ++(1.75,0.45);
\draw[black,fill=white] (0,-0.15) rectangle (3,0.15); \node at (1.5,0) {$ \Dt$};
\end{tikzpicture},
\end{align}
for all fermion parity even operators $A \in \mathcal{E}$. Algebraically, this is:
\begin{align} \label{MDdef}
    \M{D} \bdot  A = \mathfrak{D}(A) \bdot\M{D}.
\end{align}
In Eq.~\eqref{MDdef}, we have used the operation $\bdot$ defined in section \ref{subsec:contraction} for the contraction of $\Z_2$-graded tensors. For Eq.~\eqref{MDdef} to hold, it suffices to show that $\M{D}$ satisfies Eq.~\eqref{MDdef} for the generators of $\mathcal{E}$, since for any $A,B,C \in \mathcal{E}$ we have:
\begin{align}
    \M{D} \bdot  (AB+C)= & \M{D} \bdot  AB+ \M{D} \bdot  C\nonumber\\
    =& \mathfrak{D}(A) \bdot \M{D} \bdot B + \mathfrak{D}(C) \bdot \M{D}\nonumber \\
    =& \mathfrak{D}(A)\mathfrak{D}(B) \bdot \M{D} + \mathfrak{D}(C) \bdot \M{D} \nonumber\\
    =& \mathfrak{D}(AB+C) \bdot \M{D}
\end{align}
Hence, we need only find a $\M{D}$ that satisfies:
\begin{subequations}
\begin{align} 
  \M{D} \bdot  P_{f} =&   W_{f} \bdot \M{D} \label{dualityconstraint1} \\
  \M{D}    \bdot S_e=&   U_e \bdot \M{D}, \label{dualityconstraint2}
\end{align}
\end{subequations}
for all triangles $f$ and edges $e$. 
To this end, we propose the TNO ansatz for $\Dt$ shown in Fig.~\ref{fig:2DTNO}.

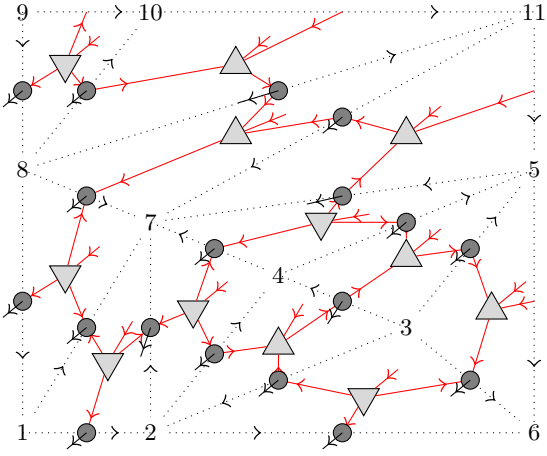
\begin{figure}
    \centering
     \begin{tikzpicture}[Sbase,xscale=1.7,yscale=1.4]
    \randt;
\Fposf{(1)}{(2)}{(7)}{(0,0.2,-0.5)};
\Fnegf{(3)}{(2)}{(4)}{(0,0.2,-0.5)};
\Fnegf{(3)}{(5)}{(6)}{(0,0,-0.7)};
\Fnegf{(3)}{(4)}{(5)}{(0,0,-0.7)};
\Fposf{(2)}{(4)}{(7)}{(0,0,-0.7)};
\Fposf{(3)}{(2)}{(6)}{(0,0,-0.7)};
\Fposf{(4)}{(5)}{(7)}{(0.3,0,-0.2)};
\Fnegf{(8)}{(11)}{(7)}{(0.2,0,-0.5)};
\Fposf{(1)}{(7)}{(8)}{(0,0,-0.7)};
\Fnegf{(8)}{(10)}{(11)}{(0,0,-0.7)};
\Fposf{(9)}{(8)}{(10)}{(0,0,-0.7)};
\Fnegf{(11)}{(5)}{(7)}{(0,0,-0.7)};
\Btwopf{(1)}{(2)}{(0,0,0.4)};\Btwopf{(1)}{(7)}{(0,0,0.4)};\Btwopf{(2)}{(7)}{(0,-0.2,0.2)};
\Btwopf{(2)}{(4)}{(0,0,0.4)};\Btwopf{(3)}{(2)}{(0,0,0.4)};\Btwopf{(6)}{(2)}{(0,0,0.4)};\Btwopf{(3)}{(5)}{(0,0,0.4)};\Btwopf{(3)}{(6)}{(0,0,0.4)};\Btwopf{(3)}{(4)}{(0,-0.1,0.3)};\Btwopf{(4)}{(7)}{(0,0,0.4)};\Btwopf{(4)}{(5)}{(0,0,0.4)};\Btwopf{(5)}{(7)}{(-0.2,0,0.2)};\Btwopf{(7)}{(8)}{(0,0,0.4)};\Btwopf{(7)}{(11)}{(0,0,0.4)};\Btwopf{(8)}{(11)}{(-0.2,0,0.3)};\Btwopf{(8)}{(10)}{(0,0,0.4)};\Btwopf{(8)}{(9)}{(0,0,0.4)};\Btwopf{(1)}{(8)}{(0,0,0.4)};
    \end{tikzpicture}
    \caption{TNO representation of the bosonization duality on a general triangulation of a 2D torus. The TNO is constructed from three types of tensors: $\Ft$ on positive triangles (downward pointing triangular nodes), $\bar \Ft$ on negative triangles (upward pointing triangular nodes), and $\Bt_\eta$ on edges (circular nodes). The TNO is a map from the fermionic legs (red legs, pointing towards the triangular nodes from behind) of $\Ft$ and $\bar \Ft$ tensors to the bosonic legs of $\Bt_\eta$ (black legs, pointing out of the page). }
    \label{fig:2DTNO}
\end{figure}

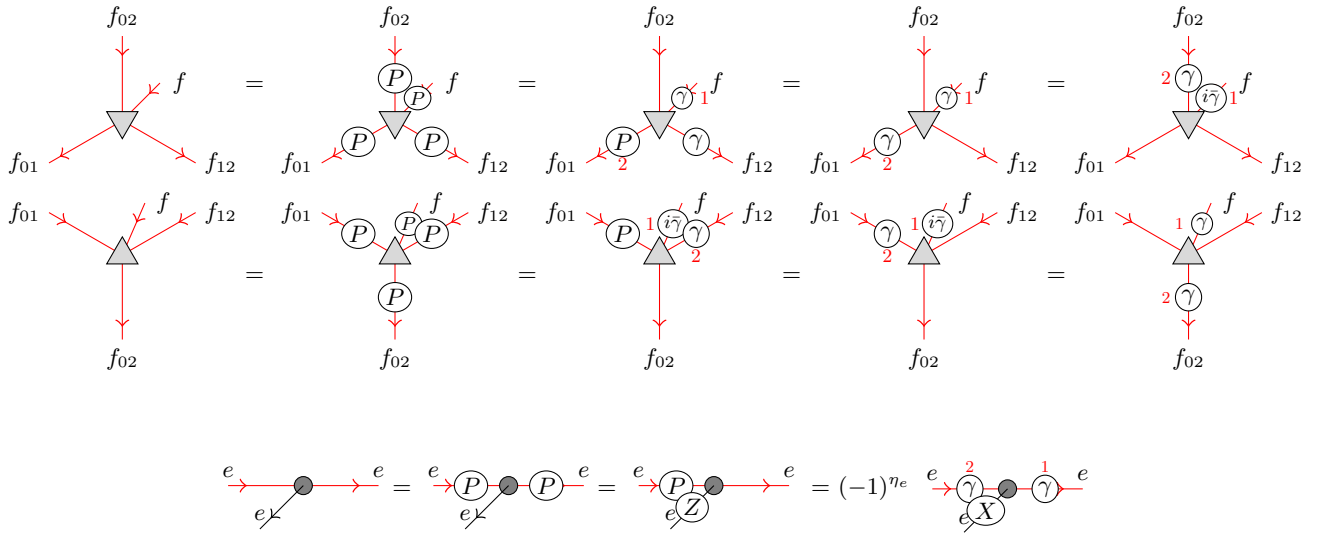
\begin{figure*}
    \centering
    \begin{minipage}{\linewidth}
    \begin{tikzpicture}[Sbase]
     \pic[ scale=1.3] at (0,0) {newf};
     \end{tikzpicture}=
     \begin{tikzpicture}[Sbase]
     \pic[ scale=1.3] at (0,0) {newf}; \node[OP,scale=0.8] at (d) {$P$};
\node[OP] at (a) {$P$};\node[OP] at (b) {$P$};\node[OP] at (c) {$P$};
     \end{tikzpicture}=\begin{tikzpicture}[Sbase]
     \pic[ scale=1.3] at (0,0) {newf};
\node[OP] at (a) {$\gamma $};\node[OP] at (c) {$P$};\node[OP,scale=0.8] at (d) {$\gamma$};
\path (d)++(0.3,0)node[red,scale=0.8]{1};
\path (c)++(0,-0.3)node[red,scale=0.8]{2};
     \end{tikzpicture}=\begin{tikzpicture}[Sbase]
     \pic[ scale=1.3] at (0,0) {newf};
\node[OP] at (c) {$\gamma $};\node[OP,scale=0.8] at (d) {$\gamma$};
\path (d)++(0.3,0)node[red,scale=0.8]{1};
\path (c)++(0,-0.3)node[red,scale=0.8]{2};
     \end{tikzpicture}=\begin{tikzpicture}[Sbase]
     \pic[ scale=1.3] at (0,0) {newf};
\node[OP] at (b) {$\gamma $};\node[OP,scale=0.8] at (d) {$i\gammabar$};
\path (d)++(0.3,0)node[red,scale=0.8]{1};
\path (b)++(-0.3,0)node[red,scale=0.8]{2};
     \end{tikzpicture} 
    \end{minipage}
    \vspace{1cm}
 \begin{minipage}{\linewidth}
 \begin{tikzpicture}[Sbase]
     \pic[ scale=1.3] at (0,0) {newfbar};
     \end{tikzpicture}=
     \begin{tikzpicture}[Sbase]
     \pic[ scale=1.3] at (0,0) {newfbar};\node[OP,scale=0.8] at (d) {$P$};
\node[OP] at (a) {$P$};\node[OP] at (b) {$P$};\node[OP] at (c) {$P$};
     \end{tikzpicture}=\begin{tikzpicture}[Sbase]
     \pic[ scale=1.3] at (0,0) {newfbar};\node[OP,scale=0.8] at (d) {$i\gammabar$};
\node[OP] at (a) {$\gamma $};\node[OP] at (c) {$P$};
\path (d)++(-0.3,0)node[red,scale=0.8]{1};
\path (a)++(0,-0.3)node[red,scale=0.8]{2};
     \end{tikzpicture}=\begin{tikzpicture}[Sbase]
     \pic[ scale=1.3] at (0,0) {newfbar};
\node[OP] at (c) {$\gamma $};\node[OP,scale=0.8] at (d) {$i\gammabar$};
\path (d)++(-0.3,0)node[red,scale=0.8]{1};
\path (c)++(0,-0.3)node[red,scale=0.8]{2};
     \end{tikzpicture}=\begin{tikzpicture}[Sbase]
     \pic[ scale=1.3] at (0,0) {newfbar};
\node[OP] at (b) {$\gamma $};\node[OP,scale=0.8] at (d) {$\gamma$};
\path (d)++(-0.3,0)node[red,scale=0.7]{1};
\path (b)++(-0.3,0)node[red,scale=0.7]{2};
     \end{tikzpicture}
    \end{minipage}
    
\begin{minipage}{0.8\linewidth}
\begin{tikzpicture}[Sbase,scale=1]
  \Bnorthp;
 \end{tikzpicture}=
 \begin{tikzpicture}[Sbase,scale=1]
    \Bnorthp; \node[OP] at (b) {$P $};\node[OP] at (c) {$P $};
 \end{tikzpicture}= \begin{tikzpicture}[Sbase,scale=1]
  \Bnorthp;\node[OP] at (b) {$P $};\node[OP] at (a) {$ Z$};
 \end{tikzpicture} = $(-1)^{\eta_e}$ \begin{tikzpicture}[Sbase,scale=1]
    \Bnorthp;\node[OP] at (b) {$\gamma $};\node[OP] at (c) {$\gamma $};\node[OP] at (a) {$X$};
    \path (b)++(0,0.3)node[red,scale=0.7] {$2$}; 
    \path (c)++(0,0.3)node[red,scale=0.7] {$1$}; 
 \end{tikzpicture}
\end{minipage}
 \caption{Graphical representation of the symmetries in Eqs.~\eqref{F2sym}, \eqref{F2sym2}, and \eqref{B2sym} for tensors $\Ft[f]$ (downward pointing triangular nodes), $\bar\Ft[f]$ (upward pointing triangular nodes), and  $\Bt_\eta[e]$ (circular nodes). }
    \label{fig:2Dsym}
\end{figure*}

The ansatz depicted in Fig.~\ref{fig:2DTNO} is created by contracting together three kinds of tensors: tensors  $\Ft[f]$ on positivly oriented triangles, tensors $\bar \Ft[f]$ on negativly oriented triangles, and tensors $\Bt_\eta[e]$ on edges.
In explicit component form, the tensors $\Ft[f]$ and $\bar\Ft[f]$ are:
\begin{align}
    \Ft[f]\equiv &\sum_{j,a,b,c} F^j_{a,b,c}|c)_{f_{01}}|a)_{f_{12}}(b|_{f_{02}} (j|_f\nonumber\\
    \bar \Ft[f] \equiv &  \sum_{j,a,b,c} \bar{F}^j_{a,b,c}|b)_{f_{02}}(a|_{f_{12}}(c|_{f_{01}}(j|_f,
\end{align}
where all sums are over binary values. Diagrammatically, we represent $\Ft[f]$ and $\bar\Ft[f]$ respectively as:
\begin{align}
     \begin{tikzpicture}[Sbase,scale=1.5]
\pic at (0,0) {newf};
\end{tikzpicture}, \, \begin{tikzpicture}[Sbase,scale=1.5]
\pic at (0,0) {newfbar};
\end{tikzpicture}.
\end{align}
The legs labeled by $f$ are the physical legs and extend into the page.  These legs contract with fermionic operators or a fermionic tensor network state when the TNO is applied. 

The tensor $\Bt_{\eta}[e]$ at each edge is obtained by making a spin-structure dependent modification to a tensor $\Bt[e]$. $\Bt[e]$ has the component form:
\begin{align}
    \Bt[e]=\sum_{j,a,b} B^j_{a,b}|a)_e|j\rangle_e (b|_e,
\end{align}
while the component form of $\Bt_{\eta}[e]$ is: 
 \begin{align}
     \Bt_{\eta}[e] =& Z^{\eta_e}_e\bdot \Bt \nonumber\\
     =& \sum_{j,a,b} (-1)^{j\eta_e}B^j_{a,b}|a)_e|j\rangle_e (b|_e,
 \end{align}

which is pictorially represented as: 
\begin{align}
    \begin{tikzpicture}[Sbase]
    \Bnorthp;
    \end{tikzpicture}  \equiv  \begin{tikzpicture}[Sbase]
    \Bnorth; \node[OP] at (1,0,1.5) {$Z^{\eta_e}$};
    \end{tikzpicture}.
\end{align}
The darker node on the left hand side represents $\Bt_{\eta}$, and the lighter node on the right hand side represents $\Bt$. The physical legs are bosonic Hilbert spaces depicted in black and pointing out of the page.


Now, we view the constraints in \eqref{dualityconstraint1} and  \eqref{dualityconstraint2} as symmetries of the tensor $\M{D}$. These symmetries can be further reduced to symmetries of the local tensors of $\Dt$, which then fixes the values of the local tensors. Indeed, we  now show that $\Dt$ satisfies Eqs.~\eqref{dualityconstraint1} and \eqref{dualityconstraint2}, if the local tensors $\Ft$, $\bar \Ft $ and $\Bt$ satisfy the symmetries depicted in Fig.~\ref{fig:2Dsym}. Algebraically, we write the symmetries for $\Ft[f]$ and $\bar\Ft[f]$ as: 
\begin{align}\label{F2sym}
\Ft &= P_{f_{01}}\bdot P_{f_{12}}\bdot  \Ft \bdot P_{f_{02}}\bdot  P_f = P_{f_{01}} \bdot \gamma_{f_{12}}\bdot \Ft\bdot  \gamma_f \nonumber \\
&=\gamma_{f_{01}}\bdot \Ft\bdot  \gamma_f = \Ft\bdot \gamma_{f_{02}} \bdot  i\gammabar_f  \\ \label{F2sym2}
\bar\Ft &= P_{f_{02}} \bdot \bar\Ft \bdot P_{f_{01}}\bdot P_{f_{12}}\bdot  P_f= \bar\Ft \bdot \gamma_{f_{12}} \bdot P_{f_{01}}\bdot  i\gammabar_f \nonumber \\
&=\bar\Ft\bdot  \gamma_{f_{01}}\bdot i\gammabar_f=\gamma_{f_{02}}\bdot \bar\Ft \bdot \gamma_f
\end{align}

and the symmetries for $\Bt[e]$ can be written as: 
\begin{align}\label{B2sym}
    \Bt_{\eta}=P_e\bdot \Bt_{\eta}\bdot  P_e = P_e\bdot  Z_e\bdot \Bt_{\eta} = (-1)^{\eta_e}\gamma_e\bdot  X_e \bdot \Bt_{\eta} \bdot \gamma_e,
\end{align}
where the contractions in Eqs.~\eqref{F2sym}, \eqref{F2sym2}, and \eqref{B2sym} should be read in conjunction with the diagrams in Fig.~\ref{fig:2Dsym}. We note that the first symmetries of $\Ft$, $\bar\Ft$, and $\Bt$ imply that each of these tensors is fermion parity even.

To see how the symmetries of the local tensors ensure that $\Dt$ satisfies the relations in Eqs.~\eqref{dualityconstraint1} and  
\eqref{dualityconstraint2} we use the graphical representations of the symmetries shown in Fig.~\ref{fig:2Dsym}. For example, consider the action of the TNO on the parity operator at face $f$: 
\begin{align*}
\M{D} \bdot P_{f} &=    \begin{tikzpicture}[Sbase]
 \DP; \node[OP]at (a){$P$};
\end{tikzpicture}\nonumber \\* & =  \begin{tikzpicture}[Sbase]
 \DP; \path (-30:{sqrt(3)})+(-30:0.4)node[OP]{$P$}+(90:0.4)node[OP]{$P$}+(-150:0.4)node[OP]{$P$};
\end{tikzpicture}\nonumber \\* & =\begin{tikzpicture}[Sbase]
 \DP; \node[OP]at (b){$Z$};\node[OP]at (c){$Z$};\node[OP]at (d){$Z$};
\end{tikzpicture} \nonumber \\* & = W_f \bdot \M{D}.
\end{align*}
Here, we have applied the symmetries of $\Ft[f]$ and $\Bt$ in succession to show that $\Dt$ satisfies Eq.~\eqref{dualityconstraint1}.
Similarly, for hopping operator we have:
\begin{align}
\M{D} \bdot S_{e} & = (-1)^{\eta_e}
    \begin{tikzpicture}[Sbase,scale=0.9]
 \path (30:{sqrt(3)})++(0,0.25,-0.5)node[OP]{$i\gammabar$}; 
\path (30:{2*sqrt(3)})++(0,0,-1)node[OP]{$\gamma$};\DSsw;
\path (30:{sqrt(3)})++(0,0.25,-0.5)++(0,0.45)node[red,scale=0.8]{2};
\path (30:{2*sqrt(3)})++(0,0,-1)++(0.3,0)node[red,scale=0.8]{1};
\end{tikzpicture}   \nonumber \\ &
= (-1)^{\eta_e}\begin{tikzpicture}[Sbase,scale=0.9]
\DSsw ;
\path (30:{sqrt(3)}) --node[OP,pos=1/4]{$\gamma$}node[OP,pos=3/4]{$\gamma$}
 (30:{2*sqrt(3)});
 \path (30:{sqrt(3)})++(150:{sqrt(3)/4})node[OP]{$P$};
 \path (30:{sqrt(3)})++(0,0.4) --node[red,scale=0.8,pos=1/4]{2}node[red,scale=0.8,pos=3/4]{1}
 ++(30:{sqrt(3)});
\end{tikzpicture}  \nonumber \\ &
=\begin{tikzpicture}[Sbase,scale=0.9]
\DSsw;
\path (60:3/2)++(0,0,1)node[OP]{$Z$};
\path (30:{3/2*sqrt(3)})++(0,-.25,0.5)node[OP]{$X$}; 
\end{tikzpicture}  \nonumber \\* &
=U_e  \bdot \M{D}.
\end{align}
Thus, $\Dt$ satisfies \eqref{dualityconstraint2} as well. This implies that $\Dt$ formed from $\Ft$, $\bar\Ft$, and $\Bt_\eta$ is indeed a representation of the operator-level duality of Ref. [\onlinecite{ChenYA17}].

The tensors $\Ft$, $\bar\Ft$, and $\Bt$ can be computed explicitly using their symmetries in Fig.~\ref{fig:2Dsym}. This is because the symmetries are independent, commute with each other, and square to identity. Hence, for $\Ft$ and $\bar\Ft$, they form a $\Z_2^5$ symmetry group, and for $\Bt$ they form a $\Z_2^3$ symmetry group.  Since $\Ft$ and $\bar\Ft$ belong to $2^5$ dimensional spaces, and $\Bt$ belongs to a $2^3$ dimensional space, their symmetries fix their values uniquely up to a normalization factor. It can be shown that the following tensors satisfy their respective symmetries: 
\begin{align} \label{FFbarBexp}
  \Ft[f]   \propto & \sum_{a,b,c}|c)_{f_{01}}|a)_{f_{12}}(b|_{f_{02}} (a+b+c|_f \nonumber\\
    \bar   \Ft[f]  \propto & \sum_{a,b,c} |b)_{f_{02}}(c|_{f_{01}}(a|_{f_{12}} (a+b+c|_f \nonumber\\
   \Bt[e] \propto & \sum_{a} |a)_{e}|a\rangle_{e}(a|_{e}.
\end{align}
Remember that all indices take values in $\{0,1\}$, and $a+b+c \equiv a+b+c \text{ mod }2$. 

\subsection{Bosonization of quantum states}

We are now able to define the bosonization of quantum states, wherein a fermionic state is bosonized by simply applying the bosonization TNO. Before providing a simple example, we comment on constraints of the state-level duality that arise from the symmetries of $\Dt$. In particular, we show that fermion parity odd states belong to the kernel of $\Dt$ and that $\Dt$ maps to bosonic states satisfying the constraint $G_v=1$ for all $v$. Hence, fermion parity even states are mapped to bosonic states in a certain $\Z_2$ gauge theory.

To show that fermion parity odd states are in the kernel of the bosonization TNO, we use that $\prod_f W_f=1$ on a closed manifold. This leads to:
\begin{align}\label{eq:domainD1}
    \Dt=\prod_f W_f \bdot \Dt = \Dt \bdot \prod_f P_f.
\end{align}
When $\Dt$ is applied to a fermionic state $|\psi_f)$, Eq.~\eqref{eq:domainD1} implies:
\begin{align}
    \Dt|\psi_f)=\Dt \bdot \prod_f P_f |\psi_f).
\end{align}
Thus, if $|\psi_f)$ is fermion parity odd, we have: ${\prod_f P_f |\psi_f)=-|\psi_f)}$, and it must be that $\Dt|\psi_f)=0$.

The constraints on the image of $\Dt$ can be determined using the relation in Eq.~\eqref{SePfidentity}. We see that:
\begin{align}
    \Dt =& \Dt \bdot  \left(\prod_{e:e_0=v}S_e\prod_{e:e_1=v}S_e \prod_{f:f_0,f_2=v} P_f\right) \nonumber\\
    =&  \left(\prod_{e:e_0=v}U_e\prod_{e:e_1=v}U_e \prod_{f:f_0,f_2=v} W_f\right) \bdot \Dt \nonumber\\
    =& G_v \bdot \Dt .\nonumber 
\end{align}
Hence, for any bosonic state $\langle \psi_b|$:
\begin{align}
     \langle \psi_b| \Dt = \langle  \psi_b| G_v \bdot \Dt, 
\end{align}
which implies that $\Dt$ projects to the $G_v=1$ subspace for each vertex $v$.

Now, we give a first example of the state-level duality and use the symmetries of $\Dt$ to show that the bosonization of an atomic insulator state yields a ground state of the toric code (a deconfined $\Z_2$ gauge theory). The atomic insulator state $|\psi_{AI})$ is the unique ground state of the Hamiltonian: $H_{AI}=-\sum_f P_f$. $H_{AI}$ is certainly unfrustrated, so $|\psi_{AI})$ satisfies $P_f|\psi_{AI})=|\psi_{AI})$ for all $f$. Applying $\Dt$ to $|\psi_{AI})$, we find:
\begin{align}
    \Dt |\psi_{AI})= \Dt \bdot P_f |\psi_{AI}) = W_f \bdot \Dt |\psi_{AI}),\, \forall f.
\end{align}
Therefore, the bosonized state $\Dt |\psi_{AI})$ is in the $+1$ eigenspace of $W_f$ for all $f$. Given the constraint on the image of $\Dt$, the bosonized state is also in the $+1$ eigenspace of $G_v$ for all $v$. Hence, $\Dt |\psi_{AI})$ is a ground state of the unfrustrated Hamiltonian ${H=-\sum_v G_v-\sum_f W_f}$.  Recalling the definition of $G_v$ defined in Eq.~\eqref{Gp}:
\begin{align}
    G_v=\prod_{e \supset v }X_{e}  \prod_{f: f_0=v} W_{f},
\end{align}
we see that the $G_v$ terms in $H$ can be replaced by $\prod_{e \supset v }X_{e}$ without changing the ground states. (${G_v=\prod_{e \supset v }X_{e}}$ in the subspace where $W_f=1$.) Thus, $\Dt |\psi_{AI})$ is a ground state of the toric code Hamiltonian ${H_{TC}=-\sum_v \prod_{e \supset v }X_{e}-\sum_f W_f}$. 

To gain intuition for the mapping, we consider acting with $\Dt$ on a state with non-trivial fermion occupancy. In particular, we apply a hopping operator $S_e$ at edge $e$ to the atomic insulator state $|\psi_{AI})$ to obtain a state with fermions at the two faces neighboring $e$. The image of $S_e|\psi_{AI})$ under $\Dt$ is:
\begin{align}
    \Dt \bdot S_e |\psi_{AI}) = U_e  \bdot \Dt |\psi_{AI}) =  U_e |\psi_{TC}\rangle.
\end{align}
$U_e$ (defined in Eq.\eqref{def:Ue}) creates a $\Z_2$ flux ($-1$ eigenvalue of $W_f$) at each face bordering the edge $e$ and moves $\Z_2$ charges ($-1$ eigenvalue of $\prod_{e \supset v }X_{e}$) to the $0$-vertices of $L_e$ and $R_e$. A $\Z_2$ flux bound to a $\Z_2$ charge has fermionic statistics -- it is an \textit{emergent} fermion.  Therefore, physical fermions are mapped to emergent fermions in the $\Z_2$ gauge theory. The gauge constraint $G_v=1,\, \forall v$ removes ambiguity in this mapping, since it enforces that charges are bound to fluxes, with the charges located at the $0$-vertex of the corresponding triangle. 

Any fermion parity even state can be created from $|\psi_{AI})$ by applying operators in $\mathcal{E}$. Hence, one strategy for mapping an arbitrary even fermion parity state $|\psi_f)$ is to identify an even operator $\mathcal{O}\big(\{S_e\}_e,\{P_f\}_f\big)$, written here explicitly in terms of the generators of $\mathcal{E}$, such that:
\begin{align}
   {|\psi_f)=\mathcal{O}{\big(}\{S_e\}_e,\{P_f\}_f{\big)}|\psi_{AI})}. 
\end{align}
Then, the duality maps:
\begin{align}
    |\psi_f) \rightarrow \mathcal{O}\big(\{U_e\}_e,\{W_f\}_f\big)|\psi_{TC}).
\end{align}
In general, it may be challenging to find an operator, expressed in terms of the generators of $\mathcal{E}$, that creates $|\psi_f)$ from $|\psi_{AI})$. Moreover, the analysis of bosonizing states, thus far, has only required the operator-level bosonization duality.  In the next section, we illustrate the true potential of the bosonization TNO. Given a fermion parity even state $|\psi_f)$ constructed from the contraction of local tensors, we show that $|\psi_f)$ can be bosonized by using $\Dt$ to modify each of the local tensors. The resulting state can then be written as a bosonic tensor network state. 

\section{Bosonization of fPEPS}
\label{sec: bosonization of fPEPS}
\begin{figure}
    \centering
     \begin{tikzpicture}[Sbase,xscale=1.7,yscale=1.4]
    \randt;
\Tposf{(1)}{(2)}{(7)}{(0,0,0.6)};
\Tnegf{(3)}{(2)}{(4)}{(0,0,0.6)};
\Tnegf{(3)}{(5)}{(6)}{(0,0,0.6)};
\Tnegf{(3)}{(4)}{(5)}{(0,-.2,0.5)};
\Tposf{(2)}{(4)}{(7)}{(0,-.1,0.6)};
\Tposf{(3)}{(2)}{(6)}{(0,0,0.7)};
\Tposf{(4)}{(5)}{(7)}{(0,-.2,0.5)};
\Tnegf{(8)}{(11)}{(7)}{(0,-0.2,0.5)};
\Tposf{(1)}{(7)}{(8)}{(0,-0.2,0.5)};
\Tnegf{(8)}{(10)}{(11)}{(0,0,0.7)};
\Tposf{(9)}{(8)}{(10)}{(0,-0.2,0.5)};
\Tnegf{(11)}{(5)}{(7)}{(-.1,0,0.5)};
    \end{tikzpicture}
    \caption{An example fPEPS on an arbitrarily triangulated torus. The square nodes represent the tensors $\M{T}$ and $\M{\bar{T}}$ in Eqs.~\eqref{Texp} and \eqref{Tpics}. The legs affixed to the center of the square nodes and pointing out of the page are the physical legs of the fPEPS. All other legs are contracted with a leg of a neighboring tensor.}
    \label{fig:fPEPS}
\end{figure}
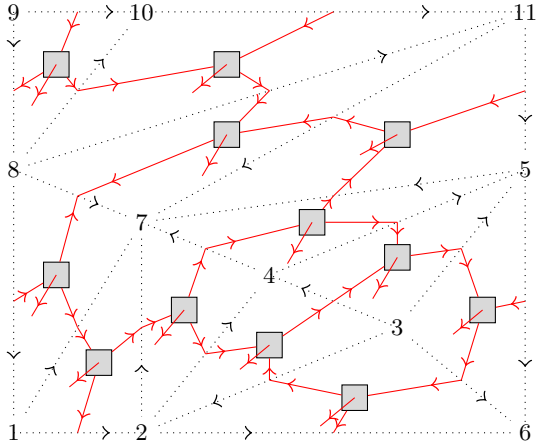

In the previous section, we introduced the bosonization of a fermionic state $|\psi_f)$ as the action of the bosonization TNO $\Dt$ on $|\psi_f)$. As we now show, the TNO is especially useful when the fermionic $|\psi_f)$ is represented as a fermionic tensor network state.  While the action of $\Dt$ on the fermionic tensor network state $|\psi_f)$ indeed yields a bosonic state $\Dt|\psi_f)=|\psi_b\rangle$, $|\psi_b\rangle$ is not manifestly a bosonic tensor network state. This is due, in part, to the $\Z_2$-graded virtual legs of the bosonization TNO.  However, if $|\psi_f)$ is in the form of a fermionic projected entangled pair state (fPEPS) (see Fig.~\ref{fig:fPEPS} for an example), we can explicitly rewrite $|\psi_b\rangle$ as a bosonic projected entangled pair state (bPEPS).  In this section, we give a detailed algorithm for converting bosonized fPEPS into bPEPS, which is well defined on arbitrary triangulations of orientable 2D manifolds without boundary. 

\subsection{Contracting the bosonization TNO with an fPEPS}

An fPEPS on a triangulated manifold is built from $\Z_2$-graded tensors $\M{T}[f']$ on positively oriented triangles and $\M{\bar T}[f']$ on negatively oriented triangles.

Assuming that the tensors are fermion parity even, they can be written in component form as:
\begin{align}\label{Texp}
\M{T}[f'] \equiv \sum_{a,b,c} T_{abc}^{(f')}|c)_{f'_{01}}|a)_{f'_{12}}(b|_{f'_{02}} (a+b+c|_{f'} 
 \nonumber\\
\M{\bar{T}}[f'] \equiv \sum_{a,b,c} \bar{T}^{(f')}_{abc}|b)_{f'_{02}}(a|_{f'_{12}}(c|_{f'_{01}}(a+b+c|_{f'},
\end{align}

where for generality, the tensor components are position dependent. $\M{T}[f']$ and $\M{\bar T}[f']$ can then be represented, respectively, as follows:
\begin{align} \label{Tpics}
 \begin{tikzpicture}[Sbase,scale=1]
 \pic[scale=1] at (0,0) {newt};
\end{tikzpicture}, \, \begin{tikzpicture}[Sbase,scale=1]
\pic[scale=1] at (0,0) {newtbar};
\end{tikzpicture}.
\end{align}
Fig.~\ref{fig:fPEPS} shows an fPEPS formed from contracting $\M{T}[f']$ and $\M{\bar T}[f']$ on an arbitrary triangulation of a torus.

In general, one can insert matrix product operators (MPO) before closing an fPEPS on a closed manifold. Though we believe our construction can be extended to such cases, in the interest of brevity and clarity, we restrict the discussion to fPEPS without any MPO insertions. 

To apply $\Dt$ to an fPEPS $|\psi_f)$, we contract the physical indices of $\Ft$ tensors with those of the $\M{T}$ tensors, and likewise, we contract $\bar \Ft$ tensors with $\M{\bar T}$ tensors. Thus, the first step in bosonizing an fPEPS is to calculate the tensors $\M{M_f}=\Ft\bdot \M{T}$ for positively oriented triangles and the tensors $\M{\bar{M_f}}= \M{\bar{F}\bdot \bar{T}} $ for negatively oriented triangles. Graphically, $\M{M_f}$ and  $ \M{\bar{M}_f}$ can be drawn, respectively, as:
 \begin{align}\label{cal:FT}
 \begin{tikzpicture}[Sbase,scale=1]
\newft;
\end{tikzpicture} ,
\begin{tikzpicture}[Sbase,scale=1]
\newftbar;
\end{tikzpicture} .
\end{align}
The bosonized state $|\psi_b\rangle$ is a tensor network state (in fact an fPEPS) generated by tensors $\M{M_f}$,  $\M{\bar{M}_f}$, as well as $\Bt_\eta[e]$ on edges. 
Since there are two layers of virtual legs to be contracted,  we refer to them as the ``state layer''  and the ``TNO layer''. Note that $\M{M_f}$ and $\M{\bar{M}_f}$ tensors have virtual legs on both layers, but $\Bt_\eta[e]$ is only on the TNO layer. 

While $|\psi_b\rangle$ is a tensor network state and an fPEPS, it is not generically a bPEPS, as there are fermionic virtual indices remaining. The challenge is then to re-express $|\psi_b\rangle$ as a bPEPS, or, in a sense, to convert the fermionic virtual legs to bosonic virtual legs.  We accomplish this by systematically accounting for the signs accrued in contracting the fermionic virtual legs -- the so-called \textit{Koszul signs}. To make our strategy clear, we first discuss {Koszul signs} and introduce the idea of a \textit{removable grading}. These concepts play a key role in the rest of this section, so we describe them in generality before returning to the problem of converting the fermionic virtual legs of $|\psi_b\rangle$ to bosonic virtual legs.

\subsection{Koszul signs and removable grading}
\label{subsub:koszul signs}

The bosonized fPEPS encodes a bosonic quantum state $|\psi_b\rangle=\sum_{\{\phi\}} C_\phi |\phi\rangle$, where the collection of $|\phi\rangle$ form a complete set of product states. The coefficients $C_\phi$ can be recovered from the bosonized fPEPS by fixing the physical indices according to $|\phi\rangle$ and summing over the virtual indices. Given the $\Z_2$-grading of the virtual legs, there are signs picked up upon re-ordering and contracting the $\Z_2$-graded vectors, which can contribute to the coefficient $C_\phi$.  A natural question is whether the grading of a virtual index is essential to the tensor network, i.e., if the grading of a particular virtual index is removed, does the value of the bosonized fPEPS change?

For illustration, consider the two simple graded tensors ${\M{A}=\sum_{a,a'}A_{aa'}(a'|_s  (a|_t}$ and ${\M{B} = \sum_{b,b'} B_{bb'}|b')_s  |b)_t}$. (We can think of $s$ and $t$ as indices corresponding to the state and TNO layers, respectively.) We want to calculate the tensor network (which is a scalar in this case) $\M{A}\bdot\M{B}$. Let us describe the contraction of these tensors as a two step process. In the first step, we contract the basis tensors:
\begin{align}
  (a'|_s  (a|_t\bdot |b')_s  |b)_t = & (-1)^{|b'||a|}\delta_{ab}\delta_{a'b'},
\end{align}
and in the second step, we calculate the components: 
\begin{align}
     \sum_{a,a,b,b'}(-1)^{|b'||a|}\delta_{ab}\delta_{a'b'}A_{aa'} B_{bb'}  =  \sum_{a,a'}(-1)^{|a'||a|}A_{aa'}  B_{aa'}.
\end{align} 
Notice that we produce an additional sign of $(-1)^{|b'||a|}$ in the basis contraction step due to the graded nature of indices. This is the key difference between virtual indices of fermionic and bosonic tensor networks -- bosonic indices do not produce any additional signs in basis contraction. We refer to these additional signs of basis contractions as Koszul signs. The point is that the grading of a virtual index contributes to the fPEPS only through possible Koszul signs. Therefore, we can remove the grading of a virtual index as long as we properly account for the Koszul signs. 

Sometimes this can be done simply by {picking a specific internal ordering for the fermionic tensors and interpreting their components, with respect to this ordering, as components of purely bosonic tensors.  We say that the grading of the virtual indices in the original fermionic tensor network can be removed, if the contraction of this new bosonic tensor network is the same as that of the original fermionic tensor network.  When we have some a priori internal ordering for the fermionic tensors in mind already -- one that does produce Koszul signs -- then we can refer to this process as \textit{changing the internal ordering} to eliminate the Koszul signs.}  For example, consider changing the internal ordering of $\M{A}$ to ${\M{A}=\sum_{a,a'}A_{aa'}(-1)^{|a'||a|} (a|_t(a'|_s }$.  Now, there is no Koszul sign in the basis contraction: ${ (a|_t(a'|_s \bdot|b')_s |b)_t  =\delta_{a,b}\delta_{a',b'}}$. Removing the grading from $\M{A}$ and $\M{B}$ yields:
\begin{align}
    \M{A_b}&=\sum_{a,a'}A_{aa'}(-1)^{|a'||a|} \langle a|_t\langle a'|_s \nonumber \\
    \M{B_b}&= \sum_{b,b'} B_{bb'}|b'\rangle_s  |b\rangle_t.
\end{align}
$\M{A_b}$ and $\M{B_b}$ are purely bosonic tensors, and they produce the same tensor network: $\M{A_b}\bdot\M{B_b}=\M{A}\bdot\M{B}$. Note that the grading is removable for only particular choices of the internal ordering.

In other cases, the Koszul signs can be accounted for with a removal of the grading, followed by an insertion of  additional operators into the tensor network. To see an example of this, consider the two even tensors $\M{A} = \sum_{a}A_{a}|a)_p(a|_q$ and $\M{B} = \sum_{b}B_{b}|b)_q(b|_p$. We then aim to compute the tensor network (a scalar) $\text{tr}\left[\M{A}\bdot \M{B}\right]$, where the $\bdot$ denotes the contraction of the $q$ leg and the trace over the $p$ index is to emphasize that we are contracting the first index with the last index to close the loop. 
Contracting the basis tensors yields:
\begin{align}
  \text{tr}\left[|a)_p(a|_q \bdot |b)_q(b|_p \right]  = \delta_{ab} \text{tr}\left[|a)_p (b|_p\right] = (-1)^{|a|}\delta_{ab}.
\end{align}
The grading of the $q$ vector did not produce a sign, so it can be removed without affecting the tensor network. However, if we try to remove the grading of the $p$ vector as well, the sign $(-1)^{|a|}$ is no longer accounted for. One way to reproduce the sign $(-1)^{|a|}$ is to insert a $Z=\sum_c(-1)^{|c|}|c\rangle \langle c|$ operator on leg $p$ after removing the grading. That is, grading removal gives bosonic tensors $\M{A_b}=\sum_a A_a |a\rangle_p \langle a|_q$ and $\M{B_b}=\sum_b B_b |b \rangle_q \langle b|_p$, which satisfy:
\begin{align}
    \text{tr}\left[\M{A_b}\bdot \M{B_b} \bdot  Z_p\right]= & \sum_{a,b,c}(-1)^{|c|}A_aB_b\text{tr}\left[|a\rangle_p\langle a|_q \bdot |b\rangle_q \langle b|_p \bdot |c\rangle_p \langle c|_p \right]\nonumber\\
    =&\text{tr}\left[\M{A}\bdot \M{B}\right]
\end{align}

When the grading of a virtual index can be accounted for by inserting an additional operator $O$ on un-graded indices, we will say that the grading is ``removable with $O$-insertion''.  

\subsection{Koszul signs in the bosonized fPEPS}
\label{sign in 2D}

We now return to the problem of re-writing the bosonized fPEPS as an explicit bPEPS. We will find that, for a particular internal ordering, the grading of the virtual legs in the bosonized fPEPS is removable with $(Z_t\otimes Z_s)^{\eta_e}$-insertion. Here, $Z_t$ is a Pauli $Z$ operator acting in the TNO layer, and $Z_s$ is a Pauli $Z$ operator acting in the state layer. In other words, the state represented by the bosonized fPEPS may be equivalently represented by the bPEPS obtained by removing the grading of the virtual legs (assuming a certain internal ordering) and applying $(Z_t\otimes Z_s)^{\eta_e}$ before contracting the tensors at each edge.  To show this, we will compute the Koszul signs explicitly. We will see that the Koszul signs have a nice geometric interpretation in terms of the branching structure of the triangulated manifold. 

\subsubsection{Simplifying the Koszul sign calculation}

{We begin by simplifying the problem. First, $\Bt_\eta=Z^{\eta_e}\bdot \Bt$ does not contribute to any Koszul signs, since $\Bt_{\eta}$ can always be contracted with one of its neighboring $\M{M_f}$ or $\M{\bar{M}_f}$ tensors without any change of internal ordering. This is possible due to the even parity of $\Bt_\eta$ and its simple, two-virtual-leg form $\Bt_\eta=\sum_{a} (-1)^{a\eta_e}|a)_{e}|a\rangle_{e}(a|_{e}$.}
Therefore, the Koszul signs accrued in contracting the tensors $\M{M_f}$, $\M{\bar{M}_f}$, and $\Bt_\eta$ are equivalent to the Koszul signs from directly contracting the $\M{M_f}$ and $\M{\bar{M}_f}$ tensors without $\Bt_\eta$. 

We continue to simplify the calculation of the Koszul signs by reducing $\M{M_f}$ and $\M{\bar{M}_f}$ from two layers of fermionic virtual legs as in Eq.~\eqref{cal:FT} to a single layer of fermionic virtual legs. The first step is to choose the following internal ordering for the tensors $\M{M_f}$ and $\M{\bar{M}_f}$, respectively:
\begin{align}\label{ftinorder}
    \begin{tikzpicture}[Sbase]
    \ftorder;
    \end{tikzpicture}\, \begin{tikzpicture}[Sbase]
    \ftbarorder;
    \end{tikzpicture}.
\end{align}
Notice that for outward pointing legs (ket vectors), the state layer index comes before the TNO layer index, while for inward pointing legs (bra vectors), the order is reversed.  
Letting $|a')_s$ and $|a)_t$ be the state and TNO layer vectors, respectively, then with the ordering in Eq.~\eqref{ftinorder}, we have $(a|_{t}(a'|_{s} \bdot |a')_{s}|a)_{t}=1$, and no Koszul sign is produced between $a'$ and $a$. Therefore, we can combine the legs and consider a composite index $(a',a)$ with tensors written in terms of $|a',a)_{st}$ and $(a',a|_{st}$. The Hilbert space of the composite leg corresponds to the Hilbert space of two spinless fermions. This is isomorphic to a single spinless fermion and a spin-1/2 under the isomorphism:
\begin{align}\label{spinless}
    |a',a) \leftrightarrow |a+a') |a\rangle.
\end{align}
Since the spin-1/2 degree of freedom does not affect the Koszul signs, we may disregard it for the present computation. 

In summary, we have reduced the calculation of the Koszul signs of the bosonized fPEPS to a calculation of the Koszul signs obtained in the contraction of single layer tensors with internal orderings inherited from Eq.~\eqref{ftinorder} and pictured below:
\begin{align}\label{KosT}
 \begin{tikzpicture}[Sbase,scale=0.8]
    \coordinate (C) at (-30:{sqrt(3)}) ;
    \draw[dotted,->-=0.75] (0,0)node[above]{0}--+(0:3)node[above,black]{2};
    \draw[dotted,->-=0.75] (0,0)--+(-60:3)node[below,black]{1};
    \draw[dotted,-<-=0.75] (0:3)--(-60:3);
    \draw[red,->-=0.75] (C)--(-60:3/2)node[left,black]{(i)};
     \draw[red,->-=0.75] (C)--+(-30:{sqrt(3)/2})node[right,black]{(ii)};
      \draw[red,-<-=0.75] (C)--(0:3/2)node[above,black]{(iii)};
    \node[Dtriangle] at (C) {$  $};
    \end{tikzpicture} , \, 
    \begin{tikzpicture}[Sbase,scale=0.8]
    \coordinate (C) at (30:{sqrt(3)}) ;
    \draw[dotted,->-=0.75] (0,0)node[below]{0}--+(0:3)node[below,black]{2};
    \draw[dotted,->-=0.75] (0,0)--+(60:3)node[above,black]{1};
    \draw[dotted,-<-=0.75] (0:3)--(60:3);
    \draw[red,-<-=0.75] (C)--(60:3/2)node[left,black]{(iii)};
     \draw[red,-<-=0.75] (C)--+(30:{sqrt(3)/2})node[right,black]{(ii)};
      \draw[red,->-=0.75] (C)--(0:3/2)node[below,black]{(i)};
    \node[Utriangle] at (C) {$ $};
    \end{tikzpicture}.
\end{align}
In Eq.~\eqref{KosT}, we have again used triangular nodes, but these tensors should not be confused with the four legged $\M{F}$ and $\M{\bar F}$ tensors. 
It should be noted that similar simplifications can be performed for the contraction (inner product) of any two fPEPS (built from fermion parity even local tensors). Consequently, the calculation of the Koszul signs below holds more generally than the application at hand -- turning a bosonized fPEPS into a bPEPS.

\subsubsection{Contraction of basis tensors}

As mentioned in section \ref{subsub:koszul signs}, the contraction of tensors can be performed in two steps: (i) the basis tensors are contracted, and (ii) the components are calculated. The Koszul signs arise only in the first step. Therefore, to calculate the Koszul signs, we focus on the contraction of basis tensors with the ordering in Eq.~\eqref{KosT}.  We denote the set of basis tensors at a positively oriented face $f$ as $\mathcal{Q}[f]$ and the set of basis tensors at a negatively oriented face $f$ as $\mathcal{\bar Q}[f]$. Explicitly, we have:
\begin{align}\label{basisset}
    \mathcal{Q}[f]= & \lbrace |a+b)_{f_{01}}|a)_{f_{12}}(b|_{f_{02}},\, a,b=0,1\rbrace \nonumber\\
    \mathcal{\bar Q}[f] =& \lbrace |b)_{f_{02}}(a|_{f_{12}}(a+b|_{f_{01}},\, a,b=0,1\rbrace.
\end{align} 
Note that the tensors in $\mathcal{Q}[f]$ and $\mathcal{\bar Q}[f]$ are fermion parity even by construction. $\M{M_f}$ and $\M{\bar{M}_f}$ are fermion parity even, so their component value for any fermion parity odd basis tensor is necessarily zero. Thus, we can disregard fermion parity odd basis tensors in computing the Koszul signs.

We now analyze the contraction of basis tensors in Eq.~\eqref{basisset}.  For each triangle, we have an independently chosen element of either $\mathcal{Q}[f]$ or $\mathcal{\bar Q}[f]$ (depending on the orientation of $f$). The resulting product of basis tensors evaluates to $0$, $-1$, or $1$. If an odd vector $|1)$ is paired with an even vector $|0)$ at any edge, then the product is $0$. (This is simply the statement that $(0|1)=(1|0)=0$.) Thus, the configurations of basis tensors that evaluate to a nonzero value must have odd legs paired at edges. Since the elements of $\mathcal{Q}[f]$ and $\mathcal{\bar Q}[f]$ have even fermion parity (an even number of odd legs), this implies that the odd legs form closed loops (on the dual lattice) for any configuration that gives a nonzero value. 


The computation of the Koszul signs then distills down to calculating the $\pm 1$ valued contraction of configurations with closed loops of $|1)$ states at edges. To formalize the problem, we define $g_e$ as the $\{0,1\}$ valued index at the edge $e$, and $\hatsigma( \lbrace g_e \rbrace)= \pm 1$ as the sign obtained by evaluating the tensor {{contractions}} corresponding to the configuration $\{ g_e \}$.


\subsubsection{Basis contraction and cohomology}

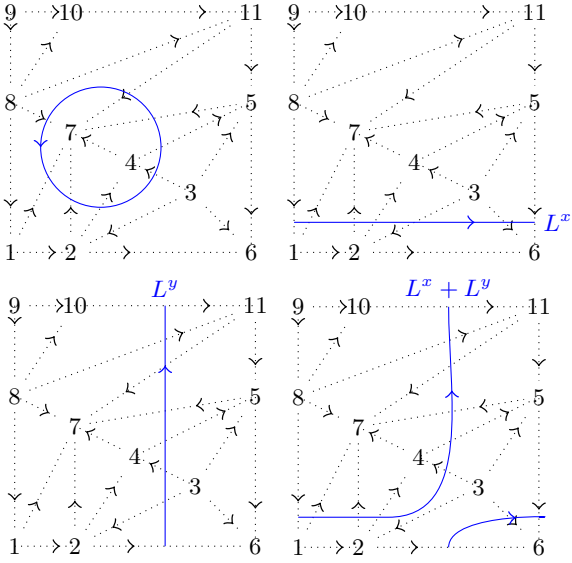
\begin{figure}
    \centering
   \begin{tikzpicture}[Sbase, scale=0.8]
    \randt;
\draw[blue,->-=0.5](1.5,1.75 ) circle (1);
 \end{tikzpicture}   \begin{tikzpicture}[Sbase, scale=0.8]
    \randt;
     \draw[blue,->-=0.75,thin] (0,0.5)--+(4,0)node[right]{$L^x$};

 \end{tikzpicture}
 
    \begin{tikzpicture}[Sbase,scale=0.8]
    \randt;
     \draw[blue,->-=0.75,thin] (2.5,0)--+(0,4)node[above]{$L^y$};
 \end{tikzpicture}
    \begin{tikzpicture}[Sbase, scale=0.8]
    \randt;
  \draw[blue,->-=0.75,thin] (0,0.5)--++(1.5,0)  to[out=0,in=-90] (2.5,4)node[above]{$L^x+L^y$};
    \draw[blue,->-=0.75,thin] (2.5,0)to[out=90,in=0](4,.5);
 \end{tikzpicture}
    \caption{Examples of $1$-cochains. Edges intersected by the blue line have coefficient $g_e=1$, while all other edges have $g_e=0$. The top left picture is an example of a contractible $1$-cocycle. The other three pictures are representative $1$-cocycles of the three non-trivial classes.}
    \label{fig:homoexample}
\end{figure}

To make our arguments precise, we find it convenient to describe configurations of odd edges using the language of cohomology. To this end, we define a $0$\textit{-cochain} as a sum $\sum_v g_v \mathbf{v}$, where $\mathbf{v}$ is a $\Z_2$-valued function of vertices such that $\mathbf{v}$ evaluates to $1$ on the vertex $v$ and $0$ otherwise, and $g_v$ are coefficients in $\Z_2$.  Similarly, \textit{$1$-cochains} and \textit{$2$-cochains} may be defined as sums $\sum_{e} g_e \mathbf{e}$ and $\sum_f g_f \mathbf{f}$, respectively.  A configuration of odd edges $\{g_e\}$ then naturally corresponds to the $1$-cochain $\sum_{e} g_e \mathbf{e}$. Furthermore, $j$-cochains can be added by combining the coefficients component wise, i.e., ${\sum_{e} g_e \mathbf{e}+\sum_{e} g'_e \mathbf{e}=\sum_{e} (g_e+g'_e) \mathbf{e}}$. 

The \textit{coboundary operator} $\delta$ from $j$-cochains to $j+1$-cochains is defined by:
\begin{align}
    \delta \mathbf{v} = & \sum_{e \supset v} \mathbf{e}, \quad \quad \delta \mathbf{e} = \sum_{f \supset e} \mathbf{f},
\end{align}
where the sum on the left is over all edges sharing the vertex $v$ and the sum on the right is over the two faces bordering the edge $e$.
For example, in Fig.~\ref{fig:homoexample}: 
 \begin{align}
     \delta \boldsymbol{\langle} \mathbf{4}\boldsymbol{\rangle} = & \boldsymbol{\langle} \mathbf{2,4} \boldsymbol{\rangle} +\boldsymbol{\langle} \mathbf{3,4} \boldsymbol{\rangle}+ \boldsymbol{\langle} \mathbf{4,5} \boldsymbol{\rangle} +\boldsymbol{\langle} \mathbf{4,7} \boldsymbol{\rangle} \nonumber\\
     \delta \boldsymbol{\langle} \mathbf{4,7}\boldsymbol{\rangle}= & \boldsymbol{\langle} \mathbf{2,4,7}\boldsymbol{\rangle}+\boldsymbol{\langle} \mathbf{4,5,7}\boldsymbol{\rangle}
 \end{align}
We call a cochain $C$ \textit{closed} if $\delta C=0$. Note that each of the $1$-cochains depicted in Fig.~\ref{fig:homoexample}, for example, are closed. More generally, a closed $1$-cochain, or $1$-cocycle, can be thought of as a sum of loops along the dual lattice. As such, the configurations $\{g_e\}$, obtained from basis contraction, are examples of $1$-cocycles.

A $1$-cochain $C$ is called a \textit{$1$-coboundary} if there exists a $0$-cochain $R$ such that $C=\delta R$. $\delta$ can be understood as a boundary operator on the dual lattice, so intuitively, a $1$-coboundary is a boundary of a region on the dual lattice. For example, the top left picture of Fig.~\ref{fig:homoexample} depicts a $1$-coboundary -- it is equal to $\delta R$ for $R= \boldsymbol{\langle} \mathbf{7} \boldsymbol{\rangle} +\boldsymbol{\langle}\mathbf{4} \boldsymbol{\rangle}$. In general, $1$-coboundaries are sums of contractible loops, which are generated by small loops $L_v \equiv \delta \mathbf{v}$ enclosing a single vertex. A configuration $L$ with a single, contractible loop is a $1$-coboundary of a $0$-cochain $R$ containing vertices enclosed by the loop, i.e., $L=\sum_{v \in L}L_{v}$ with the sum being over vertices $v$ enclosed by $L$. Some loops of odd edges such as the $1$-cocycles $L^x$, $L^y$, and $L^x+L^y$ in Fig.~\ref{fig:homoexample}, are \textit{non-contractible}. These are $1$-cocycles that cannot be written as $\delta R$ for any $0$-cochain $R$. 

We can further define an equivalence of $1$-cocycles where $C_1 \sim C_2$ if there exists a $0$-cochain $R$ such that $C_1 = C_2 + \delta R$. In other words, two $1$-cocycles are equivalent if one can be constructed from the other by appending, or adding contractible loops. Hence, all $1$-coboundaries belong to the same equivalence class -- the class of trivial $1$-cocycles. For a torus, it is well known that there are four inequivalent classes of $1$-cocycles. These may be represented by $L^x$, $L^y$, $L^x+L^y$, and $\delta R$ for a $0$-cochain $R$. Therefore, an arbitrary $1$-cocycle on a torus can be expressed as:
\begin{align}\label{Cdecomposition}
     C= g_{x}L^x +g_{y}L^y+\sum_v g_v L_v,
\end{align}
for some choice of $g_{x},g_{y},g_v \in \Z_2$.


\subsubsection{Koszul signs from a single loop}

Given that a $1$-cocycle can be decomposed in terms of constituent loops, as in Eq.~\eqref{Cdecomposition}, we begin by calculating the Koszul sign $\hatsigma(L)$ for configurations $L$ with a single loop of odd edges along the path $L$ in the dual lattice.  To propose an exact value for $\hatsigma(L)$, we introduce the following notation. We assign a direction to the path $L$ so that, with respect to a global orientation of the 2D manifold, the loop has a ``left side'' and a ``right side''.\footnote{More formally, let $v_1$ be the unit tangent vector along $L$, in the direction of $L$.  Then we say that the unit normal vector $v_2$ points to the ``left'' side of $L$ if $v_1 \wedge v_2$ is equal to the orientation of the underlying $2$D manifold.
}  $L$ overlaps with a triangle $f$ at two edges, and we call the common vertex of these two edges $f_L$. There are six possibilities for $f_L$: it can be a 0-, 1-, or 2-vertex of the triangle $f$, and it can lie to the left or to the right of the loop. We let $\bar{l}_L$ and $\bar{r}_L$ be the sets of $f_L$ for which $f_L$ is a 1-vertex of $f$ and is to the left or right of $L$, respectively.
We use $n(\bar{l}_L)$ and $n(\bar{r}_L)$ to denote the cardinality of $\bar{l}_L$ and $\bar{r}_L$. Then we have:
\begin{myprop}\label{prop:sigmaC}
\begin{align}\label{sigmaC}
     \hatsigma(L) = - (-1)^{\frac{1}{2}\boldsymbol{(}n(\bar{l}_L)-n(\bar{r}_L)\boldsymbol{)}}.
\end{align}
\end{myprop}

\begin{proof}
See Appendix \ref{app:proofofprop1}.
\end{proof}

\subsubsection{Winding number and Koszul signs}
\label{subsub: winding}
\begin{figure}
    \centering
  \begin{tikzpicture}[Sbase]
\interpos;
\end{tikzpicture} 
  \begin{tikzpicture}[Sbase]
\interneg;
\end{tikzpicture} 
    \caption{The branching structure is interpolated into the interior of each triangle to form the continuous, non-vanishing vector field  $\mathcal{V}$.}
    \label{fig:interpolation}
\end{figure}
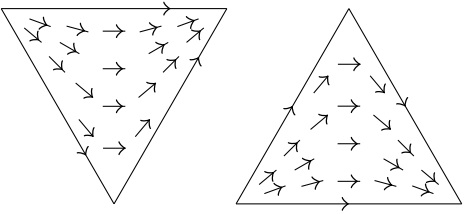
$\hatsigma(L)$ is closely related to the winding number of a certain vector field along the oriented path $L$. In particular, $\hatsigma(L)$ can be computed from the continuous, non-vanishing vector field $\mathcal{V}$ obtained from the branching structure by interpolating it into the interior of the triangles, as shown in Fig.~\ref{fig:interpolation} (see Ref. [\onlinecite{Gaiotto16}]). To calculate the winding number of $\mathcal{V}$ along $L$ we define $\hat n$ to be the left pointing unit normal vector of the loop $L$ and $\hat v$ to be the local vector of $\mathcal{V}$. Then, we integrate the {{derivate of the}} angle $\theta=\cos^{-1}(\hat v\cdot\hat n)$ between $\hat n$ and $\hat v$ along $L$. Given that $\theta$ is continuous, the change in $\theta$ around $L$ must be $2 \pi m$, where $m$ is an integer.  $m$ gives the winding number of $\mathcal{V}$ along $L$, which we denote as $w(L)$.\footnote{Note that the definition of winding number here is the winding number of the vector field relative to the normal vector of the loop $L$. We emphasize that this differs by a sign from a notion of the winding number of a vector field sometimes used in physics.} For definiteness, we choose clockwise rotation to be positive. 

\begin{myprop}\label{prop:wc}
\begin{align}
    w(L)=\frac{1}{2}\boldsymbol{(}n(\bar{l}_L)-n(\bar{r}_L)\boldsymbol{)}.
\end{align}
Equivalently, 
\begin{align}
   \hatsigma(L) = -(-1)^{w(L)}.
\end{align}
\end{myprop}
\begin{proof}
We consider the ways in which $L$ can pass through triangles and in each case, identify the change in $\theta=\cos^{-1}(\hat v \cdot \hat n)$. When $f_L$ is a $0$- or $2$-vertex, the total change in $\theta$ is $0$. This is illustrated in the following example, where $f_L$ is a $2$-vertex:
\begin{align}
    \begin{tikzpicture}[Sbase]
    \interpos;
    \draw[blue,-<-=0.75](3/2,-{sqrt(3)/2})--++(90:{sqrt(3)});
    \draw[blue,->-=0.75](3/2,-{sqrt(3)/2})--++(-30:{sqrt(3)});
    \draw[blue,->-=0.75](3/2,0.5)--++(0.35,0)node[right]{$\hat n$};
    \draw[blue,->-=0.75](3/2,-{sqrt(3)/2})++(-30:1.2)--++(60:0.35)node[right]{$\hat n$};
    \end{tikzpicture}.
\end{align}
The change in $\theta$ through the triangle above is 0, since the vector field is nearly parallel to $\hat n$ along the path. A similar argument applies whenever $f_L$ is a $0$- or $2$-vertex. Thus, the only crossings that can contribute to the winding number are when $f_L$ is a $1$-vertex. 


We first examine the case where $f_L$ is a $1$-vertex to the left of $L$, i.e., $f_L\in \bar{l}_L$. There are two such crossings:
\begin{align}
 \begin{tikzpicture}[Sbase,scale=1]
    \interneg;
    \draw[blue,-<-=0.9](3/2,{sqrt(3)/2})--++(150:{sqrt(3)});
    \draw[blue,->-=0.9](3/2,{sqrt(3)/2})--++(30:{sqrt(3)});
    \draw[blue,->-=0.9](3/2,{sqrt(3)/2})++(150:1.2)--++(60:0.35);
    \draw[blue,->-=0.9](3/2,{sqrt(3)/2})++(30:1.2)--++(120:0.35);
    \end{tikzpicture} \hspace{1cm} \begin{tikzpicture}[Sbase,scale=1]
    \interpos;
    \draw[blue,->-=0.9](3/2,-{sqrt(3)/2})--++(210:{sqrt(3)});
    \draw[blue,-<-=0.9](3/2,-{sqrt(3)/2})--++(-30:{sqrt(3)});
    \draw[blue,->-=0.9](3/2,-{sqrt(3)/2})++(210:1.2)--++(300:0.35);
    \draw[blue,->-=0.9](3/2,-{sqrt(3)/2})++(-30:1.2)--++(-120:0.35);
    \end{tikzpicture}.
\end{align}
(Note the triangle on the left is negatively oriented while the triangle on the right is positively oriented.)
For both crossings, moving along $L$, the vector field rotates clockwise relative to $\hat n$, and $\theta$ changes by $+\pi$.  If instead, $f_L$ is to the right of $L$ then the corresponding crossings are:
\begin{align}
    \begin{tikzpicture}[Sbase,scale=1]
    \interpos;
    \draw[blue,-<-=0.9](3/2,-{sqrt(3)/2})--++(210:{sqrt(3)});
    \draw[blue,->-=0.9](3/2,-{sqrt(3)/2})--++(-30:{sqrt(3)});
    \draw[blue,->-=0.9](3/2,-{sqrt(3)/2})++(210:1.2)--++(120:0.35);
    \draw[blue,->-=0.9](3/2,-{sqrt(3)/2})++(-30:1.2)--++(60:0.35);
    \end{tikzpicture} \hspace{1cm} \begin{tikzpicture}[Sbase,scale=1]
    \interneg;
    \draw[blue,->-=0.9](3/2,{sqrt(3)/2})--++(150:{sqrt(3)});
    \draw[blue,-<-=0.9](3/2,{sqrt(3)/2})--++(30:{sqrt(3)});
    \draw[blue,->-=0.9](3/2,{sqrt(3)/2})++(150:1.2)--++(240:0.35);
    \draw[blue,->-=0.9](3/2,{sqrt(3)/2})++(30:1.2)--++(-60:0.35);
    \end{tikzpicture}.
\end{align}
We see that, in this case, the vector field winds counterclockwise along $L$, and $\theta$ changes by $-\pi$.

In conclusion, whenever $f_L$ belongs to $\bar{l}_L$, $\theta$ changes by $\pi$, and when $f_L$ is in $\bar{r}_L$, $\theta$ changes by $-\pi$.  Accordingly, the winding number along $L$, with respect to $\hat n$, is: 
\begin{align}
    w(L)=\sum_{f_L} \left(\frac{1}{2} \delta_{f_L\in \bar{l}_L} -\frac{1}{2} \delta_{f_L\in \bar{r}_L}\right)= \frac{1}{2}\boldsymbol{(}n(\bar{l}_L)-n(\bar{r}_L)\boldsymbol{)}, 
\end{align}
{where $\delta_{f_L\in \bar{l}_L}$ and $\frac{1}{2} \delta_{f_L\in \bar{r}_L}$ are indicator functions for the sets $\bar{l}_L$ and $\bar{r}_L$, respectively.}
\end{proof}

{In Refs. [\onlinecite{Cimasoni07},\onlinecite{Gaiotto16}], it is argued that a function on loops of the form $-(-1)^{w(L)}$, such as $\hat \sigma(L)$, gives a \textit{quadratic refinement of the intersection pairing}. This is to say that, as a consequence of Prop.~\ref{prop:wc}, $\hat \sigma$ satisfies:
\begin{align}\label{quadraticrefinement}
    \hat \sigma(L_1+L_2)=(-1)^{\langle L_1, L_2 \rangle}\hat \sigma(L_1) \hat \sigma(L_2),
\end{align}
where $\langle L_1,L_2 \rangle$ is the intersection number ($\text{mod } 2$) of $L_1$ and $L_2$. For example, the non-contractible cycles $L^x$ and $L^y$ on a torus in Fig.~\ref{fig:homoexample} have an intersection number $\langle L^x, L^y \rangle = 1 \text{ mod }2$. Therefore, by Eq.~\eqref{quadraticrefinement}, we have: $\hat \sigma(L^x+L^y)=-\hat \sigma(L^x) \hat \sigma(L^y)$. }

{Importantly, Eq.~\eqref{quadraticrefinement} allows us to relate the sign $\hat \sigma(C)$ for a  general configuration $C=\sum_i L_i$ to the signs $\hat \sigma(L_i)$ of constituent loops. For a single contractible loop $L$, which can be decomposed into a sum of loops $L=\sum_{v \in L}L_{v}$, the sign $\hat \sigma(L)$ can be written as [using Eq.~\eqref{quadraticrefinement}]:
\begin{align}\label{Lcontractiblesingvert}
    \hat \sigma(L)=\prod_{v \in L}\hat \sigma(L_{v}).
\end{align}
The product in Eq.~\eqref{Lcontractiblesingvert} is over vertices enclosed by the loop $L$.}

{We call a vertex $v$ \textit{singular} if the loop $L_v$, enclosing only $v$, is such that $\hat \sigma(L_v)=-1$. Referring to Eq.~\eqref{Lcontractiblesingvert}, the sign $\hat \sigma(L)$ for a contractible loop $L$ can be computed by simply counting the singular vertices enclosed by $L$. Explicitly, $\hat \sigma(L)$ for a contractible loop $L$ is:
\begin{align}\label{sigmasv}
    \hat \sigma(L)=(-1)^{n_{sv}(L)},
\end{align}
where $n_{sv}(L)$ is the number of singular vertices enclosed by the loop $L$. This is a manifestation of Stokes' theorem for the winding number of the vector field along $L$. We note that, using Prop.~\ref{prop:sigmaC}, a vertex $v$ is singular if it is the $1$-vertex of $4m$ triangles, for an integer $m$. Alternatively, using Prop.~\ref{prop:wc}, $v$ is singular if $-(-1)^{w(L_v)}=-1$.}

\subsection{{Removing grading and choosing spin-structure}} \label{subsec: choosing SS}

{The function $\hat \sigma$ captures the Koszul signs accrued in the contraction of the fermionic virtual legs of the bosonized fPEPS. The goal of this section is to replace the $\Z_2$-graded virtual legs of the bosonized fPEPS with un-graded legs and simulate the Koszul signs given by $\hat \sigma$ by inserting Pauli $Z$ operators on certain bosonic virtual legs.} 

{More specifically, we first convert the fermionic virtual legs to bosonic virtual legs, i.e., with the internal ordering fixed, we map a fermion parity even state $|0)$ to an up spin $|0\rangle$ (in the $Z$ basis) and a fermion parity odd state $|1)$ to a down spin $|1\rangle$. The bosonic virtual legs fail to replicate the Koszul signs that were obtained by contracting the fermionic virtual legs. Thus, second, we fix this by choosing a set of edges $\eta$ (a choice of spin-structure) and including an extra $Z$ operator on edges $e \in \eta$ before contraction. When down spins $|1\rangle$ contract on an edge $e \in \eta$, the extra Pauli $Z$ operator results in a sign $-1$. We need to choose $\eta$ so that the contraction of a configuration $C$ of loops of down spins $|1\rangle$ yields the sign $\hat \sigma(C)$.}


{We begin by accounting for the Koszul signs $\hat \sigma(L)$ accrued by contractible loops $L$. Next, we discuss a matrix product operator (MPO) which captures the Koszul signs from non-contractible loops. We focus on the case when the manifold is a torus and only outline the procedure for more general 2D manifolds.}

\subsubsection{{Reproducing Koszul signs for contractible loops}}\label{contractibleloops}

 \begin{figure}
     \centering
    \begin{tikzpicture}[Sbase,scale=0.8]
    \randt;
  \path[](3)--node[]{$Z$}(4)--node[]{$Z$}(7);
  \path[](8)--node[]{$Z$}(10);
    \node at (2,-0.3) {(a)};
 \end{tikzpicture}
 \begin{tikzpicture}[Sbase,scale=0.8]
    \randt;
  \path[](7)--node[]{$Z$}(8);
  \path[](3)--node[]{$Z$}(4)--node[]{$Z$}(2);
    \node at (2,-0.3) {(b)};
 \end{tikzpicture}
     \caption{An arbitrary triangulation of a torus with the four singular vertices: $\langle 3 \rangle$, $\langle 7 \rangle$, $\langle 5=8 \rangle$, $\langle 2=10 \rangle$, $\langle 2 = 11 \rangle$.  (a) $Z$-operators placed at edges corresponding to the spin-structure ${\eta = \lbrace \langle 3,4\rangle,\langle 4,7 \rangle,\langle 8,10 \rangle  \rbrace}$. (b) $Z$-operators placed at edges for an alternative choice of spin-structure ${\eta = \lbrace \langle 3,4\rangle,\langle 4,2 \rangle,\langle 7,8 \rangle  \rbrace}$.  }
     \label{fig:spinstructureexample}
 \end{figure}
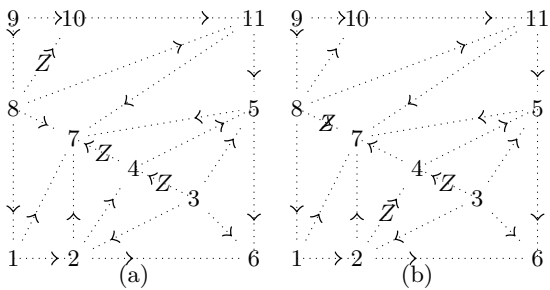

{Our strategy for accounting for $\hat \sigma(L)$, when $L$ is a contractible loop is to `pair-up' the singular vertices and construct the set $\eta$ from edges that connect the two singular vertices in each pair. More precisely, Stokes' theorem guarantees an even number of singular vertices on a closed manifold, so we can always find a set of edges $\eta$ such that the boundary of $\eta$ gives the set of singular vertices. Here, the boundary of $\eta$ is defined as the set of vertices that are endpoints of an odd number of edges in $\eta$. Intuitively, $\eta$ can then be understood as `pairing-up' singular vertices with each other through arbitrary paths. Fig.~\ref{fig:spinstructureexample} provides an example of choosing $\eta$ on a torus.}

{To replicate the sign $\hat \sigma(L)$, we insert $Z$ operators on all edges in $\eta$ (see Fig.~\ref{fig:spinstructureexample}). Now, in evaluating a configuration with a single loop $L$ of down spins $|1\rangle$, we incur the sign:
\begin{align} \label{sigmaeta}
    \sigma_\eta(L) \equiv (-1)^{n(L,\eta)},
\end{align}
where $n(L,\eta)$ denotes the number of common edges (or crossings) between the loop $L$ and the edges in $\eta$. Given our construction of $\eta$, $n(L,\eta)$ is equal ($\text{mod }2$) to the number of singular vertices enclosed in $L$. Therefore, for any contractible loop $L$:
\begin{align} \label{sigmaeta2}
    \sigma_\eta(L)=(-1)^{n_{sv}(L)},
\end{align}
in agreement with Eq.~\eqref{sigmasv}. Consequently, for any $1$-cocycle $C$ and $0$-cochain $R$, we have:
\begin{align}\label{trivialquadratic}
    \sigma_\eta(C+\delta R)=\sigma_\eta(C)\sigma_\eta(\delta R)=\sigma_\eta(C)\prod_{v \subset R}\sigma_\eta(L_v).
\end{align}
The product $\prod_{v \subset R}$ is over all vertices such that the coefficient of $\mathbf{v}$ in $R$ is nontrivial. }

{We note that a choice of $\eta$ can be modified by including any set of edges forming a contractible loop. We call two sets $\eta$ and $\eta'$ equivalent spin-structures, if one can be obtained from the other by appending contractible loops of edges.}

\subsubsection{{Reproducing Koszul signs for non-contractible loops}}

{$\sigma_\eta(C)$ simulates $\hat \sigma(C)$ for trivial $1$-cocycles $C$. This is sufficient to account for Koszul signs when the fermionic system is defined on a sphere or an infinite plane. However, $\sigma_\eta(C)$ does not capture the sign $\hat \sigma(C)$ when $C$ is a non-trivial $1$-cocycle. To account for the Koszul signs on a torus or higher genus manifolds, we insert MPOs along non-contractible loops to perform a certain sum over inequivalent spin-structures.  In the following, we describe the case of a torus in detail and only sketch the generalization to higher genus manifolds.}

{To start, we consider a particular triangulation of a torus without any singular vertices, as shown in Fig.~\ref{fig:triangulartorus}(a-c). Since there are no singular vertices, $\hat \sigma(L)=1$ for any contractible loop $L$. For non-contractible loops, however, the Koszul signs are non-trivial.  To see this, we let $L^x$ and $L^y$ be distinct non-contractible loops lying parallel to the $x$-axis and $y$-axis, respectively. The specific choices of $L^x$ and $L^y$ do not matter, because contractible loops can freely be appended without changing $\hat \sigma(L^x)$ and $\hat \sigma(L^y)$. This follows from Eq.~\eqref{quadraticrefinement} and the fact that there are no singular vertices. Now, using either Prop.~\ref{prop:sigmaC} or Prop.~\ref{prop:wc}, one finds:
\begin{align}\label{noncontractguwen}
    \hat \sigma(L^x)=\hat \sigma(L^y)=\hat \sigma(L^x+L^y)=-1.
\end{align}}

{Hence, after converting the $\Z_2$-graded virtual legs to bosonic virtual legs, a modification is necessary to simulate the sign in Eq.~\eqref{noncontractguwen}. A possible solution is to insert Pauli $Z$ operators along non-contractible loops. Naively, we can insert $Z$ operators along only the $x$-axis [Fig.~\ref{fig:triangulartorus}(a)], only the $y$-axis [Fig.~\ref{fig:triangulartorus}(b)], or both the $x$-axis and the $y$-axis [Fig.~\ref{fig:triangulartorus}(c)]. We find that all of these options fail to reproduce the sign in Eq.~\eqref{noncontractguwen}. The solution is a certain superposition of these options, which can be expressed using an MPO. Before describing this MPO, we develop some notation and discuss the effects of inserting $Z$ operators along the axes.}

{First, we define the spin-structure $\eta^x$, which contains the edges along the $x$-axis.
The product of $Z$ operators applied along the edges in $\eta^x$ can be expressed as $\prod_e Z_e^{\eta^x}$, where, in this expression, $\eta^x$ is the indicator function for the set $\eta^x$. The operator $\prod_e Z_e^{\eta^x}$ is pictured in Fig.~\ref{fig:triangulartorus}(a). Then, letting $\sigma_{\eta^x}(L)$ be the sign obtained in contracting a configuration of down spins along $L$ with the added operator $\prod_e Z_e^{\eta^x}$, we have:
\begin{align}\label{sigmax}
    \sigma_{\eta^x}(L^x)=1,\quad \sigma_{\eta^x}(L^y)=-1, \quad \sigma_{\eta^x}(L^x+L^y)=-1.
\end{align}}
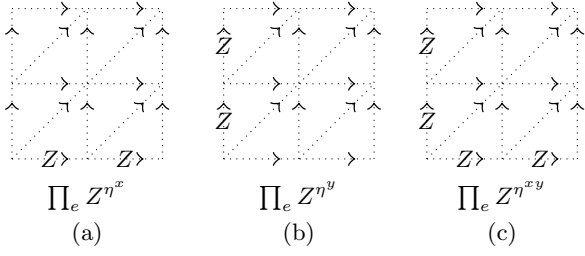
\begin{figure}
    \centering


\begin{minipage}{0.3\linewidth} 
 \begin{tikzpicture}[Sbase,scale=1]
\foreach \j in {0,1,2}
{\draw[dotted,OES={AR1=0.75}](0,\j)--(1,\j)--(2,\j);
\draw[dotted,OES={AR1=0.75}](\j,0)--(\j,1)--(\j,2); }
\foreach \j in {(0,0),(0,1),(1,0),(1,1)}
{\draw[dotted,OES={AR1=0.75}] \j --++(45:{sqrt(2)}); }
\path (0.5,0)node{$Z$}(1.5,0)node{$Z$};
\node at (1,-0.5) {$\prod_e Z^{\eta^x}$};
\node at (1,-1) {(a)};
\end{tikzpicture}
\end{minipage}
\begin{minipage}{0.3\linewidth}
\begin{tikzpicture}[Sbase,scale=1]
\foreach \j in {0,1,2}
{\draw[dotted,OES={AR1=0.75}](0,\j)--(1,\j)--(2,\j);
\draw[dotted,OES={AR1=0.75}](\j,0)--(\j,1)--(\j,2); }
\foreach \j in {(0,0),(0,1),(1,0),(1,1)}
{\draw[dotted,OES={AR1=0.75}] \j --++(45:{sqrt(2)}); }
\path (0,0.5)node{$Z$}(0,1.5)node{$Z$};
\node at (1,-1) {(b)};
\node at (1,-0.5) {$\prod_e Z^{\eta^y}$};
\end{tikzpicture}
\end{minipage}
\begin{minipage}{0.3\linewidth}
\begin{tikzpicture}[Sbase, scale=1]
\foreach \j in {0,1,2}
{\draw[dotted,OES={AR1=0.75}](0,\j)--(1,\j)--(2,\j);
\draw[dotted,OES={AR1=0.75}](\j,0)--(\j,1)--(\j,2); }
\foreach \j in {(0,0),(0,1),(1,0),(1,1)}
{\draw[dotted,OES={AR1=0.75}] \j --++(45:{sqrt(2)}); }
\path (0.5,0)node{$Z$}(1.5,0)node{$Z$};
\path (0,0.5)node{$Z$}(0,1.5)node{$Z$};
\node at (1,-1) {(c)};
\node at (1,-0.5) {$\prod_e Z^{\eta^{xy}}$};
\end{tikzpicture} 
\end{minipage}
    \caption{Triangulation of a torus without any singular vertices, but with $Z$-operators placed along (a) the $x$-axis (b) the $y$-axis (c) both the $x$-axis and the $y$-axis.}
    \label{fig:triangulartorus}
\end{figure}

{Next, we define the spin-structures $\eta^y$ and $\eta^{xy}$ similarly.  $\eta^y$ is the set of edges along the $y$-axis and corresponds to the insertion of the operator $\prod_e Z_e^{\eta^y}$, depicted in Fig.~\ref{fig:triangulartorus}(b). In this case, the sign accrued in contracting the bosonic legs is:
\begin{align}\label{sigmay}
    \sigma_{\eta^y}(L^x)=-1,\quad \sigma_{\eta^y}(L^y)=1, \quad \sigma_{\eta^y}(L^x+L^y)=-1.
\end{align}
If we instead insert $Z$ operators along both the $x$-axis and $y$-axis, as in Fig.~\ref{fig:triangulartorus}(c), and define $\eta^{xy}$ as the union of $\eta^x$ and $\eta^y$, we obtain the signs:
\begin{align}\label{sigmaxy}
    \sigma_{\eta^{xy}}(L^x)=-1,\quad \sigma_{\eta^{xy}}(L^y)=-1, \quad \sigma_{\eta^{xy}}(L^x+L^y)=1.
\end{align}}

{In each case above [Eqs.~\eqref{sigmax}-\eqref{sigmaxy}], the operator insertion fails to replicate the sign in Eq.~\eqref{noncontractguwen}. However, the sign in Eq.~\eqref{noncontractguwen} can be simulated using the following superposition of operators:
\begin{align}\label{MPOtargetoperator}
    \frac{1}{2}\left( -1+\prod_{e}Z_e^{\eta^x}+\prod_{e}Z_e^{\eta^y}+\prod_eZ_e^{\eta^{xy}} \right).
\end{align}
Explicitly, the sign obtained in contracting an arbitrary loop of down spins $L$ is then:
\begin{align}\label{superpositionofsigns1}
    \frac{1}{2}\boldsymbol{(}-1+\sigma_{\eta^x}(L)+\sigma_{\eta^y}(L)+\sigma_{\eta^{xy}}(L)\boldsymbol{)}.
\end{align}
One can check that, for loops $L^x$, $L^y$, and $L^x+L^y$, the sign given by Eq.~\eqref{superpositionofsigns1} matches the sign in Eq.~\eqref{noncontractguwen}. Furthermore, the sign in Eq.~\eqref{superpositionofsigns1} agrees with $\hat \sigma$ on all loops.}

\begin{figure}
    \centering
    \begin{tikzpicture}[Sbase,scale=1.5]
\ttorus;
\draw[](0,2)--(0,0)--(2,0);
\draw[] (-0.25,1.5)--node[Sq={0},scale=0.7]{}++(0.5,0);\draw[] (-0.25,.5)--node[Sq={0},scale=0.7]{}++(0.5,0);
\draw[] (1.5,-0.25)--node[Sq={0},scale=0.7]{}++(0,.5);\draw[] (.5,-0.25)--node[Sq={0},scale=0.7]{}++(0,.5);
\node[Bcir,scale=1] at (0,0){};\node[Bcir,scale=1] at (0,0){$\M{G}$};
\end{tikzpicture}
    \caption{Triangulation of a torus without any singular vertices, and with the MPOs generated by $\M{W}$ (square nodes) and $\M{G}$ (circular nodes). The MPOs generated by $\M{W}$ wrap around both the $x$-axis and the $y$-axis and the $\M{G}$ tensor is placed at their intersection. }
    \label{fig:insertmpo}
\end{figure}
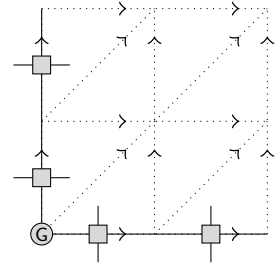

{The operator in Eq.~\eqref{MPOtargetoperator} can be represented using MPOs. We start by considering an MPO of the form:
\begin{align}
\begin{tikzpicture}[Sbase]
  \draw[](0,0)--(6,0);\draw[loosely dotted](-0.5,0)--(0,0);
  \draw[loosely dotted](6,0)--+(0.5,0);
    \foreach \j in {1,3,5}
    { \draw[](\j,0)++(0,-0.5)--++(0,1); \node[Sq={0},scale=0.9]  at (\j,0){$\M{W}$} ; }
\end{tikzpicture},
\end{align}
generated by the local tensors $\M{W}$, given by:
\begin{align}\label{Wtensor}
     \begin{tikzpicture}[Sbase,scale=0.6]
     \draw[](-1,0)node[below]{$p$}--++(2,0)node[below]{$q$};\draw[](0,-1)node[below]{$r$}--++(0,2)node[above]{$s$};\node[Sq={0},scale=1] {$\M{W}$};
     \end{tikzpicture}=&\sum_{a,b}(-1)^{(a)(b)} |a\rangle_p|b\rangle_r \langle b|_s\langle a|_q \nonumber\\
     = & |0\rangle I \langle 0|+ |1\rangle Z \langle 1|.
\end{align}
When the virtual (horizontal) legs take value $0$, $W$ acts as the identity, and when they take value $1$, $W$ acts as a $Z$ operator. Therefore, $\M{W}$ generates a {{controlled $Z$}} operator of the form $\ldots III\ldots + \ldots ZZZ\ldots $.}

{If we insert this MPO on the virtual level of the bosonic tensor network, say, along the $x$-axis, it is equivalent to inserting the operator $1+\prod_eZ_e^{\eta^x}$. Similarly, inserting it along the $y$-axis is equal to the operator $1+\prod_eZ_e^{\eta^y}$.  If we insert the MPO along both the $x$-axis and the $y$-axis they cross as a single vertex, and we link the MPOs together at their intersection using another tensor $\M{G}$ (see Fig.~\ref{fig:insertmpo}). We take $\M{G}$ to be: 
 \begin{align}\label{Lwtensor}
     \begin{tikzpicture}[Sbase,scale=0.6]
     \draw[](-1,0)node[below]{$p$}--++(2,0)node[below]{$q$};\draw[](0,-1)node[below]{$r$}--++(0,2)node[above]{$s$};\node[Bcir,scale=1] {$\M{G}$};
     \end{tikzpicture}=&\frac{1}{2}\sum_{a,b}(-1)^{(a+1)(b+1)} |a\rangle_p|b\rangle_r \langle b|_s\langle a|_q\nonumber\\
     =&-\frac{1}{2}|0\rangle_p|0\rangle_r \langle 0|_s\langle 0|_q+\frac{1}{2}|0\rangle_p|1\rangle_r \langle 1|_s\langle 0|_q \nonumber\\ &+\frac{1}{2}|1\rangle_p|0\rangle_r \langle 0|_s\langle 1|_q+\frac{1}{2}|1\rangle_p|1\rangle_r \langle 1|_s\langle 1|_q.
 \end{align}
Now, when virtual legs of the MPOs in both the $x$ and $y$ direction take value $0$, $\M{G}$ is $-\frac{1}{2}$, and otherwise it is $\frac{1}{2}$. Thus, the total MPO produces the superposition of operators: 
\begin{align}
     \frac{1}{2}\left( -1+\prod_{e}Z_e^{\eta^x}+\prod_{e}Z_e^{\eta^y}+\prod_eZ_e^{\eta^{x}} \prod_eZ_e^{\eta^{y}}\right).
\end{align}
Recalling that $\eta^{xy}$ is the union of $\eta^x$ and $\eta^y$, we see that the operator above is equivalent to the operator in Eq.~\eqref{MPOtargetoperator}. The Koszul signs, therefore, can be accounted for using the MPO generated by $\M{W}$ and $\M{G}$, pictured in Fig.~\ref{fig:insertmpo}.} 

{In effect, the MPO implements a sum over inequivalent spin-structures to capture the Koszul signs of non-contractible loops. The tensor $\M{G}$ dictates the particular sum over spin-structures and, in general, depends on the branching structure. To see this, we next consider a general triangulation of a torus, where we must incorporate $\sigma_\eta$, accounting for singular vertices, with the sum over inequivalent spin-structures given by the MPO. }

\subsubsection{{General triangulation of a torus}}\label{generaltriG}

{Thus far, we have argued that we can account for Koszul signs in the following two cases: (i) trivial $1$-cocycles formed from contractible loops of $|1)$ states and (ii) non-contractible loops formed by $|1)$ states in the absence of singular vertices. To reproduce the Koszul signs from contraction on a general triangulation of a torus, we then must be able to simulate the Koszul signs from non-contractible loops in the presence of singular vertices.  We will find that we require a branching structure dependent choice of the tensor $\M{G}$ to obtain an appropriate sum over inequivalent spin-structures.}  

{The first step is to account for the Koszul signs of contractible loops, as in \ref{contractibleloops}. That is, we choose a set of edges $\eta$ such that the edges in $\eta$ pair up the singular vertices and insert $Z$ operators at the edges in $\eta$. The sign from evaluating a loop $L$ of down spins is then $\sigma_\eta(L)$ as in Eqs.~\eqref{sigmaeta} and \eqref{sigmaeta2}.}

{After choosing $\eta$ we can account for the Koszul signs from non-contractible loops. As before, we choose representative non-contractible loops $L^x$ and $L^y$ lying parallel to the $x$- and $y$-axis, respectively, such as those in Fig.~\ref{fig:homoexample}.  However, unlike the case with no singular vertices, the choice of $L^x$ and $L^y$ matters. For example, the sign $\hat \sigma(L^x)$ changes if $L^x$ is shifted across a singular vertex. Likewise the sign of $\sigma_\eta(L^x)$ changes if $L^x$ is shifted across a singular vertex. Therefore, to remove the ambiguity, we define the $\{0,1\}$ valued $\alpha_x$ and $\alpha_y$ by:
\begin{align}\label{alpha1}
    (-1)^{\alpha_x}&\equiv{\hat \sigma(L^x) }/{ \sigma_\eta(L^x)} \\ \label{alpha2}
    (-1)^{\alpha_y}&\equiv{\hat \sigma(L^y) }/{ \sigma_\eta(L^y)} \\ \label{alpha3}
    (-1)^{\alpha_{x}+\alpha_y+1}&={\hat \sigma(L^x+L^y) }/{ \sigma_\eta(L^x+L^y)}.
\end{align}
We emphasize that the expressions above are independent of the particular choice of representative non-contractible loops $L^x$ and $L^y$, using Eqs.~\eqref{quadraticrefinement} and 
\eqref{trivialquadratic}.}

{We now only need to reproduce the signs on the left side of Eqs.~\eqref{alpha1},~\eqref{alpha2}, and \eqref{alpha3} for non-contractible loops using the MPO generated by $\M{W}$ and $\M{G}$. A superposition of operators that yields these signs from contracting bosonic legs is:
\begin{align}
\frac{1}{2}(-1)^{\alpha_x\alpha_y}\bigg(1&+(-1)^{\alpha_y}\prod_{e}Z_e^{\eta^x}\nonumber \\
&+(-1)^{\alpha_x}\prod_{e}Z_e^{\eta^y}+(-1)^{\alpha_x+\alpha_y}\prod_eZ_e^{\eta^{xy}}\bigg).
\end{align}
It can be shown that this operator is generated by $\M{W}$ and $\M{G}$ with the components of $\M{G}$ given by:
\begin{align} \label{gab2}
    G_{ab}=\frac{1}{2}(-1)^{(\alpha_y+a)(\alpha_x+b)}.
\end{align}
For the special case of the triangulation in Fig.~\ref{fig:triangulartorus}, $\alpha_x=\alpha_y=1$ and Eq.~\eqref{gab2} gives $G_{ab}=\frac{1}{2}(-1)^{(a+1)(b+1)}$, which matches our previous result. Given the spin structure in Fig.~\ref{fig:etatorus}, $\alpha_x=\alpha_y=0$. Thus, in this case, to capture the Koszul signs from non-contractible loops, the components of $\M{G}$ should be $G_{ab}=\frac{1}{2}(-1)^{ab}$.}

\subsubsection{{Higher genus manifolds}}

\begin{figure}
    \centering
  \includegraphics[scale=0.1,trim={2.5cm 10 10 0},clip]{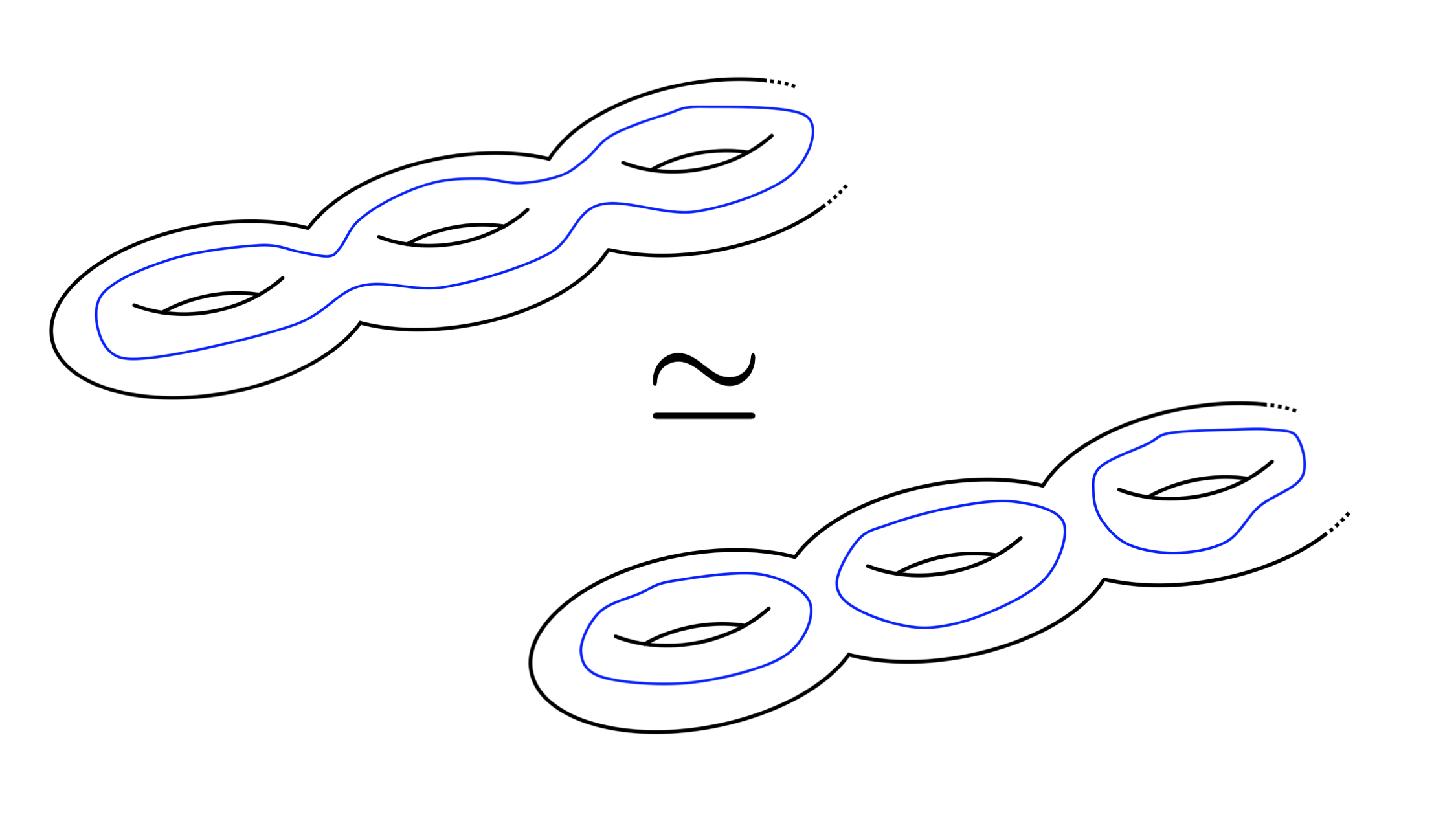}  
\caption{A cocycle on a genus $g$ manifold is cohomologous to a $\Z_2$ sum of cocycles on the component torii. The non-trivial cocycles on independent torii on the right hand side have a trivial intersection number. }
    \label{fig:HigherGenusDecomp}
\end{figure}

\begin{figure}
    \centering
  \includegraphics[scale=0.08,trim={5cm 10 15 0},clip]{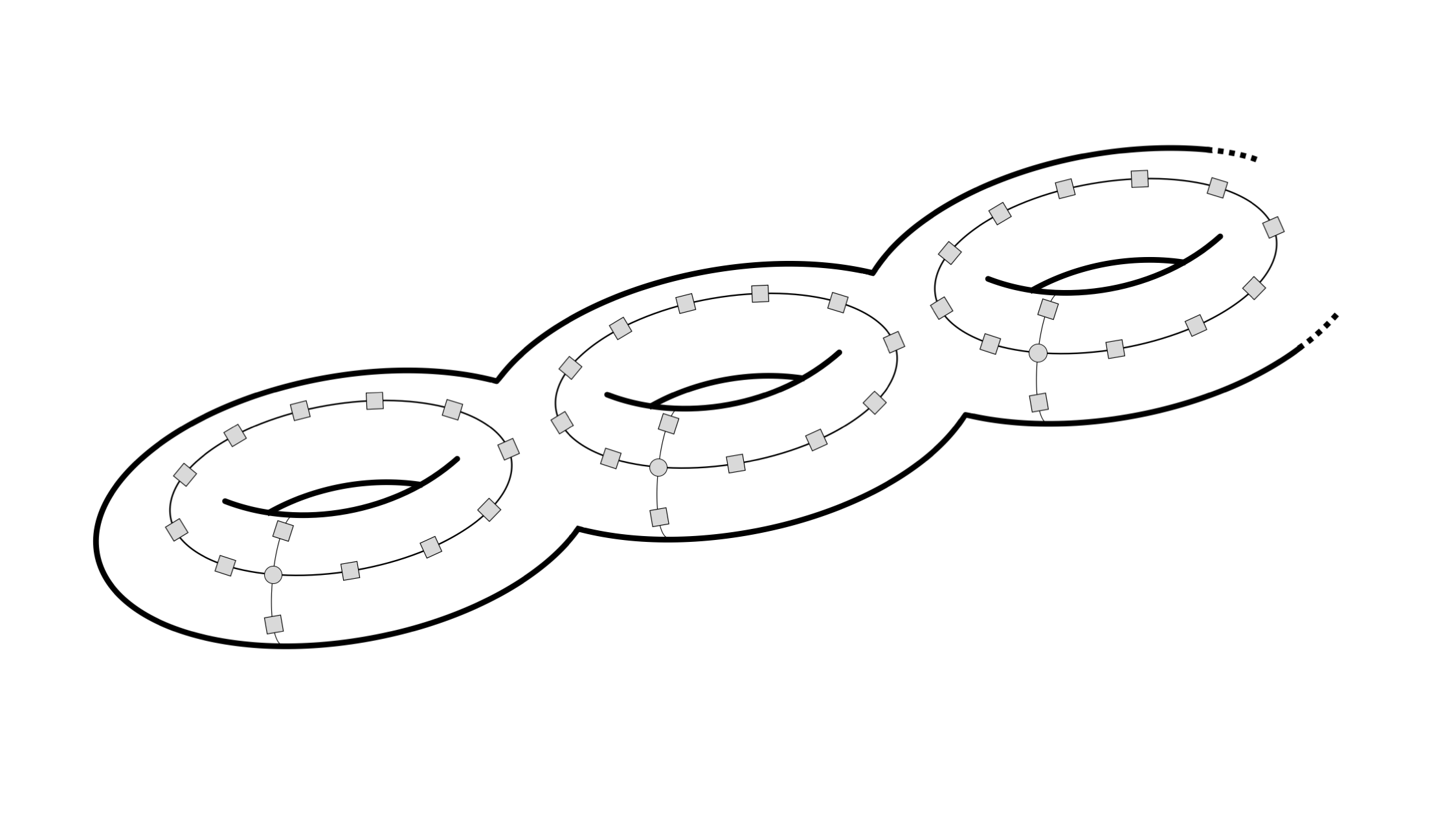}   
\caption{MPOs generated by $\M{W}$ and $\M{G}$ are inserted on each component torii. The $\M{G}$ tensor may differ between the torii.}
    \label{fig:HigherGenusMpo}
\end{figure}

We briefly describe how our results can be extended to higher genus manifolds. We exploit the fact that any 2D oriented manifold $M$ with genus $g$ is topologically equivalent to a manifold constructed from the connected sum $\#$ of a sphere with $g$ torii: $M\simeq S^2\#T^2\#\cdots\#T^2$. Furthermore, given a decomposition of $M$ into a connected sum of a sphere and torii, any cocycle $C$ can be written as:
\begin{align}\label{cocycledecompHG}
    C = \sum_{j=1}^{g} a_{j,x} L^{j,x} + \sum_{j=1}^{g} a_{j,y}L^{j,y} + \delta R.
\end{align}
Here, $L^{j,x}$ and $L^{j,y}$ denote generators of the non-trivial cocycles around the $j^\text{th}$ torus in the connected sum decomposition, and $\delta R$ gives a trivial cococyle. According to Eq.~\eqref{cocycledecompHG}, any cocycle $C$ is cohomologous to one of the form [see Fig.~\ref{fig:HigherGenusDecomp}]:
\begin{align}\label{cocycledecompHG2}
    \sum_{j=1}^{g} a_{j,x} L^{j,x} + \sum_{j=1}^{g} a_{j,y}L^{j,y}.
\end{align}

With this, we can now describe how to account for the Koszul signs from contraction on a genus $g$ manifold. As before, the Koszul signs from contractible loops can be taken care of by making a choice of $\eta$ and inserting Pauli $Z$ operators along the edges in $\eta$. As for non-trivial cocycles, we first decompose the cocycle as in Eq.~\eqref{cocycledecompHG}. Then, we identify the cohomologous cocycle given in Eq.~\eqref{cocycledecompHG2}, which differs by a trivial cocycle. (The difference in the Koszul sign between the cohomologous cocycles is already accounted for by the choice of $\eta$.) Using that the Koszul sign corresponds to a quadratic refinement of the intersection pairing [Eq.~\eqref{quadraticrefinement}], the computation of the Koszul signs for a cocycle in the form of Eq.~\eqref{cocycledecompHG2} reduces to a computation of the Koszul signs for the loops $L^{j,x}$ and $L^{j,y}$. This is because loops belonging to different torii have trivial intersection number: 
\begin{align}
   \langle L^{x,j}, L^{x,k}\rangle =  \langle L^{x,j}, L^{x,k}\rangle  =  \langle L^{x,j}, L^{x,k}\rangle  
   =0 \text{ mod }2,
\end{align}
for all $j \neq k$. Therefore, the problem is reduced to that of $g$ independent arbitrarily triangulated torii. The Koszul signs of non-contractible loops can be accounted for by inserting MPOs generated by $\M{W}$ and $\M{G}$ as in Fig.~\ref{fig:HigherGenusMpo}. A similar strategy as in section \ref{generaltriG} can be used to choose $\M{G}$ at each intersection of the MPOs.

\subsubsection{{Grading removal for the bosonized fPEPS}}

{Now, we return to the original problem of writing bosonized fPEPS as bPEPS. To compute the Koszul signs, we worked with a single layer of fermionic virtual legs, while a bosonized fPEPS has both the state layer and TNO layer of fermionic virtual legs. Therefore, we need to translate our results for accounting for Koszul signs back to the case of two layers of virtual legs.}

{To simplify the calculation of the Koszul sign, we noticed that [Eq.\eqref{spinless}], with the chosen ordering of the fermionic virtual legs, the pair of virtual legs $|a)_t|a')_s$ could be mapped to a spinless fermionic degree of freedom and a spin-1/2 via the isomorphism:
\begin{align}
    |a)_t|a')_s \rightarrow |a+a')|a\rangle.
\end{align}
Then, we worked only with the fermionic leg. Ultimately, we converted the fermionic legs $|a+a')$ to bosonic legs $|a+a'\rangle$ with the addition of $Z$ operators on certain edges. A $Z$ operator acting on $|a+a'\rangle$ correponds to acting with a parity operator on $|a+a')$ or the operator $P_s\otimes P_t$ on $|a)_t|a')_s$. Therefore, to replace the two layers of fermionic legs with two layers of bosonic legs: $|a)_t|a')_s \rightarrow |a\rangle_t|a'\rangle_s$, we see that we need to apply operators $Z_t \otimes Z_s$ at edges to account for Koszul signs. }

{In summary, we convert the two layers of fermionic legs (with the fixed internal ordering) to bosonic legs. Then, we insert $(Z_t\otimes Z_s)^{\eta_e}$ at every edge to account for the Koszul signs from contractible loops. To account for the Koszul signs from non-contractible loops, we modify the MPO so that $\M{W}$ in Eq.~\eqref{Wtensor} is: $|0\rangle (I_t \otimes I_s) \langle 0|+ |1\rangle (Z_t \otimes Z_s) \langle 1|$. }

\subsection{{Algorithm for bosonizing an fPEPS}}

{The following gives an algorithm for bosonizing an fPEPS on a torus and writing it explicitly as a bPEPS.}

{\begin{enumerate}
   \item Given a triangulated 2D manifold with branching structure, determine the singular vertices. Singular vertices are those that are $1$-vertices of $4m$ number of triangles. Pair singular vertices along convenient paths. The edges along these paths are the elements of $\eta$.
   \item Construct the tensors $\M{M_f}=\M{F}\bdot \M{T}$ and $\M{\bar{M}_f}=\M{\bar F}\bdot \M{\bar T}$. Rearrange the virtual indices to match the order shown in Eq.~\eqref{ftinorder}. Remove the grading of the virtual indices of $\M{M_f}$, $\M{\bar{M}_f}$, and $\Bt_\eta$. The resulting tensors are $\M{M_b}$, $\M{\bar{M}_b}$, and $\M{B_b}$, respectively. 
   \item Choose convenient generators of the non-contractible loops parallel to the $x$-axis and $y$-axis, say $L^x $ and $L^y$, and calculate $(-1)^{\alpha_x}= \hatsigma(L^x)/\sigma_\eta(L^x)$ and $(-1)^{\alpha_y}= \hatsigma(L^y)/\sigma_\eta(L^y)$. Determine the tensor $\M{G}$ using Eq.~\eqref{gab2}. 
   \item Insert factors of $Z_t \otimes Z_s$ on virtual legs corresponding to the edges in $\eta$. Insert tensors $\M{W}$ along convenient generators of non-contractible loops parallel to the $x$- and $y$-axes, and glue the MPOs at their intersection with the $\M{G}$ tensor calculated in the previous step.  Finally, contract $\M{M_b}$, $\M{\bar{M}_b}$, and $\M{B_b}$ with the inserted factors of $Z_t \otimes Z_s$ and the MPO generated by $\M{W}$ and $\M{G}$.  
\end{enumerate}}

\subsection{{Example of bosonizing an fPEPS}} \label{examplebosonizingfpeps}

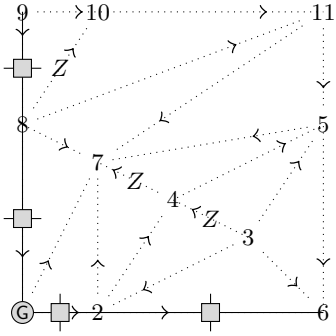
\begin{figure}
     \centering
    \begin{tikzpicture}[Sbase,scale=1]
    \randt;
  \path[](3)--node[]{$Z$}(4)--node[]{$Z$}(7);
  \path[](8)--node[]{$Z$}(10);
   \draw[](0,4)--(0,0)--(4,0);
\draw[] (-0.25,3.25)--node[Sq={0},scale=0.7]{}++(0.5,0);\draw[] (-0.25,1.25)--node[Sq={0},scale=0.7]{}++(0.5,0);
\draw[] (2.5,-0.25)--node[Sq={0},scale=0.7]{}++(0,.5);\draw[] (.5,-0.25)--node[Sq={0},scale=0.7]{}++(0,.5);
\node[Bcir,scale=1] at (0,0){$\M{G}$};
 \end{tikzpicture}
     \caption{  Choice of spin structure $\eta = \lbrace \langle 3,4\rangle,\langle 4,7 \rangle,\langle 8,10 \rangle  \rbrace$ and placement of the MPO generated by $\M{W}$ and $\M{G}$. The $Z^{\eta_e}$ operators shown represent the $(Z_t \otimes Z_s)^{\eta_e}$-operators that are inserted in the example of section \ref{examplebosonizingfpeps}.}
     \label{fig:etatorus}
 \end{figure}
 
{As an example, we bosonize the atomic insulator state $|\psi_{AI} )$ on a 2D torus, triangulated as shown in Fig.~\ref{fig:etatorus}. 
The tensor network representation has a vacuum tensor everywhere: $T_{000}=1, \bar{T}_{000}=1$ and all other components zero.} 

\vspace{2mm}

\noindent {\textbf{Step 1:} The singular vertices are $\langle3\rangle$, $\langle7\rangle$, $\langle8\rangle$ and $\langle10\rangle$. We pair them along the paths shown in Fig.~\ref{fig:etatorus}, so the spin structure $\eta$ is $\eta = \lbrace \langle 3,4\rangle,\langle 4,7 \rangle,\langle 8,10 \rangle  \rbrace$.}

\vspace{2mm}

\noindent {\textbf{Step 2:} We compute the tensors $\M{M_f}$, $\M{\bar{M}_f}$, and $\Bt_\eta$, order the legs according to Eq.~\eqref{ftinorder}, and replace the fermionic legs with bosonic legs to obtain: 
\begin{align}
 \M{M_b}=&\sum_{a,b,c}\delta_{a+b+c,0} 
     |c\rangle_{f_{01}}|0\rangle_{f'_{01}} |a\rangle_{f_{12}} |0\rangle_{f'_{12}} \langle 0|_{f'_{02}} \langle b|_{f_{02}} \\
  \M{\bar{M}_b} =& \sum_{a,b,c} \delta_{a+b+c,0} 
    |b\rangle_{f_{02}} |0\rangle_{f'_{02}}  \langle 0|_{f'_{12}} \langle a|_{f_{12}} \langle 0|_{f'_{01}} \langle c|_{f_{01}}  \\ \label{Bbdef}
    \M{B_b}=&\sum_{a} Z_e^{\eta_e} |a\rangle_{e} |a\rangle_{e} \langle a|_{e}.
\end{align}}

\vspace{2mm}

\noindent {\textbf{Step 3:} We use $L^x$ in the upper right corner of Fig.~\ref{fig:homoexample} to calculate:
\begin{align}
    (-1)^{\alpha_x} = {\hatsigma(L^x)}/{\sigma_\eta(L^x)}=1,
\end{align}
and we use $L^y$ in the lower left corner of Fig.~\ref{fig:homoexample} to calculate:
\begin{align}
     (-1)^{\alpha_y} = {\hatsigma(L^y)}/{\sigma_\eta(L^y)}=1.
\end{align}
$\alpha_x=\alpha_y=0$, so $\M{G}$  has components $ G_{ab}=(-1)^{ab}$.}

\vspace{2mm}

\noindent {\textbf{Step 4:} We insert the operator $(Z_t \otimes Z_s)^{\eta_e}$ at each edge and insert the MPO generated by $\M{W}$ and $\M{G}$ as in Fig.~\ref{fig:etatorus}. The state layer has only $0$ indices, so the factors of $Z_s$ do not affect the state. We freely remove all of the state layer indices. The factors of $Z_t^{\eta_e}$ cancel with the factor of $Z_e^{\eta_e}$ in the definition of $\M{B_b}$ in Eq.~\ref{Bbdef}.}

\vspace{2mm}

{The result is a bPEPS generated by the tensors:
\begin{align}
\M{M_b} =& \sum_{a,b}
 |a+b\rangle_{f_{01}}|a\rangle_{f_{12}} \langle b|_{f_{02}}\\
\M{\bar{M}_b} =& \sum_{a,b} |b\rangle_{f_{02}} \langle a|_{f_{12}}\langle a+b|_{f_{01}} \\
\M{B_b} =& \sum_{a} |a\rangle_{e} |a\rangle_{e} \langle a|_{e},
\end{align}
with the MPO generated by $\M{W}$ and $\M{G}$ inserted along the $x$- and $y$-axes (Fig.~\ref{fig:etatorus}). The bPEPS is in a ground state of Kitaev's toric code Hamiltonian. One way to see this is to notice that $\M{M_b}$ and $\M{\bar{M}_b}$ have the MPO symmetries:
\begin{align}
    Z_{f_{01}}Z_{f_{12}}Z_{f_{02}}\bdot \M{M_b} &= \M{M_b} \nonumber \\
    Z_{f_{01}}Z_{f_{12}}Z_{f_{02}}\bdot \M{\bar{M}_b} &= \M{\bar{M}_b},
\end{align}
indicative of the toric code phase. Moreover, the symmetry implies that $\M{M_b}$ and $\M{\bar{M}_b}$ have an even number of down spins. Since $\M{B_b}$ copies the virtual legs to the physical leg, the ground state is a superposition of all loops (on the dual lattice) of down spins. The tensor $\M{G}$ dictates the particular toric code ground state. For the branching structure in Fig.~\ref{fig:etatorus}, the  ground state is an equal amplitude superposition of all loops (on the dual lattice) of down spins acted on by the following operator:
\begin{align}
    \frac{1}{2}\left( 1-\prod_e Z_e^{\eta^x} - \prod_e Z_e^{\eta^y} + \prod_e Z^{\eta^{xy}} \right).
\end{align}}
{While the example of an atomic insulator state is rather simple, we expect the algorithm to extend naturally to more complicated problems.}

\section{Conclusion and future work}

Tensor networks provide a powerful framework for studying quantum many-body systems.  Their simple parameterization allows for efficient numerical computations, and their diagrammatic representation elucidates the structure of entanglement in quantum states. More abstractly, tensor networks provide a uniform language for discussing quantum many-body systems. Here, we have extended the formalism of tensor networks to exact bosonization dualities. In particular, we have constructed a TNO that implements the two dimensional bosonization duality first discussed in Ref. [\onlinecite{ChenYA17}]. Furthermore, our bosonization TNO can be applied directly to fermionic tensor network states, thus defining bosonization at the level of quantum states. 

A critical step of our bosonization procedure is to express the bosonized state as an explicit bosonic tensor network. To this end, we described how to account for Koszul signs accrued in contracting fermionic tensor networks, and we constructed matrix product operators to be placed along non-trivial cycles for this purpose. As a result, our bosonization duality at the level of states can be applied to fermionic systems on arbitrary triangulations of 2D manifolds without a boundary. 

{We would also like to emphasize that the calculation of Koszul signs in section \ref{sec: bosonization of fPEPS} has potential for applications outside of the bosonization of fPEPS. In fact, the calculation applies to the contraction of \textit{any}\footnote{Assuming the 2D fPEPS is defined on a triangulation of an orientable 2D manifold.} 2D fPEPS generated by fermion parity even local tensors and without fermionic physical legs. In particular, it may be useful for efficiently evaluating the overlap between two fPEPS. Explicitly, one can use the technology developed in section \ref{sec: bosonization of fPEPS} to replace the fermionic virtual legs with bosonic virtual legs and account for the Koszul signs. Notably, for a regular triangular lattice or square lattice, our results show that the fermionic virtual legs can freely be replaced with bosonic virtual legs as long as the MPO generated by $\M{W}$ and $\M{G}$ is inserted before closing the tensor network on a manifold with non-trivial genus.}

Directions for future work include generalizing our bosonization duality, identifying tensor network representations of wider classes of dualities, and utilizing the bosonization TNO to study fermionic systems.  A natural generalization is to develop a 3D bosonization duality at the level of quantum states. Recently, [\onlinecite{Chen18}] presented a bosonization duality in 3D, and we expect that this duality admits a tensor network representation. Formulating a TNO for the 3D duality might also make it clear how to bosonize in dimensions greater than three. Another possible generalization is to extend our bosonization duality to manifolds with boundaries. 

It would be interesting to construct tensor network representations for other operator-level dualities. Ref. [\onlinecite{Djordje18}] describes a duality for parafermionization in 2D, in which a system of constrained spins is dual to a system of parafermions.  Formulating a corresponding TNO would require a careful understanding of paraspin-structure -- a generalization of spin-structure to parafermions. We also expect that recently developed dualities for gauging subsytem symmetries can be naturally formulated in terms of tensor networks [\onlinecite{Shirley19}]. Further, it would be nice to interpret the results of Ref. [\onlinecite{Kubica18}] using tensor networks.

We anticipate that our bosonization procedure will be useful for studying fermionic topological orders. Beginning with a fermionic tensor network state, one can apply the bosonization procedure outlined in the text and subsequently analyze the topological order of the bosonic state using the myriad of techniques developed to study bosonic topological orders.  In addition, MPO symmetries of the fermionic state can be tracked through the bosonization procedure to obtain the MPO symmetries of the bosonic system. 

Going the other direction, the Hermitian conjugate of the bosonization TNO can be applied to a bosonic state to obtain a fermionic tensor network state. Two dimensional (non-chiral) bosonic topological orders have been well studied using tensor networks, so we can use the known tensor network representations of fixed point states to construct fixed point states for fermionic topological phases. Furthermore, the MPO symmetries of the bosonic system descend to MPO symmetries of the fermionic system. While fixed point states for intrinsic fermionic topological orders and fermionic symmetry protected topological phases were identified in Refs. [\onlinecite{Bultinck17},\onlinecite{Bultinck17a}], our bosonization procedure gives a means for constructing and studying fixed point states for fermionic symmetry-enriched topological orders as well. 

\vspace{0.1in}
\noindent{\it Acknowledgements -- } SS would like to acknowledge Alex Turzillo for valuable discussions about related work on $\Z_2$-graded tensors. SS and TE also thank Dave Aasen for helpful conversations about fermion condensation. LF is supported by NSF DMR-1519579.

\appendix

\section{\texorpdfstring{$\mathbb{Z}_2$}{Z2}-graded tensor representation of Majorana operators}
\label{app: Majorana tensors}

In this appendix, we show that the tensors introduced in section \ref{sec:gradedTN} and rewritten here:
\begin{align} 
    \begin{tikzpicture}[Sbase]
        \draw[red,OES={AR2=0.5}] (-1,0)node[below,black]{$e$}--(0,0)--(1,0)node[below,black]{$e$}; \node[OP] {$\gamma$};
    \end{tikzpicture} &=  \sum_{a}|a+1)_e(a|_e  \\ 
    \begin{tikzpicture}[Sbase]
        \draw[red,OES={AR2=0.5}] (-1,0)node[below,black]{$e$}--(0,0)--(1,0)node[below,black]{$e$}; \node[OP] {$\gammabar$};
    \end{tikzpicture} 
    &= \sum_{a} (-1)^{a} i |a+1)_e(a|_e,
\end{align}
are indeed good representations of Majorana operators. To do so, we explicitly show that the algebraic relations of the Majorana tensors match those of the Majorana operators introduced in section \ref{subsec: representing fermions}.

We begin by analyzing the algebra at a single site $e$. To this end, we apply the Majorana tensors to an arbitrary state $\M{A}$ at site $e$:
\begin{align}
    \begin{tikzpicture}[Sbase]
      \draw[red,-<-=0.5](0,0)node[below,black]{$e$}--++(1,0); \node[Tsq] at (1,0) {$\M{A}$};
    \end{tikzpicture}  \equiv   \sum_{a} A_{a} |a)_{e}\nonumber.
\end{align}
According to section \ref{subsec: representing fermions}, at site $e$, $\gamma_e^2=\gammabar_e^2=1$. Applying a single $\gamma$ tensor to $\M{A}$, we find: 
\begin{align}
\begin{tikzpicture}[Sbase]
      \draw[red,-<-=0.5](0,0)node[below,black]{$ $}--++(1,0); \node[Tsq] at (1,0) {$\M{A}$};
      \draw[red,-<-=0.5](-1,0)node[below,black]{$e$}--++(1,0);
      \node[OP] at (0,0) {$\gamma$};
    \end{tikzpicture} 
    &\equiv \sum_{b} |b+1)_e(b|^\ct_e \sum_{a} A_{a} |a)^\ct_{e} \\ \nonumber
    &= \sum_{a} A_{a} |a+1)_{e}.
\end{align}
Then, acting with another $\gamma$ on $\gamma \M{A}$ gives: 
\begin{align}
\begin{tikzpicture}[Sbase]
     \draw[red,OES={AR2=0.5}](-1,0)--++(.5,0)--++(1,0)--++(0.5,0); \node[Tsq] at (1,0) {$\M{A}$};
      \draw[red,-<-=0.5](-1,0)node[below,black]{$e$}--++(1,0);
      \node[OP] at (-0.5,0) {$\gamma$}; \node[OP] at (0.25,0) {$\gamma$};
      \node[red,scale=0.7] at (-0.5,-0.4) {2};
      \node[red,scale=0.7] at (0.25,-0.4) {1};
    \end{tikzpicture} 
     &\equiv \sum_{b} |b+1)_e(b|^\ct_e \sum_{a} A_{a} |a+1)^\ct_{e} \\ \nonumber
     &= \sum_{a} A_a |a)_e.
\end{align}
Since $\M{A}$ was arbitrary, we see that $\gamma$ contracted in succession with another $\gamma$ acts as the identity. The relation $\gammabar_e^2=1$, can be shown similarly.

Next, we show that the relation $\gammabar_e \gamma_e=-\gamma_e \gammabar_e$ is represented by the Majorana tensors. $\gammabar \gamma \M{A}$ is: 
\begin{align} \label{gammagammabarstats1}
\begin{tikzpicture}[Sbase]
     \draw[red,OES={AR2=0.5}](-1,0)--++(.5,0)--++(1,0)--++(0.5,0); \node[Tsq] at (1,0) {$\M{A}$};
      \draw[red,-<-=0.5](-1,0)node[below,black]{$e$}--++(1,0);
      \node[OP] at (-0.5,0) {$\gammabar$}; \node[OP] at (0.25,0) {$\gamma$};
        \node[red,scale=0.7] at (-0.5,-0.4) {2};
      \node[red,scale=0.7] at (0.25,-0.4) {1};
    \end{tikzpicture}
    \equiv & \sum_{c} (-1)^c i|c+1)_e(c|^{\ct_2}_e \nonumber\\
    &\sum_{b} |b+1)^{\ct_2}_e(b|^{\ct_1}_e \sum_{a} A_a |a)^{\ct_1}_e\nonumber \\ 
    = & - \sum_{a} (-1)^{a} i |a)_e,
\end{align}
while $\gamma \gammabar \M{A}$ is:
\begin{align} \label{gammagammabarstats2}
\begin{tikzpicture}[Sbase]
      \draw[red,OES={AR2=0.5}](-1,0)--++(.5,0)--++(1,0)--++(0.5,0); \node[Tsq] at (1,0) {$\M{A}$};
      \draw[red,-<-=0.5](-1,0)node[below,black]{$e$}--++(1,0);
      \node[OP] at (-0.5,0) {$\gamma$}; \node[OP] at (0.25,0) {$\gammabar$};
        \node[red,scale=0.7] at (-0.5,-0.4) {2};
      \node[red,scale=0.7] at (0.25,-0.4) {1};
    \end{tikzpicture}
    \equiv & \sum_{b} |b+1)_e(b|^{\ct_2}_e   \nonumber\\
    &\sum_{c} (-1)^c i|c+1)^{\ct_2}_e(c|^{\ct_1}_e\sum_{a} A_a |a)^{\ct_1}_e \nonumber\\ 
    =& \sum_{a}(-1)^{a} i |a)_e.
\end{align}
Comparing \eqref{gammagammabarstats1} and \eqref{gammagammabarstats2}, we see that the tensors $\gamma$ and $\gammabar$ capture the relation $\gamma_e \gammabar_e=-\gammabar_e \gamma_e$. It is important to note that in going from \eqref{gammagammabarstats1} to \eqref{gammagammabarstats2}, the contractions are different. The difference in sign is not simply due to the odd grading of $\gamma$ and $\gammabar$.

Now, we consider the algebraic relations of Majorana operators at different sites. Majorana operators acting at different sites anti-commute, so we must show that switching the order of contraction, for Majorana tensors applied to different legs yields a sign. This property follows from the odd grading of the Majorana tensors. We write an arbitrary state $|\psi)$ with $N$ two dimensional fermionic site Hilbert spaces as:
\begin{align}
  &\begin{tikzpicture}[Sbase]
    \foreach \j in {1,2,4,5}
    { \draw[red,->-=0.75](\j,0.5)--++(0,1);  }
    \draw[fill=white,draw=black] (0,0) rectangle (6,0.5);
    \draw[dotted,red]  (2.75,1)--+(0.5,0); 
    \node at (3,0.25) {$|\psi)$};
    \end{tikzpicture} \nonumber \\
  & =\sum_{a_1,\ldots,a_N} \Psi_{a_1,\ldots,a_N} |a_1)_{e^1}\ldots|a_N)_{e^N}.
\end{align}
First acting with $\gamma$ at site $e^j$ and second acting with $\gamma$ at site $e^i$, we have:
\begin{widetext}
\begin{align} \label{gammaigammajpsi}
 &\begin{tikzpicture}[Sbase]
    \foreach \j in {1,2,4,5}
    { \draw[red,->-=0.75](\j,0.5)--++(0,1);  }
    \draw[fill=white,draw=black] (0,0) rectangle (6,0.5);
    \draw[dotted,red]  (2.75,1)--+(0.5,0); 
    \node at (3,0.25) {$|\psi)$};
    \node[OP] at (2,1) {$\gamma_i$};
    \node[red,scale=0.7] at (2.5,1) {$2$};
     \node[OP] at (4,1) {$\gamma_j$};
     \node[red,scale=0.7] at (4.5,1) {$1$};
    \end{tikzpicture} \nonumber\\
  &\equiv\left( \sum_{c} |c+1)_{e^i}(c|^{\ct_2}_{e^i}\right) \left(\sum_{b} |b+1)_{e^j}(b|^{\ct_1}_{e^j}\right) 
  \sum_{a_1,\ldots,a_N} \Psi_{a_1,\ldots,a_N} |a_1)_{e^1}\ldots|a_i)^{\ct_1}_{e^i}\ldots|a_j)^{\ct_2}_{e^j}\ldots|a_N)_{e^N} \\ \label{psiprimedef}
  &=\sum_{a_1,\ldots,a_N} \Psi'_{a_1,\ldots,a_N} |a_1)_{e^1}\ldots\left(\sum_{c} |c+1)_{e^i}(c|^{\ct_2}_{e^i}\right)|a_i)^{\ct_2}_{e^i}\ldots \left(\sum_{b} |b+1)_{e^j}(b|^{\ct_1}_{e^j}\right)|a_j)^{\ct_1}_{e^j}\ldots|a_N)_{e^N} \\ \label{switchcontraction}
    &=\sum_{a_1,\ldots,a_N} \Psi'_{a_1,\ldots,a_N} |a_1)_{e^1}\ldots\left(\sum_{c} |c+1)_{e^i}(c|^{\ct_1}_{e^i}\right)|a_i)^{\ct_1}_{e^i}\ldots \left(\sum_{b} |b+1)_{e^j}(b|^{\ct_2}_{e^j}\right)|a_j)^{\ct_2}_{e^j}\ldots|a_N)_{e^N}\\ \label{movemajoranaleft}
    &=\left(\sum_{c} |c+1)_{e^i}(c|^{\ct_1}_{e^i}\right) \left(\sum_{b} |b+1)_{e^j}(b|^{\ct_2}_{e^j}\right) \sum_{a_1,\ldots,a_N} \Psi_{a_1,\ldots,a_N} |a_1)_{e^1}\ldots|a_i)^{\ct_1}_{e^i}\ldots |a_j)^{\ct_2}_{e^j}\ldots|a_N)_{e^N} \\ \label{Majoranainterchange}
    &=- \left(\sum_{b} |b+1)_{e^j}(b|^{\ct_2}_{e^j}\right) \left(\sum_{c} |c+1)_{e^i}(c|^{\ct_1}_{e^i}\right) \sum_{a_1,\ldots,a_N} \Psi_{a_1,\ldots,a_N} |a_1)_{e^1}\ldots|a_i)^{\ct_1}_{e^i}\ldots |a_j)^{\ct_2}_{e^j}\ldots|a_N)_{e^N}\\
    &=-  \begin{tikzpicture}[Sbase]
    \foreach \j in {1,2,4,5}
    { \draw[red,->-=0.75](\j,0.5)--++(0,1);  }
    \draw[fill=white,draw=black] (0,0) rectangle (6,0.5);
    \draw[dotted,red]  (2.75,1)--+(0.5,0); 
    \node at (3,0.25) {$|\psi)$};
    \node[OP] at (2,1) {$\gamma_i$};
    \node[red,scale=0.7] at (2.5,1) {$1$};
     \node[OP] at (4,1) {$\gamma_j$};
     \node[red,scale=0.7] at (4.5,1) {$2$};
    \end{tikzpicture}
\end{align}
\end{widetext}
In \eqref{psiprimedef}, we have absorbed the signs from moving the Majorana tensors past odd vectors into the coefficient $\Psi'$. After moving the Majorana tensors, the ordering of the contractions are switched [line \eqref{switchcontraction}]. Lastly, we have moved the Majorana tensors to the left and interchanged their order [\eqref{movemajoranaleft} and \eqref{Majoranainterchange}].  The contraction $\ct_1$ is then to the right of $\ct_2$, and we can read line \eqref{Majoranainterchange} as first a $\gamma$ acts on site $e^i$ then a $\gamma$ acts on site $e^j$. We thus have that $\gamma$ tensors acting on different sites anti-commute. Looking at \eqref{gammaigammajpsi} and \eqref{Majoranainterchange}, we see that the difference in sign is purely a consequence of the odd parity of $\gamma$. Indeed, more generally, $\gamma$ tensors acting on different legs of an arbitrary tensor will anti-commute. 
An analogous calculation for $\gammabar$ tensors or a mixture of $\gamma$ and $\gammabar$ tensors shows that they anti-commute when acting on different legs.  

\section{Calculation of the Koszul sign for a single loop}
\label{app:proofofprop1}

Here, we provide a proof of Prop.~\ref{prop:sigmaC}. We choose an arbitrary edge of the loop $L$ to be $e^0$ and label the $j^{\text{th}}$ edge of the path as $e^j$. Starting with the triangle following $e^0$, along the orientation of $L$, we denote the $j^\text{th}$ triangle on the path as $f^j$. For each triangle $f^j$, we have a specific basis tensor $\M{Q}^L_{f^j}$ from the set $\mathcal{Q}[f]$ or $\mathcal{\bar Q}[f]$ [see Eq.~\eqref{basisset}] depending on the orientation of $f^j$. The sign to be calculated is then:
\begin{align} \label{Tf0Tfn}
    \hatsigma(L)= \text{tr}\left[\M{Q}_{f^0}^L\bdot\M{Q}_{f^1}^L\bdot \ldots \bdot \M{Q}_{f^n}^L\right].
\end{align}
As already mentioned, there are six possible ways for the loop to cross a triangle. We list the six possible crossings of a positive triangle and its associated basis tensors $\M{Q}_{f}^L$ (ignoring legs with even parity): 
\begin{align}
    \begin{tikzpicture}[Sc=0.6]
     \post;
     \node at (-60:1) {$e^j$}; \path  (0:3)++(-120:1) node{$e^{j+1}$};
\draw[-<-=0.75,blue](C)--+(210:{sqrt(3)});
\draw[->-=0.75,blue](C)--+(-30:{sqrt(3)});
    \end{tikzpicture}\equiv &|1)_{e^{j}}|1)_{e^{j+1}}=-i|1)_{e^{j}}\big[i|1)_{e^{j+1}}\big] ,\nonumber\\  \begin{tikzpicture}[Sc=0.6]
  \post;  \node at (-60:1) {$e^{j+1}$}; \path  (0:3)++(-120:1) node{$e^j$};
\draw[->-=0.75,blue](C)--+(210:{sqrt(3)});
\draw[-<-=0.75,blue](C)--+(-30:{sqrt(3)});
    \end{tikzpicture}\equiv& -|1)_{e^{j}}|1)_{e^{j+1}}=i|1)_{e^{j}}\big[i|1)_{e^{j+1}}\big]\nonumber\\
\begin{tikzpicture}[Sc=0.6]
  \post;\node at (0:1) {$e^j$}; \path  (0:3)++(-120:1) node{$e^{j+1}$};
\draw[-<-=0.75,blue](C)--+(90:{sqrt(3)});
\draw[->-=0.75,blue](C)--+(-30:{sqrt(3)});
    \end{tikzpicture}\equiv& |1)_{e^{j+1}}(1|_{e^{j}}={\big[}i(1|_{e^{j}}{\big]}{\big[}i|1)_{e^{j+1}}{\big]},\nonumber\\ \begin{tikzpicture}[Sc=0.6]
  \post;\node at (0:1) {$e^{j+1}$}; \path  (0:3)++(-120:1) node{$e^j$};
\draw[->-=0.75,blue](C)--+(90:{sqrt(3)});
\draw[-<-=0.75,blue](C)--+(-30:{sqrt(3)});
    \end{tikzpicture}\equiv& |1)_{e^{j}}(1|_{e^{j+1}}\nonumber\\
\begin{tikzpicture}[Sc=0.6]
  \post;\node at (0:1) {$e^j$}; \path  (-60:1) node{$e^{j+1}$};
\draw[-<-=0.75,blue](C)--+(90:{sqrt(3)});
\draw[->-=0.75,blue](C)--+(210:{sqrt(3)});
    \end{tikzpicture}\equiv& |1)_{e^{j+1}}(1|_{e^{j}}=\big[i(1|_{e^{j}}\big]\big[i|1)_{e^{j+1}}\big] \nonumber\\ \begin{tikzpicture}[Sc=0.6]
  \post;\node at (0:1) {$e^{j+1}$}; \path  (-60:1) node{$e^j$};
\draw[->-=0.75,blue](C)--+(90:{sqrt(3)});
\draw[-<-=0.75,blue](C)--+(210:{sqrt(3)});
    \end{tikzpicture}\equiv &|1)_{e^{j}}(1|_{e^{j+1}}.
\end{align}
The blue arrows denote the loop $L$, which enters at edge $e^j$ and exists from edge $e^{j+1}$. Notice that when the loop goes around a 0-vertex or a 2-vertex (bottom four pictures), both edges point to the same side of $L$, but when the loop goes around a 1-vertex (top two pictures), a right-left transition of edge directions occurs.  The relation between the diagrams and the tensors can be summarized as follows:
\begin{enumerate}[label=(\roman*)]
    \item Edges $e^j$ pointing to the \textit{right} of $L$ contribute $(1|_{e^j}$ to the tensor $\M{Q}^L_{f^{j-1}}$ and  $|1)_{e^j}$ to the tensor $\M{Q}^L_{f^{j}}$. 
    \item Edges $e^j$ pointing to the \textit{left} of $L$ contribute $i|1)_{e^j}$  to the tensor $\M{Q}^L_{f^{j-1}}$ and $i(1|_{e^j}$  to the tensor  $\M{Q}^L_{f^{j}}$. 
    \item If $f_L$ is an $1$-vertex, then we accrue an additional phase $i^{\delta_{f_L \in \bar{l}_L}}i^{-\delta_{f_L \in \bar{r}_L}}$, where $\delta_{f_L \in \bar{l}_L}=1$ if $f_L \in \bar{l}_L$ and $\delta_{f_L \in \bar{l}_L}=0$ otherwise. $\delta_{f_L \in \bar{r}_L}$ is defined similarly.  Therefore, if $f_L$ is a 1-vertex, we accrue a phase $i$, if it lies to the left of $L$ or a phase $-i$, if it lies to the right of $L$.
\end{enumerate}

Negatively oriented triangles also have 6 possible crossings. It can be checked that the same rules as in (i)-(iii) above apply to negative triangles. For example, consider the following crossing on a negative triangle:
\begin{align}\label{negTex}
    \begin{tikzpicture}[Sc=0.8]
     \negt;
     \node at (60:1) {$e^j$}; \path  (0:3)++(120:1) node{$e^{j+1}$};
\draw[-<-=0.75,blue](C)--+(-210:{sqrt(3)});
\draw[->-=0.75,blue](C)--+(30:{sqrt(3)});
    \end{tikzpicture}&\equiv (1|_{e^{j+1}}(1|_{e^{j}},
\end{align}
where the RHS is an element of \eqref{basisset} (ignoring even parity legs). Now, we verify that the rules (i)-(iii) yield the RHS of Eq.~\eqref{negTex}. Rule (ii) implies $e^j$ contributes $i(1|_{e^j}$, rule (i) implies edge $e^{j+1}$ contributes $(1|_{e^{j+1}}$, and finally, rule (iii) implies that the $f_L$ vertex contributes an $i$ phase. Putting it together, we get the tensor $ i(1|_{e^j} (1|_{e^{j+1}} i=(1|_{e^{j+1}}(1|_{e^{j}}$, which is indeed the RHS of Eq.~\eqref{negTex}.  The other five cases of crossing across negatively oriented triangles can be checked similarly. 

With this, we calculate the sign in Eq.~\eqref{Tf0Tfn}. We consider the contraction of tensors $ \M{Q}_{f^{j-1}}^L$ and $ \M{Q}_{f^{j}}^L$ at the edge $e^j$. If $e^j$ points to the the right of $L$, then, according to rule (i), $ \M{Q}_{f^{j-1}}^L$ has $(1|_{e^j}$ and $\M{Q}_{f^{j}}^L$ has  $|1)_{e^j}$. No Koszul sign is produced in contraction at $e^j$  because $(1|_{e^j}\bdot |1)_{e^j}=1$. Similarly, if $e^j$ points to the the left of $L$, then, according to rule (ii), $ \M{Q}_{f^{j-1}}^L$ has $i|1)_{e^j}$ and $\M{Q}_{f^{j}}^L$ has  $i(1|_{e^j}$, and again no Koszul sign is produced: $ i|1)_{e^j}\bdot i(1|_{e^j}=1$. The remaining sources of signs are triangles that contribute a sign $i^{\delta_{f_L \in \bar{l}_L}}i^{-\delta_{f_L \in \bar{r}_L}}$ according to rule (iii), and the overall $-1$ supertrace sign that comes from contracting the first and last indices in Eq.~\eqref{Tf0Tfn}. Therefore, the total sign is:
\begin{align}
    \hatsigma(L) = & - \prod_{f_L} i^{\delta_{f_L \in \bar{l}_L}}i^{-\delta_{f_L \in \bar{r}_L}} \nonumber\\
    =& -i^{\boldsymbol{(}n(\bar{l}_L)- n(\bar{r}_L)\boldsymbol{)}}=-(-1)^{\frac{1}{2}\boldsymbol{(}n(\bar{l}_L)- n(\bar{r}_L)\boldsymbol{)}}.
\end{align}
Note that $\hatsigma(C)$ is always $\pm 1$, because the total number of transition points $n(\bar{l}_L)+n(\bar{r}_L)$ has to be even. This implies $n(\bar{l}_L)- n(\bar{r}_L)$ is even as well. 

\section{Tensor Network Bosonization in 1D} \label{operatorduality1d}

For completeness, we give a detailed description of the TNO representation of bosonization in 1D. To start, we present 1D bosonization as a map of local fermionic operators to local bosonic operators.  


\subsection{Review of 1D bosonization} \label{reviewbosonization}
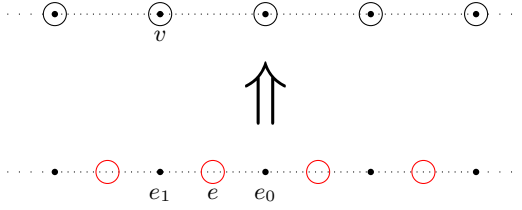
\begin{figure}
    \centering
\begin{tikzpicture}[scale=0.7]
 \draw[dotted](0,0)edge[loosely dotted]+(-1,0)--(8,0)edge[loosely dotted]+(1,0);
  \draw[dotted](0,-3)edge[loosely dotted]+(-1,0)--++(8,0)edge[loosely dotted]+(1,0);
 \foreach[count=\k from 1] \j in {0,2,...,8} {
 \node[circle,draw=black] at (\j,0) {};
 \ifthenelse{\j = 8}{}{\node[circle,draw=red] at (\j+1,-3) {}};
 \node[circle,inner sep=0.3mm,fill] at (\j,0) {};
 \node[circle,inner sep=0.3mm,fill] at (\j,-3) {};}
\node[below=1mm] at (2,0) {$ {v} $}; 
\node[below=3pt] at (3,-3) {$e$};
\node[below=3pt] at (2,-3) {$e_1$};\node[below=3pt] at (4,-3) {$e_0$};
 \node[scale=3] at (4,-1.5) {$ \Uparrow$};
\end{tikzpicture} 
\caption{The bosonization duality is a map from a fermionic system to a bosonic system. In the fermionic system there is a spinless complex fermion degree of freedom (red circles) at each edge $e$. In the bosonic system there is a spin-1/2 at each vertex $v$.}  
    \label{fig:1d chain}
\end{figure}
On the fermionic side of the duality, we consider a one dimensional lattice with a spinless complex fermion at each edge, as pictured in Fig.~\ref{fig:1d chain}.  The complex fermion at edge $e$ may be described using the familiar fermionic creation and annihilation operators: $c^{\dagger}_{e}$, $c_{e}$. These generate the full fermionic operator algebra at $e$. However, it will be convenient to instead work with Majorana operators, $\gamma_e, \gammabar_e$, as discussed in section \ref{subsec: representing fermions}. 

To ensure the bosonization duality maps local operators to local operators, we define the duality on a subset of the full fermionic operator algebra - the subalgebra of fermion parity even operators $\mathcal{E}$. The fermion parity even operators are those that commute with the global fermion parity operator $\prod_{e}P_{e}$, where $P_{e}$ is the fermion parity at the edge $e$. 
$\mathcal{E}$ can be generated by two types of operators: fermion parity $P_{e}$ at each edge and the hopping operators $S_v$ at each vertex $v$. The hopping operators transfer fermion parity between edges and are defined by:
\begin{align}
    S_v\equiv i \gamma_{{L_v}} \bar \gamma_{{R_v}},
\end{align}
with ${L_v}$ and ${R_v}$ the edge to the left and right of vertex $v$, respectively. The hopping operators are mutually commuting and commute with all parity operators besides the neighboring two, i.e.:
\begin{align}\label{1SPcom}
  S_v P_{e} =& (-1)^{\delta_{v \subset e}}  P_eS_v,
\end{align}
where $\delta_{v \subset e}=1$ if vertex $v$ is at one of the endpoints of the edge $e$ and $\delta_{v \subset e}=0$ otherwise.
With open boundary conditions, the set of fermion parity operators and hopping operators are independent.  However on a closed manifold they satisfy the relation:
\begin{align}
    \prod_v S_v \prod_e P_e = -1
\end{align}

On the bosonic side of the duality we have a spin-1/2 at each vertex (see Fig.~\ref{fig:1d chain}).  The operator algebra of the spin-1/2 at vertex $v$ can be generated by the Pauli operators: $X_v$, $Z_v$. Thus, the set of $X_v$ and $Z_v$ for all vertices generates the full bosonic operator algebra, which we denote as $\mathcal{A}$.

We now define the duality map $\mathfrak{D}: \mathcal{E} \to \mathcal{A}$ on the generators of $\mathcal{E}$:
\begin{align} \label{D1def}
         \mathfrak{D} (P_{e})= & Z_{e_0}Z_{e_1} \nonumber \\
         \mathfrak{D} (S_{v}) = & X_{v}.
\end{align}
where $e_0$ and $e_1$ denote the vertices at the endpoints of $e$ such that $e$ points from $e_0$ to $e_1$ (Fig.~\ref{fig:1d chain}). 
$\mathfrak{D}$ is an injective homomorphism from $\mathcal{E}$ to $\mathcal{A}$ so that for $A_1,A_2 \in \mathcal{E}$:
\begin{align}
    \mathfrak{D}(A_1+A_2) = & \mathfrak{D}(A_1) +\mathfrak{D}(A_2) \nonumber \\
        \mathfrak{D}(A_1 A_2) = & \mathfrak{D}(A_1) \mathfrak{D}(A_2). 
\end{align}
One can check that $\mathfrak{D}$ preserves the commutation relations in \eqref{1SPcom}.

{Note that the bosonization duality in Eq. \eqref{D1def} is not the usual Jordan-Wigner transformation, defined, for example, in Ref. [\onlinecite{Laumann09}]. 
$\mathfrak{D}$ is instead the composition of the familiar Jordan-Wigner transformation (restricted to $\mathcal{E}$) with the Kramers-Wannier duality.  We have chosen the duality $\mathfrak{D}$ to define bosonization, because it is locality preserving and more naturally relates to the $2$D bosonization in section \ref{2Dbosonization}.}

 
To translate the operator duality defined in \eqref{D1def} to a TNO, we employ the formalism of {$\Z_2$-graded Hilbert spaces} and {graded tensor products}.

\subsection{TNO representation of the duality}
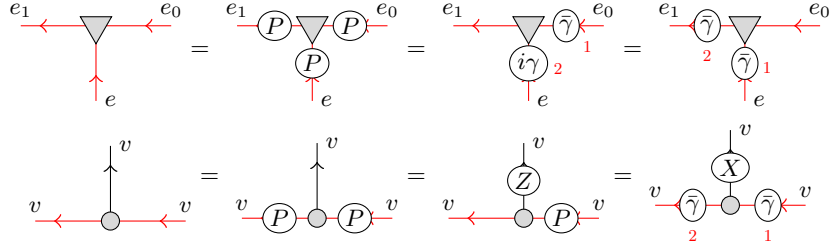
\begin{figure*}
    \centering
  \begin{tikzpicture}[Sbase]
  \Fer; 
  \end{tikzpicture}= \begin{tikzpicture}[Sbase]
  \Fer; \node[OP] at (1/2,0) {$P$}; \node[OP] at (3/2,0) {$P$}; \node[OP] at (1,-1/2) {$P$};
  \end{tikzpicture}= \begin{tikzpicture}[Sbase]
\Fer;\node[OP] (e0) at (3/2,0) {$\gammabar $}; \node[OP] (e) at (1,-1/2) {$i\gamma$};  \path (3/2,0)++(-45:0.4) node[red,scale=0.7]{1};\path (1,-1/2)++(-10:0.4) node[red,scale=0.7]{2};
  \end{tikzpicture}= \begin{tikzpicture}[Sbase]
  \Fer; \node[OP] (e1) at (1/2,0) {$ \gammabar$};\node[OP] (e) at (1,-1/2) {$\gammabar$};  \path (1/2,0)++(-90:0.4) node[red,scale=0.7]{2};\path (1,-1/2)++(-10:0.3) node[red,scale=0.7]{1};
  \end{tikzpicture} 
  
    \begin{tikzpicture}[Sbase]
  \Bos; 
\end{tikzpicture}= \begin{tikzpicture}[Sbase]]
  \Bos; \node[OP] at (1/2,0) {$P$}; \node[OP] at (3/2,0) {$P$};
  \end{tikzpicture}=
  \begin{tikzpicture}[Sbase]
  \Bos; \node[OP] at (3/2,0) {$P$}; \node[OP] at (1,1/2) {$Z$};
  \end{tikzpicture}=
  \begin{tikzpicture}[Sbase]
  \Bos; \node[OP] (e1) at (1/2,0) {$\gammabar $};\node[OP] (e0) at (3/2,0) {$\gammabar$}; \node[OP] at (1,1/2) {$X$};  \path (3/2,0)++(-90:0.4) node[red,scale=0.7]{1};\path (1/2,0)++(-90:0.4) node[red,scale=0.7]{2};
  \end{tikzpicture}
    \caption{Diagrammatic representation of the symmetries of $\Ft$ (first line) and $\Bt$ (second line) written algebraically in Eqs.~\eqref{Fsym} and \eqref{Bsym}.}
    \label{fig:1Dsym}
\end{figure*}
We now give a tensor network operator (TNO)  representation $\M{D}$ of the bosonization duality $\mathfrak{D}$ in Eq.~\eqref{D1def}. We begin with the following TNO ansatz:

\begin{align} \label{eq:UJod}
\M{D}=\begin{tikzpicture}[Sbase,xscale=0.9]
\draw[red,OES={AR2=0.55}] (0.5,0)--++(.5,0)--++(1.5,0)--++(1.5,0)--++(1.5,0)--++(1.5,0)edge[-<-=0.75]++(0.5,0);
\draw[red,loosely dotted] (0,0)--(0.5,0);
\draw[red,loosely dotted] (7.5,0)--(8,0);
\draw[red,-<-=0.75] (1,0)node[Dtriangle]{$\M{F}$}--+(0,-1); \draw[red,-<-=0.75] (4,0)node[Dtriangle]{$\M{F}$}node[black,above=2.5mm]{$e$}--+(0,-1);\draw[red,-<-=0.75] (7,0)node[Dtriangle]{$\M{F}$}node[black,above=1mm]{}--+(0,-1);
\draw[black,->-=0.75] (2.5,0)node[Bcir]{$\M{B}$}node[black,below=2.5mm]{$e_1$}--+(0,1);
\draw[black,->-=0.75] (5.5,0)node[Bcir]{$\M{B}$}node[black,below=2.5mm]{$e_0$}--+(0,1);
\end{tikzpicture}.
\end{align}
 For now, we leave the boundary conditions of $\M{D}$ unspecified -- they will enter the construction later. 

$\Dt$ is constructed by gluing together two kinds of local tensors, $\M{F}$ (triangular nodes) and $\M{B}$ (circular nodes), as pictured in \eqref{eq:UJod}.  An $\M{F}$ tensor is placed at each edge $e$ and is represented as follows:
\begin{align}\label{B2d}
    \begin{tikzpicture}[Sbase]
  \Fer;
\node[Dtriangle] at (1,0){$\M{F}$} ; 
\end{tikzpicture} \equiv \sum_{j,a,b} F^j_{a,b} |a)_{e_1} (j|_{e} (b|_{e_0}.
\end{align}
At each vertex $v$, we place a tensor $\M{B}$:
\begin{align} 
\begin{tikzpicture} [Sbase]
\Bos;
\node[Bcir] at (1,0){$\M{B}$} ; 
\end{tikzpicture} \equiv \sum_{j,a,b}  B^{j}_{a,b}|a)_v |j\rangle_v(b|_v.
\end{align}
Notice that in Eq.~\eqref{B2d}, we have three distinct Hilbert spaces labeled by the same site -- one fermionic space (to which $|a)_v$ belongs),  one dual fermionic space (to which $(a|_v$ belongs) and one bosonic space (to which $|j\rangle_v$ belongs). 

To implement the duality map $\mathfrak{D}$ of Eq.~\eqref{D1def}, we need to choose tensors $\Ft$ and $\Bt$ such that the following relations hold for all even operators $A \in \mathcal{E}$:
\begin{align} \label{dualitytno1}
    \M{D} \bdot  A = \mathfrak{D}(A) \bdot\M{D},
\end{align}
or diagrammatically:
\begin{align} \label{dualitytno2}
\begin{tikzpicture}[Sbase]
\draw[red] (0,0)--(3,0); \draw[red,loosely dotted] (-0.25,0)--(0,0);
\draw[red,loosely dotted] (3,0)--(3.25,0);
\foreach \j in {0.5,1,...,2.5}
{\draw[red] (\j,0)node[scale=0.4,Dtriangle]{}--+(0,-1); 
}
\foreach \j in {0.25,0.75,...,2.75}
{\draw[black!99] (\j,0)node[scale=0.4,Bcir]{}--+(0,1);}
\draw[draw=black,fill=white] (0.95,-0.75) rectangle node{A} ++(1.1,0.5);
\end{tikzpicture} = & \begin{tikzpicture}[Sbase]
\draw[red] (0,0)--(3,0);\draw[red,loosely dotted] (-0.25,0)--(0,0);
\draw[red,loosely dotted] (3,0)--(3.25,0);
\foreach \j in {0.5,1,...,2.5}
{\draw[red] (\j,0)node[scale=0.4,Dtriangle]{}--+(0,-1); 
}
\foreach \j in {0.25,0.75,...,2.75}
{\draw[black!99] (\j,0)node[scale=0.4,Bcir]{}--+(0,1); 
}
\draw[draw=black,fill=white] (0.65,0.25) rectangle node{$\mathfrak{D}(A)$} ++(1.75,0.5);
\end{tikzpicture}.
\end{align}
Note that, we need only show that the relations are satisfied for $P_e$ and $S_v$ -- the generators of $\mathcal{E}$. That is, we need to show that:
\begin{subequations}
\begin{align}
    \Dt  \bdot P_{e} = \mathfrak{D}(P_e) \bdot\Dt=Z_{e_0} Z_{e_1} \bdot \Dt \label{DP1d}\\
    \Dt \bdot S_{v} = \mathfrak{D}(S_v) \bdot\Dt=X_v  \bdot\Dt. \label{DS1d}
\end{align}
\end{subequations}

We can look at these constraints as symmetries of the tensor $\Dt$, which can be reduced to symmetries of local tensors $\Ft$ and $\Bt$. We claim that $\Dt$ satisfies \eqref{DP1d} and \eqref{DS1d} if $\Ft$ and $\Bt$ have the following symmetries: 
\begin{subequations} \label{FBsym}
\begin{align}
\Ft = P_{e_1} \bdot \Ft \bdot P_{e_0} \bdot P_{e} = \Ft \bdot i \gamma_{e} \bdot \gammabar_{e_0}   = \gammabar_{e_1}\bdot  \Ft\bdot \gammabar_e \label{Fsym}\\
    \Bt =  P_v \bdot \Bt \bdot P_{v} = Z_v\bdot \Bt \bdot P_{v}=  \gammabar_v\bdot X_v \bdot\Bt .
     \label{Bsym}
\end{align}
\end{subequations}
These symmetries are represented graphically in Fig.~\ref{fig:1Dsym}.

Using the diagrammatic representation of the symmetries, we can illustrate that $\Dt$ obeys \eqref{DP1d} and \eqref{DS1d}.  
By successive applications of the symmetries in Fig.~\ref{fig:1Dsym}, we have:
\begin{align} \label{P image}
\M{D} \bdot P_{e} =&\begin{tikzpicture}[Sbase,SS]
\DPone;\node[OP]  at (2,-.75) {$P_{e}$}; 
\end{tikzpicture} \nonumber\\
=&\begin{tikzpicture}[Sbase,SS]
\DPone;\node[OP] at (1.25,0) {$P$}; \node[OP]  at (2.75,0) {$P$}; 
\end{tikzpicture}\nonumber\\	
=&\begin{tikzpicture}[Sbase,SS]
\DPone;\node[OP]  at (0.5,0.75) {$Z_{e_1}$}; \node[OP]   at (3.5,0.75) {$Z_{e_0}$};
\end{tikzpicture} = Z_{e_1}Z_{e_0} \bdot \M{D}.
\end{align}
Similarly, for the hopping operator, we have:
\begin{align} \label{S image}
\M{D} \bdot S_{v}= &\M{D} \bdot i\gamma_{{L_v}} \gammabar_{{R_v}}  \\
=&\begin{tikzpicture}[Sbase,SS]
    \DS; \node[OP] at (0.5,-0.75) {$i\gamma_{{L_v}}$};\node[OP] at (3.5,-0.75) {$\gammabar_{{R_v}}$};
    \node[below=.9mm,red,scale=0.7] at (1,-0.75) {$2$};
    \node[below=.9mm,red,scale=0.7] at (4,-0.75) {$1$};
\end{tikzpicture} \nonumber \\
=& \begin{tikzpicture}[Sbase,SS]
    \DS; \node[OP] at (1.25,0) {$ \gammabar$};\node[OP] at (2.75,0) {$ \gammabar$};
    \node[below=.9mm,red,scale=0.7] at (2.95,0) {$1$};
    \node[below=.9mm,red,scale=0.7] at (1.45,0) {$2$};
\end{tikzpicture} \nonumber \\
=&\begin{tikzpicture}[Sbase,SS]
    \DS; \node[OP] at (2,.75) {$ X_v$};
\end{tikzpicture} \nonumber  = X_v \bdot\M{D}.
\end{align}
Hence, $\Dt$ is a good representation of the operator duality $\mathfrak{D}$.

Furthermore, we can use the symmetries of $\Ft$ and $\Bt$ to compute their explicit component form.  Notice that the three symmetries of $\Ft$ are independent, commute with each other, and square to the identity. Thus, they generate a $\Z^3_2 = \Z_2 \times \Z_2 \times \Z_2$ symmetry group. Similarly, the three symmetries of $\Bt$ form a $\Z^3_2$ group. Since both tensors are vectors in a $2^3=8$ dimensional Hilbert space, the symmetries fix the tensors completely (up to a normalization).  The explicit tensors can then be calculated by projecting the vaccuum tensor onto the symmetric subspace:
\begin{align}
    \Ft \propto & \sum_{a,b,c} (\gammabar_{e_1}\gammabar_{e})^a(i\gamma_e)^b P^c_{e_1} |0)_{e_1}(0|_e(0|_{e_0}P^c_{e}P^c_{e_0} \gamma_{e_0}^b \nonumber \\
    =&\sum_{a,b}  |a)_{e_1}(a+ b |_{e}(b|_{e_0}.
\end{align}
Applying the same strategy to compute $\Bt$, we find:
\begin{align}
     \Bt \propto & \sum_{a} |a)_{v}|a\rangle_{v}(a|_{v}.
\end{align}

Thus far, we have constructed a TNO that implements a map of local operators to local operators.  In the next subsection, we will illustrate one of the key advantages of the TNO representation of the bosonization duality.  That is, we will see that $\Dt$ may be applied to fermionic tensor network states to map them to bosonic tensor network states.

\section{Bosonization of fermionic matrix product states}
\label{subsec:TNbosonization1d}

We now show that certain fermionic matrix product states (fMPS) can be directly bosonized using the bosonization TNO, $\Dt$, defined in the previous subsection. In particular, we will describe the bosonization procedure for fMPS of the form: 
\begin{align} \label{Opsi}
 |\psi) = \begin{tikzpicture}[Sbase]
     \draw[dashdotted](0,-0.5)--+(0,1);\draw[dashdotted](6,-0.5)--+(0,1);
     \draw[red,OES={AR2=0.5}] (0,0)--++(0.5,0) --++(1,0)--++(1,0)--++(1,0); \draw[red,loosely dotted] (3,0) --(4,0); 
     \draw[red,OES={AR2=0.85}] (4,0) --++(1,0)--++(1,0); 
     \foreach \j in {0.5,1.5,2.5,4.5}
     {\draw[red,->-=0.5] (\j,0)--+(0,1); \node[Tsq] at (\j,0) {$\M{T}$};  } \node[OP] at (5.5,0) {$O_\psi$};
 \end{tikzpicture},
\end{align}
where $\M{T}$ is a fermion parity even tensor and $O_\psi$ is an operator with definite parity. $O_\psi$ is inserted before closing the fermionic matrix product state to dictate the parity of the state and the boundary conditions. We will use vertical dash-dotted lines to denote closing the boundary (or taking the trace, algebraically). Unless otherwise stated, we assume the Hilbert spaces are two dimensional.  Algebraically, $|\psi)$ can be written as:
\begin{align} \label{psidef}
|\psi) =\sum_{j_0,\ldots,j_N  } \text{tr}\left[ T^{j_0}T^{j_1}\ldots T^{j_N} O_\psi\right] |j_0)_{e^0}|j_1)_{e^1}  \ldots |j_N)_{e^N}&,
\end{align}
where $e^k$ denotes the edge connecting the $k-1$ vertex and the $k$ vertex.  The first step in bosonizing $|\psi)$ is to close $\Dt$ with an operator $O_\Dt$:
\begin{align} \label{step1bosonization}
 \begin{tikzpicture}[Sbase,scale=1.3]
 \def \h{0.75};
     \draw[red,OES={AR2=0.5}] (0,\h) --++(.5,0)--++(.5,0)--++(.5,0)--++(1,0); \draw[red,loosely dotted] (2,\h) --(3,\h); \draw[red,OES={AR2=0.5}] (3,\h) --(3.5,\h)--(4,\h)--(4.5,\h)--(5.2,\h);
  \draw[dashdotted](0,\h-0.25)--+(0,.5); \draw[dashdotted](5.2,\h-0.25)--+(0,.5);
     \foreach \j in {0.5,1.5,3.5}
     {\draw[red,->-=0.5] (\j,.1)--+(0,.75);  \draw[black,->-=0.5] (\j+.5,\h)--+(0,.75); \node[Dtriangle] at (\j,\h) { };\node[Bcir] at (\j+0.5,\h) { };  }   
         \node[OP] at (4.7,\h) {$O_\Dt$};
 \end{tikzpicture}.
\end{align}
 As we will show now, the choice of $O_\Dt$ determines both the subspace of the fermionic Hilbert space mapped non-trivially by the duality as well as the subspace of the bosonic Hilbert space in the image of the duality. 
\subsection{Boundary conditions in 1D}
\label{app: BC in 1D}

Here, we discuss how the choice of $\M{O}_\Dt$ affects the duality. With $\M{O}_\Dt$ parameterized as $\M{O}_\Dt=(-i\gammabar)^\alpha P^\beta$ and $\alpha , \beta \in \{0,1\}$, we will show that $\alpha$ determines the parity of the fermionic states that are mapped to non-trivial bosonic states and $\beta$ dictates the image of the duality. Specifically, for $\alpha=0(1)$, the subspace of states with even (odd) parity are mapped to the subspace of bosonic states invariant under the operator $(-1)^{\beta+1}\prod_v X_v$.

To begin, we note that the choice of $\M{O}_\Dt$ does not affect the duality away from the boundary.  Away from the boundary, the graphical calculation in \eqref{P image} and \eqref{S image} is unchanged by the choice of $\M{O}_\Dt$. For a chain with $N+1$ sites, we constrain $\M{O}_\Dt$ by considering the image of the fermion parity operator $P_{e^0}$ and the hopping operator $S_N$. (Recall that we have defined $e^k$ as the edge connecting vertices $k-1$ and $k$, so $e^0$ connects 0 and $N$.) We will also require that, similar to the case away from the boundary, $\Dt$ maps local operators near the boundary to local operators. 


Now, we consider acting on $P_{e^0}$ with $\Dt$.  The diagrammatic calculation yields:

\begin{align} \label{P boundary}
\mathsf{D}\bdot P =&\begin{tikzpicture}[Sbase,SS]
\DPOprime;\node[OP] (A) at (2,-.75) {$P_{e^0}$}; \node[OP] at (1.25,0) {$\M{O}_\Dt$};
\end{tikzpicture}\\ \nonumber
=&\begin{tikzpicture}[Sbase,scale=1.2, every node/.style={scale=0.8}]
\DPOprime;\node[OP] (A) at (1.25,0) {$P\M{O}_\Dt P$};  \node[OP] (A) at (0.5,0.75) {$Z_{N}$}; \node[OP] (B) at (3.5,0.75) {$Z_{0}$};
\end{tikzpicture} \\ \nonumber
=&Z_NZ_0\bdot \Dt'.
\end{align}
Note that the operator $Z_N$ is required to ensure that the commutation relations between $P_{e^0}$ and $S_N$ are preserved by the duality. In the last line of \eqref{P boundary}, $\Dt'$ is the same as $\Dt$ but with $\M{O}_\Dt$ replaced by $P\M{O}_\Dt P$.  The bosonization TNO should be left unmodified, so we require that $\Dt'\propto \Dt$. This means that $P\M{O}_\Dt P=c\M{O}_\Dt$ for some $c\in \mathbb{C}$, and we have:
\begin{align}
    \Dt=\Dt \bdot P^2_{e^0} = c^2 \Dt.
\end{align} 
Therefore, $c$ must be $\pm 1$, or $P\M{O}_\Dt P=\pm \M{O}_\Dt$. We then see that $\M{O}_\Dt$ must have definite fermion parity, so it can be parameterized as $\M{O}_\Dt=(-i\gammabar)^\alpha P^\beta$ with $\alpha, \beta \in \{0,1\}$. 

Next, we act on the hopping operator $S_N$ with $\Dt$:

\begin{align} \label{S boundary}
\mathsf{D}\bdot  S_{N}= &\mathsf{D}\bdot i\gamma_{e^N} \gammabar_{e^0} \nonumber \\
=&(-1)^{|O_\Dt|}\begin{tikzpicture}[Sbase,SS]
    \DSOprime; \node[OP] at (0.5,-0.75) {$i\gamma_{e^N}$};\node[OP] at (3.5,-0.75) {$\gammabar_{e^0}$};
    \node[red,scale=0.8] at (1.15,-.8) {2}; \node[red,scale=0.8] at (4,-.8) {1};
\end{tikzpicture} \nonumber \\
=& \begin{tikzpicture}[Sbase,scale=1.4, ]
    \DSOprime; \node[OP] at (2.75,0) {$ \gammabar O_\Dt\gammabar$}; \node[OP] at (2,.75) {$ (-1)^{|O_\Dt|}X_{N}$};
\end{tikzpicture} \nonumber\\ 
=& (-1)^{|\M{O}_\Dt|}X_N \bdot \Dt''.
\end{align}
The factor of $(-1)^{|\M{O}_\Dt|}$ is a consequence of moving $i\gamma_{e^N}$ past $\M{O}_\Dt$. In the last line of \eqref{S boundary}, $\Dt''$ is the same as $\Dt$ except with $\M{O}_\Dt$ replaced by $\gammabar \M{O}_\Dt \gammabar$. To obtain a relation as in \eqref{dualitytno1}, we require that $\Dt''\propto \Dt$. Assuming $\gammabar \M{O}_\Dt \gammabar =a\M{O}_\Dt$ for $a \in \mathbb{C}$, we obtain:
\begin{align}
    \Dt=\Dt\bdot S^2_N = a^2 \Dt,
\end{align}
and thus, $a=\pm 1$.

We are now able to discuss the affect of $\M{O}_\Dt$ on the mapping of states. We define $\Dt_{\alpha \beta}$ to be the TNO formed by closing $\Dt$ with $\M{O}_\Dt$ parameterized by $\M{O}_\Dt=(-i\gammabar)^\alpha (P)^\beta$. Then, \eqref{P boundary} and \eqref{S boundary} are summarized by:
\begin{align} \label{ODPS}
    \Dt_{\alpha\beta} \bdot P_{{e^0}}= &(-1)^{\alpha }Z_{N}Z_{0}\bdot \Dt_{\alpha\beta} \\
    \Dt_{\alpha\beta} \bdot S_{N}=& (-1)^{\alpha+\beta }X_{N}\bdot \Dt_{\alpha\beta}.
\end{align}
Acting on global fermion parity with $\Dt_{\alpha\beta}$, we find:
\begin{align} \label{DPar}
    \Dt_{\alpha\beta}\bdot  \prod_{e} P_{e} =&   (-1)^{\alpha}\left( \prod_{e} Z_{e_0}Z_{e_1} \right) \bdot \Dt_{\alpha\beta}  \nonumber\\
    =& (-1)^{\alpha}\Dt_{\alpha\beta}.
\end{align} 
This implies that $\Dt_{\alpha\beta}$ maps fermionic states $|\psi)$ with $|\psi| \neq \alpha$ to zero. Explicitly, we have:
\begin{align}
    \Dt_{\alpha\beta}  |\psi ) = & (-1)^{|\psi|}\Dt_{\alpha\beta} \bdot \prod_{e} P_{e} |\psi)  \nonumber \\
    = & (-1)^{|\psi|+\alpha}\Dt_{\alpha\beta} |\psi  ) .
\end{align}
Therefore, $\Dt_{\alpha\beta}  |\psi ) =0$ whenever $|\psi| \neq \alpha$. To bosonize an even state, $\alpha$ should be equal to $0$, and accordingly, $\M{O}_\Dt$ is proportional to $I$ or $P$.  For an odd state, one should use $\alpha=1$, in which case, $\M{O}_\Dt$ is proportional to $-i\gammabar$ or $\gamma$.

To understand the role of the $\beta$ parameter, we act on $\Dt_{\alpha\beta}$ with $\prod_v X_v$:
\begin{align} \label{prodXonD}
    \prod_v X_v \bdot \Dt_{\alpha \beta}=(-1)^{\alpha + \beta}\Dt_{\alpha \beta} \bdot\prod_v S_v.
\end{align}
Now we use a global relation of the fermionic operator algebra. It can be checked that:
\begin{align}
    \prod_v S_v = - \prod_e P_e.
\end{align}
Hence, continuing the calculation from \eqref{prodXonD}:
\begin{align}
    \prod_v X_v \bdot \Dt_{\alpha \beta} &= (-1)^{\alpha + \beta +1} \Dt_{\alpha \beta}\bdot \prod_e P_e \nonumber \\ 
    &= (-1)^{\beta+1} \Dt_{\alpha \beta}.
\end{align}
This means that the duality maps a fermionic state $|\psi)$ to the $(-1)^{\beta+1}$ eigenspace of $\prod_v X_v$, as can be seen from the following:
\begin{align}
    \prod_v X_v \bdot \Dt_{\alpha \beta}  |\psi) = (-1)^{\beta+1} \Dt_{\alpha \beta}  |\psi).
\end{align}

We have thus shown that $\M{O}_\Dt$ can be parameterized by $\M{O}_\Dt=(-i\gammabar)^\alpha P^\beta$ with $\alpha, \beta \in \{0,1\}$, and that $\Dt$ closed with $\M{O}_\Dt$ gives a map from the $(-1)^\alpha$ eigenspace of $\prod_e P_e$ to the $(-1)^{\beta+1}$ eigenspace of $\prod_v X_v$.

\subsection{Converting virtual indices to bosonic indices }

 The second step  is to contract $\Dt$ with $|\psi)$ to form $ |\psi_\text{bos}\rangle $:
\begin{align} |\psi_\text{bos}\rangle=
 \begin{tikzpicture}[Sbase,scale=1.3]
 \def \h{0.7};
     \draw[red,OES={AR2=0.5}] (0,0) --++(2,0); \draw[red,loosely dotted] (2,0) --(3,0); \draw[red,-<-=0.5] (3,0) --(5.2,0); 
     \draw[red,OES={AR2=0.5}] (0,\h) --++(0.5,0)--++(0.5,0)--++(0.5,0)--++(0.5,0); \draw[red,loosely dotted] (2,\h) --(3,\h); \draw[red,OES={AR2=0.5}] (3,\h) --(3.5,\h)--(4,\h)--(4.5,\h)--(5.2,\h); 
     \draw[dashdotted](0,-0.25)--+(0,1.25); \draw[dashdotted](5.2,-0.25)--+(0,1.25);
     \foreach \j in {0.5,1.5,3.5}
     {\draw[red,->-=0.5] (\j,0)--+(0,\h); \draw[black,->-=0.5] (\j+0.5,\h)--+(0,.75);
     \node[Tsq] at (\j,0) { }; 
     \node[Dtriangle] at (\j,\h) { };
     \node[Bcir] at (\j+0.5,\h) { }; 
     } \node[OP] at (4.7,0) {$O_\psi$};\node[OP] at (4.7,\h) {$O_\Dt$};
 \end{tikzpicture}.
\end{align}
 We can then see that $|\psi_\text{bos}\rangle$ is built from the local tensors $\M{M_f}\equiv \M{T} \bdot \M{F}\bdot \M{B}$:
\begin{align} \label{TFB2} 
&\begin{tikzpicture}[Sbase,scale=1]
    \draw[red,-<-=0.5](0,0)node[left,black]{$e_1'$}--(3,0)node[right,black]{$e'_0$}; \draw[red,-<-=0.5](0,1)node[left,black]{$e_1$}--(3,1)node[right,black]{$e_0$};\draw[red](1,0)--+(0,1);\draw[->-=0.5](2,1)--+(0,1)node[left,black]{$e_0$}; \node[Tsq] at (1,0) { }; \node[Dtriangle] at (1,1) { };  \node[] at (1,1.4) {$ $};\node[Bcir] at (2,1) { }; 
\end{tikzpicture}, 
\end{align}
and the tensor networkis closed with the operator $\M{O_f} \equiv O_\psi O_\Dt$:
\begin{align}
& \begin{tikzpicture}[Sbase,scale=1.3]
 \def \h{0.75};
  \draw[red,-<-=0.15] (5,\h)node[left,black]{$N$} --(6,\h)node[right,black]{$N$}; 
 \draw[red,-<-=0.15] (5,0)node[left,black]{$N'$} --(6,0)node[right,black]{$N'$}; 
  \filldraw[fill=white, draw=black] (5.5,.375) ellipse (.3cm and .55cm);
        \node[] at (5.5,.375) {$\M{O_f}$};
 \end{tikzpicture} \equiv
 \begin{tikzpicture}[Sbase,scale=1.3]
 \def \h{0.75};
  \draw[red] (5,\h) --(6,\h); 
 \draw[red] (5,0) --(6,0); 
        \node[OP] at (5.5,0) {$O_\psi$};\node[OP] at (5.5,\h) {$O_\Dt$};
 \end{tikzpicture}.
\end{align}
The state $ |\psi_\text{bos}\rangle $ formed by contracting together $\M{M_f}$ and closing with $\M{O_f}$ is indeed a bosonic state. However, it is not manifestly a bosonic matrix product state (bMPS), since the virtual legs may have nontrivial grading.

In the third and final step of the bosonization procedure, we write $ |\psi_\text{bos}\rangle $ as a bonafide bMPS -- constructed from a local tensor with bosonic virtual legs.  As suggested in section \ref{subsub:koszul signs}, we do so by choosing a particular internal ordering of the virtual legs of $\M{M_f}$ and $\M{O_f}$, in which they become convertible to bosonic indices. 
\begin{figure}
    \centering
     \begin{tikzpicture}[Sbase,scale=1,every node/.style={scale=1}]
      \def \h{0.4};
\draw[red] (0,0)--++(2,0)(2.5,0)--(4.5,0);     \draw[red,loosely dotted] (2,0)--(2.5,0);  
\draw[red] (0,\h)--(2,\h)(2.5,\h)--(4.5,\h);     \draw[red,loosely dotted] (2,\h)--(2.5,\h); 
\foreach \j in {0.5,1.5,3}
{\draw (\j,\h)--+(0,0.5);\node[Mrec,scale=1.2] at (\j,.2) {$M_f$};}
 \node[OP,yscale=1.3] at (4,0.2) {$\M{O_f}$};
 \draw[dashdotted](0,-0.25)--+(0,1); \draw[dashdotted](4.5,-0.25)--+(0,1);
 \end{tikzpicture} 
 
  \begin{tikzpicture}[Sbase,scale=1]
 \node[scale=2] {$=$};
 \end{tikzpicture}   
      \begin{tikzpicture}[Sbase,scale=1,every node/.style={scale=1}]
      \def \h{0.4};
\draw[] (0,0)--++(2,0)(2.5,0)--(5,0);     \draw[loosely dotted] (2,0)--(2.5,0);  
\draw[] (0,\h)--(2,\h)(2.5,\h)--(5,\h);     \draw[loosely dotted] (2,\h)--(2.5,\h); 
\foreach \j in {0.5,1.5,3}
{\draw (\j,\h)--+(0,0.5);\node[Mrec,scale=1.2] at (\j,.2) {$M_b$};}
 \node[OP,yscale=1.3] at (3.75,0.21) {$\M{O_b}$};
 \draw[dashdotted](0,-0.25)--+(0,1); \draw[dashdotted](5,-0.25)--+(0,1);
 \node[OP] at (4.5,0) {$Z$};\node[OP] at (4.5,0.4) {$Z$};
 \end{tikzpicture} 
    \caption{With the internal ordering chosen in \eqref{intord1D1} and \eqref{intord1D2}, the virtual legs of $\M{M_f}$ and $\M{O_f}$ can be replaced with un-graded virtual legs. The supertrace sign produced between the first and last indices on both layers is accounted for by inserting the operator $Z_{N}\otimes Z_{N'}$ before closing the state generated by $\M{M_b}$ and $\M{O_b}$.}
    \label{fig:step3}
\end{figure}
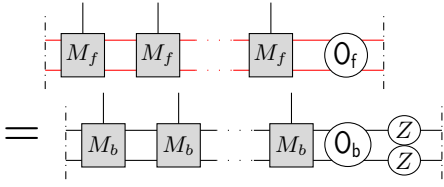
We start by writing $\M{M_f}$ and $\M{O_f}$ in tensor component form. In tensor component form, a generic $\M{M_f}$ is:
\begin{align}\label{Mf}
    \M{M_f} = \sum_{\substack{j,a',a,\\b,b'}=0} (M_f)^j_{aa',b'b} |a')_{e'_1}|a)_{e_1} |j\rangle_{e_0} (b|_{e_0}(b'|_{e'_0}, 
\end{align}
where the components of $\M{M_f}$ can of course be expressed in terms of the components of $\Ft$, $\Bt$, and $\M{T}$. Note that we have chosen a specific ordering of the vectors in $\M{M_f}$.  Schematically, the vectors are ordered as:
\begin{align} \label{intord1D1} 
\begin{tikzpicture}[Sbase,scale=1]
      \draw[red](0,0)node[left,black]{$\M{(i)}$}--(3,0)node[right,black]{$\M{(v)}$}; \draw[red](0,1)node[left,black]{$\M{(ii)}$}--(3,1)node[right,black]{$\M{(iv)}$};\draw[red](1,0)--+(0,1);\draw(2,1)--+(0,1)node[left,black]{$\M{(iii)}$}; \node[Tsq] at (1,0) { }; \node[Dtriangle] at (1,1) { };\node[Bcir] at (2,1) { }; 
\end{tikzpicture}.
\end{align}
Next, we write a generic $\M{O_f}$ in tensor component form:
\begin{align}
  \M{O_f} = \sum_{a',b',a,b}(O_f)_{a',b',a,b}|a')_{N'}|a)_N(b|_N (b'|_{N'},
\end{align}
where we have intentionally ordered the graded vectors according to the diagram:
\begin{align} \label{intord1D2}
    \begin{tikzpicture}[Sbase,scale=1.3]
 \def \h{0.75};
  \draw[red] (5,\h)node[left,black]{(ii)} --(6,\h)node[right,black]{(iii)}; 
 \draw[red] (5,0)node[left,black]{(i)} --(6,0)node[right,black]{(iv)}; 
  \filldraw[fill=white, draw=black] (5.5,.375) ellipse (.3cm and .55cm);
        \node[] at (5.5,.375) {$\M{O_f}$};
 \end{tikzpicture}.
\end{align}

It may be checked that with the special choices of ordering in \eqref{intord1D1}, we do not produce any Koszul signs while contracting the $\M{M_f}$ with each other. Therefore, as suggested in section \eqref{subsub:koszul signs}, we can simply replace all fermionic virtual legs shared by two $\M{M_f}$ tensors with bosonic legs.
Similarly, with the choice of ordering in \eqref{intord1D2}, no sign is produced in the contraction of $\M{M_f}$ with $\M{O_f}$, so their common indices can also be replaced with bosonic indices. However, a Koszul sign \textit{is} produced in the trace operation (contraction of first and last indices) due to the supertrace phase [see \eqref{supetrace}]. However, these indices are convertible to bosonic indices with $(Z_{N} \otimes Z_{N'})$-insertion (see section \ref{subsub:koszul signs}). That is, we can replace them with bosonic indices as long as we insert an operator $Z_{N} \otimes Z_{N'}$ (one $Z$ on each of the two virtual indices) before closing the MPS. 

We denote the bosonic tensor obtained by replacing the fermionic virtual legs of $\M{M_f}$ as $\M{M_b}$, and similarly, we denote the bosonic tensor obtained by replacing the fermionic virtual legs of $\M{O_f}$ as $\M{O_b}$. Then, the state generated by $\M{M_f}$ and $\M{O_f}$ and the state generated by $\M{M_b}$ with $\M{O_b}$ and $Z_{N} \otimes Z_{N'}$ is the same state (see Fig.~\ref{fig:step3}).  It is convenient to further absorb the $Z$ factors into the definition of $\M{O_b}$. With this, the bMPS is generated by the tensors:
\begin{align}
    \M{M_b}=&  \sum_{\substack{j,a',a\\,b,b'}=0} (M_f)^j_{aa',b'b} |a'\rangle_{e'_0}|a\rangle_{e_0} |j\rangle_{e_1} \langle b|_{e_1}\langle b'|_{e_1'} \\
      \M{O_b}=&  \sum_{a',b',a,b}(O_f)_{a',b',a,b}(-1)^{b+b'}|a'\rangle_{N'}|a\rangle_N\langle b|_N \langle b'|_{N'},
\end{align}
where the phase $(-1)^{b+b'}$ comes from the application of $Z_{N} \otimes Z_{N'}$. Now, contracting $\M{M_b}$ and closing the tensor networkwith the bosonic tensor $\M{O_b}$ yields $|\psi_\text{bos}\rangle$, and in this way, $ |\psi_\text{bos}\rangle $ is expressly a bMPS. Thus, we have successfully mapped the fMPS $|\psi)$ to the bMPS $ |\psi_\text{bos}\rangle $.

In summary, bosonization of a fMPS defined by a tensor $\M{T}$ and operator $O_\psi$ as in Eq. \eqref{psidef} proceeds in three steps.
\begin{enumerate}
\item Choose an operator $O_\Dt=(-i\gammabar)^{\alpha}P^{\beta}$ with $\alpha,\beta \in \{0,1\}$ with which to close the bosonization TNO. 
\item Construct $\M{M_f}$ by contracting $\M{T}$, $\Ft$, and $\Bt$.  Form $\M{O_f}$ by combining $O_\psi$ and $O_\Dt$.
\item Rearrange the vectors in $\M{M_f}$ and $\M{O_f} $ to match the ordering in \eqref{intord1D1} and \eqref{intord1D2}, respectively.  Form $\M{M_b}$ and $\M{O_b}$ from $\M{M_f}$ and $\M{O_f}$ by taking the graded vectors to have trivial grading and modifying the components $(O_f)_{a',b',a,b}$ of $\M{O_f}$ by $(-1)^{b+b'}$ to account for the supertrace.  
\end{enumerate}
In the next subsection, we provide explicit examples of the tensor network bosonization steps above.


\subsection{Examples}

We will illustrate the tensor network bosonization procedure of the previous section on two examples -- a trivial atomic insulating state and the nontrivial ground state of the Kitaev chain.  To motivate the TNO duality, we also analyze the examples at the operator level using the duality of section \ref{reviewbosonization}. 

\textbf{Example 1:} The trivial atomic insulating state is the ground state of the Hamiltonian $H_\text{triv}=-\sum_{e }P_{e}$. It has zero fermion occupancy at each site and can be expressed in the form:
\begin{align} \label{psidef2}
  &|\psi_\text{triv})=\sum_{j_0,\ldots,j_N } \text{tr}\left[ T^{j_0}\ldots T^{j_N} O_\psi \right] |j_1)_{e^0}\ldots |j_N)_{e^N}.
\end{align}
with $\M{T}$ being the trivial tensor: 
\begin{align}\label{trivialT2}
\begin{tikzpicture}[Sbase]
\draw[red,OES={AR2=0.5}](0,0)node[above,black]{$e_1$}--++(1,0)--++(1,0)node[above,black]{$e_0$};
\draw[red,->-=0.5](1,0)node[Tsq]{$\M{T}$}--+(0,1)node[right,black]{$e$};
\end{tikzpicture}
    = |0)_{e_1} |0)_{e} (0|_{e_0},
\end{align}
and $O_\psi$ equal to the parity operator $P$.

Using the 1D operator duality in section \ref{reviewbosonization} [Eq.~\eqref{D1def}], we see that $H_\text{triv}$ is mapped to the spontaneous symmetry breaking Hamiltonian $H_{SSB}=-\sum_{e}Z_{e^1} Z_{e^0}$. In accordance, we will see that the bosonization TNO maps the ground state of $H_\text{triv}$ to a ground state of $H_{SSB}$.

The first step of the tensor network bosonization procedure is to choose an operator ${O}_\Dt$ with which to close the bosonization TNO.  For simplicity, let us choose $O_\Dt$ to be fermion parity $P$.  This choice of ${O}_\Dt$ gives a map from the set of fermion parity even states to the set of states symmetric under $\prod_v X_v$ (Appendix \ref{app: BC in 1D}).

Next, we construct the tensor $ \M{M_f}$ and the operator $\M{O_f}$.  $\M{M_f}$ is obtained by contracting $\M{T}$, $\Ft$, and $\Bt$ as in \eqref{TFB2}:
\begin{align}
    &\M{M_f}=\M{T} \bdot \M{F}\bdot \M{B}\nonumber\\
    &= |0)_{e'_1}|0)_{e'}^{\ct_1} (0|_{e'_0}  \sum_{a,b}|a)_{e_1}(a+b|^{\ct_1}_e (b|^{\ct_2}_{e_0}  \sum_c|c)^{\ct_2}_{e_0} |c\rangle_{e_0} (c|_{e_0}  \nonumber\\
    &= \sum_a |0)_{e_1'}(0|_{e_0'}|a)_{e_1} |a \rangle_{e_0} (a|_{e_0}\nonumber\\
    &=\sum_a |0)_{e_1'}|a)_{e_1} |a \rangle_{e_0} (a|_{e_0}(0|_{e_0'},
\end{align}
and $\M{O_f}$ is simply
\begin{align}
    \M{O_f}=&\Bigg( \sum_{b'} (-1)^{b'} |b')_{N'}(b'|_{0'} \Bigg)  \Bigg ( \sum_{b} (-1)^b |b)_{N}(b|_{0} \Bigg).
\end{align}

Then, we rearrange the order of the graded vectors in $\M{M_f}$ and $\M{O_f}$ according to \eqref{intord1D1} and \eqref{intord1D2}.  In the final step of the tensor network bosonization procedure, we construct $\M{M_b}$ and $\M{O_b}$ by removing the grading and appropriately accounting for the supertrace.  Following these steps, $\mathsf{M_b}$ is:
\begin{align} \label{Mbtriv}
    \M{M_b}= \sum_a |0\rangle_{e_0'}|a\rangle_{e_0} |a \rangle_{e_1} \langle a|_{e_1} \langle 0|_{e'_1}, 
\end{align}
and re-ordering the vectors of $\M{O_f}$ and accounting for the supertrace gives:
\begin{align}
  \M{O_b}=\sum_{b',b}  |b'\rangle_{e_0'}|b\rangle_{e_0} \langle b|_{e_1}\langle b'|_{e_1'}.
\end{align}

The bosonized state is constructed by gluing together $\M{M_b}$ and closing the tensor network with $\M{O_b}$. To see that $\M{M_b}$ generates the ground state of $H_\text{SSB}$, we first notice that, for $\M{M_b}$ in Eq. \eqref{Mbtriv}, the $e'_0$ and $e'_1$ indices do not affect the bosonized state. Therefore, $\M{M_b}$ and $\M{O_b}$ can be reduced to:
\begin{align}
  \tilde{\Mt}_\M{b}&= \sum_{a}  |a\rangle_{e_0}  |a \rangle_{e_1}\langle a|_{e_1} \\
    \tilde{\M{O}}_\M{b}&=\sum_{b} |b \rangle_{e_0} \langle b|_{e_1}.
\end{align}
$\tilde{\Mt}_\M{b}$ generates the state: 
    \begin{align}
    |\psi_\text{bos} \rangle = |00\ldots \rangle + |11\ldots \rangle,
\end{align}
which is the ground state of $H_\text{SSB}$.  Therefore, $\M{M_b}$ also generates the ground state of $H_\text{SSB}$.

\textbf{Example 2:} Now, we turn to the example of the non-trivial ground state $|\psi_\text{K})$ of the Kitaev chain.  $|\psi_\text{K})$ is the ground state of the Hamiltonian $H_\text{K}=-\sum_{v} S_{v}$, and it can be written as a fMPS with [\onlinecite{Bultinck17}]:
\begin{align}
    \M{T}= \sum_{a,b} (-1)^{a(a+b)}|a)_{e_1'} |a+b)_{e} (b|_{e_0'},
\end{align}
and $ {O}_\psi= -i\gammabar$. 

The operator duality $\mathfrak{D}$ maps $H_\text{K}$ to the paramagnet Hamiltonian $H_\text{para}=-\sum_{v} X_{v}$, so the bosonization TNO should transform $|\psi_\text{K})$ to the paramagnet ground state $|\psi_\text{para}) = |++\ldots \rangle $, where $|+\rangle =\frac{1}{\sqrt{2}}\left( |0\rangle + |1 \rangle\right) $. 

Following the three steps outlined in the previous section, we first choose ${O}_\Dt$. Since ${O}_\psi$ is fermion parity odd, we choose ${O}_\Dt=(-i\gammabar)P=\gamma$. 
  First, we compute $\M{M_f}$ by contracting $\M{T}$, $\Ft$, and $\Bt$:
\begin{widetext}
\begin{align}
    \M{M_f}=\M{T}\bdot\M{F}\bdot\M{B}=& \Bigg[ \sum_{a',b'}(-1)^{|b'|(|a'|+|b'|)}|a')_{e_0'}|a'+b')^{\mathcal{C}_1}_{e'} (b'|_{e_1'} \Bigg] \Bigg[ \sum_{a,b}|a)_{e_0}(a+b|^{\mathcal{C}_1}_e (b|^{\mathcal{C}_2}_{e_1} \Bigg]\Bigg[ \sum_c|c)^{\mathcal{C}_2}_{e_1} |c\rangle_{e_1} (c|_{e_1} \Bigg] \nonumber \\
    =& \sum_{a,b,a',b'}(-1)^{(b'+b)(a'+b')}\delta_{a+b,a'+b'}|a')_{e_0'}|a)_{e_0} |b\rangle_{e_1} (b|_{e_1} (b'|_{e_1'}
\end{align}
\end{widetext}
Next, we remove the grading of the vectors in $\M{M_f}$ and $\M{O_f}=-i\gammabar$ and account for the supertrace to form $\M{M_b}$ and $\M{O_b}$:
\begin{align}
    \mathsf{M_b}=&\sum_{\substack{a,b\\a',b'}}(-1)^{(a'+b)(a'+b')}\delta_{a+b,a'+b'}\nonumber\\
    &|a'\rangle_{e_0'}|a\rangle_{e_0} |b\rangle_{e_1} \langle b|_{e_1} \langle b'|_{e_1'} \nonumber\\
     \mathsf{O_b} =& \sum_{b,b'} (-1)^{b+b'}|b'\rangle_{N'}|b\rangle_{N}\langle b+1|_{0}\langle b'+1|_{0'}. 
\end{align}
Explicitly, we have:
\begin{align}
    \mathsf{M_b}    = &\big(|00 \rangle-|11 \rangle \big) \big(|0\rangle \langle 00|-|1\rangle \langle 11| \big)+ \\ \nonumber &\big(|10 \rangle+|01 \rangle\big) \big(|0\rangle \langle 10|+|1\rangle \langle 01|\big)  \\
    \mathsf{O_b} = &|00\rangle \langle 11 | +|11\rangle \langle 00|  \nonumber \\  -&|01\rangle \langle 10|- |10\rangle \langle 01 |. 
\end{align}

Defining $|v_0\rangle = |00 \rangle-|11 \rangle $ and $|v_1\rangle=|10 \rangle+|01 \rangle $ and the corresponding  projectors: $P_j=|v_j\rangle\langle v_j|$, $j=0,1$, then $\mathsf{M_b}$ satisfies $\mathsf{M_b}P_j = P_j\mathsf{M_b}P_j$. The boundary operator also satisfies $P_j \mathsf{O_b}=P_j \mathsf{O_b} P_j$. Thus, there are two canonical blocks: 
\begin{align}
    P_j \mathsf{M_b} P_j &= |v_j\rangle (|+\rangle) \langle v_j| \quad j=0,1\nonumber \\
    \langle v_j |\mathsf{O_b} |v_j\rangle  &= -1,
\end{align}
where $|+\rangle = |0\rangle +|1\rangle$. Both blocks give the same state: $- |+\rangle^{\otimes N}$. 


\bibliographystyle{unsrt}
\bibliography{TNORef.bib}
\end{document}